\documentclass[a4paper,twocolumn,10pt,accepted=2022-02-07]{quantumarticle}
\pdfoutput=1
\usepackage[utf8]{inputenc}
\usepackage{soul}
\usepackage[english]{babel}
\usepackage[T1]{fontenc}
\usepackage{amsmath}
\usepackage{hyperref}

\usepackage{tikz}
\usepackage{lipsum}
\usepackage{tikz} 
\usepackage{amsmath,amsfonts,amsthm, mathtools}
\usepackage{float}
\usepackage{dsfont}
\usepackage{csquotes}
\usepackage{amssymb}
\usepackage{verbatim}
\usepackage{bm}
\usepackage[numbers,sort&compress]{natbib}
\usepackage{hyperref}
\definecolor{darkred}  {rgb}{0.5,0,0}
\definecolor{darkblue} {rgb}{0,0,0.5}
\definecolor{darkgreen}{rgb}{0,0.5,0}
\hypersetup{
  colorlinks = true,
  urlcolor  = blue,         
  linkcolor = darkblue,     
  citecolor = darkgreen,    
  filecolor = darkred       
}
\theoremstyle{definition}

\newtheorem{corollary}{Corollary}
\newtheorem{definition}{Definition}
\newtheorem{conjecture}{Conjecture}
\newtheorem{lemma}{Lemma}
\newtheorem{proposition}{Proposition}
\newtheorem{theorem}{Theorem}

\newtheorem*{remark}{Remark}
\newcommand{\conv}{\text{conv}}
\newcommand{\mbf}{\mathbf}
\newcommand{\mbb}{\mathbb}
\newcommand{\mc}{\mathcal}
\newcommand{\msf}{\mathsf}
\newcommand{\mf}{\mathfrak}

\newcommand{\tr}{\textrm{Tr}}

\newcommand{\cl}{\text{cl}}
\newcommand{\crk}{\text{crk}}
\usepackage{multirow}
\newcommand{\ket}[1]{|#1\rangle}
\newcommand{\bra}[1]{\langle #1|}
\newcommand{\op}[2]{|#1\rangle\langle #2|}
\newcommand{\ip}[2]{\langle #1| #2\rangle}

\newcommand{\bcA}{\bm{\mc{A}}}
\newcommand{\bsA}{\bm{\msf{A}}}
\newcommand{\bA}{\bm{A}}
\newcommand{\bs}{\boldsymbol}

\definecolor{cool_green}{rgb}{0.0, 0.5, 0.0}

\newcommand{\norm}[1]{\left\lVert#1\right\rVert}
\usepackage{blkarray}
\usepackage{graphicx}
\usepackage{dcolumn}
\usepackage{bm}

\begin{document}

\title{Building Multiple Access Channels with a Single Particle}

\author{Yujie Zhang}
 \affiliation{Department of Physics, University of Illinois at Urbana-Champaign, Urbana, IL 61801, USA}
\orcid{0000-0002-7858-7476}
\author{Xinan Chen}%
 \affiliation{Department of Electrical and Computer Engineering, University of Illinois at Urbana-Champaign, Urbana, IL 61801, USA}
\author{Eric Chitambar}
\email{echitamb@illinois.edu}
\orcid{0000-0001-6990-7821}
\affiliation{Department of Electrical and Computer Engineering, University of Illinois at Urbana-Champaign, Urbana, IL 61801, USA}

\maketitle

\begin{abstract}
A multiple access channel describes a situation in which multiple senders are trying to forward messages to a single receiver using some physical medium. In this paper we consider scenarios in which this medium consists of just a single classical or quantum particle. In the quantum case, the particle can be prepared in a superposition state thereby allowing for a richer family of encoding strategies. To make the comparison between quantum and classical channels precise, we introduce an operational framework in which all possible encoding strategies consume no more than a single particle.  We apply this framework to an $N$-port interferometer experiment in which each party controls a path the particle can traverse.  When used for the purpose of communication, this setup embodies a multiple access channel (MAC) built with a single particle. 

We provide a full characterization of the $N$-party classical MACs that can be built from a single particle, and we show that every quantum particle can generate a MAC outside the classical set.  To further distinguish the capabilities of a single classical and quantum particle, we relax the locality constraint and allow for joint encodings by subsets of $1<K\le N$ parties.  This generates a richer family of classical MACs whose polytope dimension we compute.  We identify a "generalized fingerprinting inequality" as a valid facet for this polytope, and we verify that a quantum particle distributed among $N$ separated parties can violate this inequality even when $K=N-1$.  Connections are drawn between the single-particle framework and multi-level coherence theory.  We show that every pure state with $K$-level coherence can be detected in a semi-device independent manner, with the only assumption being conservation of particle number.
\end{abstract}

\section{Introduction}

A quantum particle is fundamentally different than its classical counterpart.  From the second quantization picture, a quantum ``particle'' can be regarded as a single mode excitation of some field, and hence it is endowed with certain wave-like features that a classical particle lacks.  Modern experimental techniques are able to individually address these quantum particles, or ``quanta'', and use them for information processing.  Hence it is a practically relevant endeavor to identify the advantages that a single quantum particle can offer in some particular information task over a classical particle.  This provides an operational comparison between classical and quantum information systems in terms of {one of} their basic building block -- particle.

In this work, we consider the task of multi-party communication with $N$ spatially-separated senders $(\msf{A}_1, \msf{A}_2,\cdots \msf{A}_N)$ and one receiver $(\msf{B})$. The receiver $\msf{B}$ obtains some output data $b$ that depends on the collection of messages $(a_1,a_2,\cdots,a_N)$ chosen by the senders.  Ideally, $b$ would be a perfect copy of all the $N$ messages, $b=(a_1,a_2,\cdots,a_N)$.  However in practice there are some physical limitations that prevent perfect communication.  In such scenarios, the communication is described by the transition probabilities $p(b|a_1\cdots  a_N)$.  The distributions $p(b|a_1\cdots  a_N)$ collectively represent a multiple-access channel (MAC) \cite{Cover-2006a,Tse-2005a}.  Ultimately, the probabilities $p(b|a_1\cdots a_N)$ are determined by the particular physical system used to transmit the information.  The question we raise here is what MACs can be generated under the restriction that the communication channel be implemented using only a single particle, with none of its internal degrees of freedom being accessible.  More precisely, information is only allowed to be encoded in external relational degrees of freedom, such as what particular points in space-time the particle occupies.  We are interested in comparing the MACs that can be realized when a quantum versus classical particle is used to transmit information in this way.

This question can be pushed in a variety of different directions. In point-to-point communication, one very active line of research considers communication enhancements that can arise when a single quantum particle is subjected to different configurations of quantum communication devices \cite{Kristjansson-2020a}.  For instance, non-classical effects can be generated when a particle is subjected to two noisy channels with indefinite causal order \cite{Ebler-2018a,Goswami-2020,salek-2018,Procopio-2019,Procopio-2020, Chiribella-2021}.
  
It is also possible to explore the communication power of a single particle that travels through channels in a superposition of different trajectories or times \cite{Alastair-2018a,Gisin-2005, Chiribella-2019a, Kristjansson-2020a}. In these
strategies, the superposition in trajectories or causal order have to be determined by an extra control bit, and there are subtleties in interpreting precisely what types of enhancements are physically realizable \cite{Kristjansson-2020b}.
In this paper, 
we adopt a more basic model in which quantum superpositions are limited to just the spatial path traveled by the particle and accessed by the different senders. The encoding of information is delocalized among the parties $(\msf{A}_1, \msf{A}_2,\cdots \msf{A}_N)$, and the encoding is assumed to be implemented at a definite moment in time.  Yet even in this more familiar and simpler setting, operational advantages of using quantum versus classical MACs can be identified.

Consider, for example, the task of quantum fingerprinting.  In the two-party scheme, the goal of this task is for receiver $\msf{B}$ to decide whether or not $a_1=a_2$ based on some partial information sent from $\msf{A}_1$ and $\msf{A}_2$, which represents their ``fingerprints'' \cite{Buhrman-2001a}.  While there are many variations to the problem, one version compares the task when $\msf{A}_1$ and $\msf{A}_2$ have pre-shared classical versus quantum correlations \cite{Horn-2005a, Massar-2005a, Horn-2005b}.  As observed by Massar \cite{Massar-2005a}, a single quantum particle initially prepared in a spatial superposition state can be used for $\msf{B}$ to decide if $a_1=a_2$ when $a_1,a_2\in\{0,1\}$, a feat that is impossible with a single classical particle.  In this example, parties $\msf{A}_1$ and $\msf{A}_2$ encode the parity of their bit values in the phase of the single-particle state.

{
The results of this paper are inspired by this fingerprinting task but go well beyond it.  Mathematically, every classical single-particle fingerprinting protocol can be described as a point in a polytope, the so-called ``classical polytope''.  The best possible classical fingerprinting protocol, in terms of the probability that $\msf{B}$ correctly decides matching inputs, lies on the boundary of this polytope.   Using a single quantum particle allows for protocols that go beyond this boundary, thereby identifying a genuine quantum advantage.  As we move to more senders and more inputs/outputs, the classical polytope becomes much more complex, and the classical-quantum boundary is no longer characterized by the simple fingerprinting task.  One goal of this paper is describe new features that emerge in this higher-dimensional polytope that do not reduce to two-party fingerprinting.


Other sources of inspiration for this work are a series of recent results showing different types of single-particle enhancement in quantum communication.}  Del Santo and Daki\'{c} have devised a protocol that allows for two-way communication between two parties using a single quantum particle, whereas the communication is always one-way if a single classical particle is employed \cite{DelSanto-2018a}.  A cascaded implementation of the Del Santo and Daki\'{c} protocol has also been devised in which one of the parties can gain even more information through the use of a single particle \cite{Hsu-2020a}.   Advantages of quantum particles in the reverse communication setting of two senders and one receiver have also been discovered \cite{DelSanto-2019a}, and very recently, Horvat and Daki\'{c} have found that the superposition state of a single particle can provide significant enhancement to the speed of information retrieval in some globally-encoded data \cite{Horvat-2019a}.  More generally, methods are known for mapping multi-qubit quantum information into the spatial or temporal degrees of freedom of a single photon \cite{Garcia-Escartin-2013a, Arrazola-2014a}.  These results reveal the resource character of single-particle superposition states for building enhanced multipartite classical communication channels.  It is natural to consider how powerful this resource can be for generating classical channels and whether it can be formally characterized in the context of a quantum resource theory \cite{Horodecki-2013a, Coecke-2016a, Chitambar-2019a}. 

Traditional studies of quantum-enhanced MACs have focused on the capacity rate region \cite{Winter-2001a, Hsieh-2008a, Leditzky-2020a}, which quantifies the optimal asymptotic communication rates between the senders and the receiver.  From an information-theoretic sense, this is perhaps the most natural object to consider.  However, capacity is just one property of a channel, and two different channels can have the same rate region.  Consequently, a comparison of rate regions can be too broad for the purpose of separating classical and quantum-generated MACs.  Our analysis is fine-grained in that it distinguishes MACs on the level of individual transition probabilities $p(b|a_1\cdots  a_N)$.  This has the advantage that, when expressed in terms of these probabilities, the collection of MACs generated by a single classical particle forms a convex polytope, provided the output variable is binary.  This allows us to employ standard techniques from convex analysis to construct experimentally-implementable methods for certifying non-classical MACs.

The main results of this paper and its organization are as follows.  In Section \ref{Sect:operational framework}, we propose an operational framework for building multiple-access channels with a single particle based on an $N$-path interferometer experiment.  Sender $\msf{A}_i$ encodes information along path $i$ using a map that obeys the central constraint of not increasing the overall particle number (as characterized in Eq. \eqref{Eq:Kraus-operators-NP}).  Included in these allowed operations are phase encoding and path blocking \cite{Rozema-2020}, the latter is also known as vacuum encoding \cite{Chiribella-2019a, Kristjansson-2020a}.  We also introduce communication models in which shared randomness is introduced and other relaxations on the allowed operations are considered.  {The interested reader can jump to Section \ref{Sect:SR-model} and compare the figures within for a high-level overview of the different models we consider.}

In Section \ref{Sect:classical MACs} we analyze the structure of classical local MACs in more detail. In Theorem \ref{Thm:classical-I2-N}, we show that these MACs are characterized entirely in terms of vanishing second-order interference terms.  Roughly speaking, this means that the particular linear combination
\begin{equation*}
\label{Eq:I2-intro}
    I_2=p(b|a_i,a_j)+p(b|a_i',a_j')-p(b|a_i',a_j)-p(b|a_i,a_j')
\end{equation*}
must equal zero for any output $b$ and arbitrary inputs $a_i,a_i',a_j,a_j'$ chosen by parties $\msf{A}_i$ and $\msf{A}_j$.  The quantity $I_2$ is well-known in the study of double-slit experiments \cite{Grangier-1986,Jacques-2005,Merzbacher-1998a}, and here we prove that its vanishing is essentially the only constraint that assures the MAC has an implementation using a single classical particle. A similar observation for binary input/output MACs was made independently in Ref. \cite{Horvat-2019b}. In contrast, every quantum state with a non-zero off-diagonal term in the path basis can generate $I_2\not=0$, as explicitly shown in Section \ref{Sect:quantum MACs} Proposition \ref{Prop:N-Local-quantum-I2}.

We also consider the family of MACs obtained by partially relaxing the locality constraints on the encoders and allowing joint encoding schemes among subgroups of size $K<N$.  In low dimensions, the resulting classical polytope is presented in Section \ref{Sect:N,K binary}.  For arbitrary $N$ and $K$, we show that the classical polytope exhibits a tight facet inequality, which we call the generalized fingerprinting inequality \cite{Horvat-2019a}:
\[p(0|0,\cdots,0)+\sum_{i=1 }^{K+1}p(1|0,\cdots,1_i,\cdots,0)\le K+1,\]
This inequality is shown to be violated by a fully local quantum MAC in Section \ref{Sect:N-local-quantum}, a result previously demonstrated in Ref. \cite{Horvat-2019a} but one we slightly optimize here.   Along the way we also derive a number of other structural properties of single-particle classical and quantum MACs.  Finally, in Section \ref{Sect:Multi-level-coherence}, we draw a direct connection between our framework and the resource theory of multi-level coherence.  In particular, we design new semi-device independent witnesses for multi-level coherent states.


Near the completion of this manuscript, we became aware of another paper \cite{Horvat-2020a} by Horvat and Daki\'{c} that is similar in spirit to this work.  The setup in their paper also involves local encoding in an $N$-path interformeter, however the restriction to a single particle is not made and the analysis focuses on the emergence of higher-order interference effects.  The latter refers to realizing nonzero values for the quantities
\begin{align*}
    \label{Eq:Ik-intro}
    I_{K}=\sum_{a_1,\cdots,a_{K}\in\{0,1\}}\prod_{i=1}^{K}(-1)^{a_i}p(0|a_1,\cdots a_{K}),
\end{align*}
with $K=2,\cdots,N$. Notice that $K=2$ corresponds to the second-order interference given above. The quantity $I_K$ arises naturally in $K$-slit interference experiments, and while quantum mechanics can generate $I_2\not=0$, it is known that $I_K=0$ for $K>2$ in all standard single-particle quantum mechanical setups, an intriguing fact that has motivated multiple studies into the nature of higher-order interference \cite{Sorkin-1994a, Sinha-2010a, Ududec-2010a, Lee-2016a, Daki-2014, Horvat-2020a, Rozema-2020}.  In Ref. \cite{Horvat-2020a}, it is shown that $I_{2K}$ can nevertheless be nonzero in interferometer experiments involving $K$ particles, while $I_{K'}$ still vanishes in $K$-particle experiments when $K'>2K$.  In Section \ref{Sect:quantum-B>2} of this paper, we describe a similar effect in single-particle interferometer experiments.  Namely, when $K$ parties are allowed to jointly encode on a quantum system of just one particle using operations constrained to particle-number conservation, $I_{2K}$ can be nonzero, but $I_{K'}$ vanishes for all $K'>2K$.  We thus see Ref. \cite{Horvat-2020a} and this work as being complementary and reflecting once again the richness of $N$-path interforemeter experiments for demonstrating quantum information primitives.

\section{Definitions and Operational Framework}
\label{Sect:operational framework}
Let us now introduce our framework in more detail. In Subsection \ref{Sect:single-particle-MACs}, we will provide a general recipe for building MACs with $M$ quantum particles. Moving beyond that section, we focus exclusively on the single-particle case with $M=1$. In Subsections \ref{Sect:Macs with Q-particle} and \ref{Sect:Macs with C-particle}, the structure of single-particle MACs is analyzed for both a quantum and classical particle. In Subsection \ref{Sect:SR-model}, we include shared randomness as an extra resource, which further modifies the communication setup. Finally, in Subsection \ref{Sect:N-K Macs}, we partially relax the locality constraint on the $N$ senders and introduce the idea of $(N,K)$-local MACs.  

\subsection{MACs with Quantum Particles and Number-Preserving extendible Operations}
\label{Sect:single-particle-MACs}

 We model our communication scenario as a generalized $N$-port interferometer experiment in which sender $\msf{A}_i$ sits along path $i$ and wishes to send a classical message $a_i$ drawn from set $\mc{A}_i$ (see Fig. \ref{Fig:N-Local}).  We denote the full collection of senders as $\bm{\msf{A}}=(\msf{A}_1,\msf{A}_2,\cdots,\msf{A}_N)$  Throughout this work we will assume that all message sets are finite, and they can thus be represented by a set of non-negative integers.  For an arbitrary real number $x$, we denote $[x]:=\{0,\cdots,\lfloor x\rfloor -1\}$ so that $\mc{A}_i$ is in a one-to-one correspondence with the set $[|\mc{A}_i|]$.  We will be particularly interested in the binary set $[2]=\{0,1\}$ and its $N$-fold Cartesian product $[2]^N:=\{0,1\}^{\times N}$.  For a general Cartesian product of input sets $\mc{A}_i$ we simply write $\bm{\mc{A}}:=\mc{A}_1\times\mc{A}_2\times\cdots\times \mc{A}_N$, and likewise for random variables $A_i$ we write their joint variable as $\bm{A}:=A_1A_2\cdots A_N$.  For an output set $\mc{B}$, the $N$-sender MACs we consider are {the collection of} stochastic maps $\mbf{p}_{B|\bm{A}}:\bm{\mc{A}}\to\mc{B}$ whose transitions probabitilies we denote by $p(b|\bm{a})$ for $b\in\mc{B}$ and $\bm{a}\in\bm{\mc{A}}$.  In this way $\mbf{p}_{B|\bm{A}}$ can be viewed as a ($|\mc{B}|\cdot{|\bcA|}$)-dimensional vector with coordinates being $p(b|\bm{a})$.

\begin{figure}[h]
  \centering
    \includegraphics[width=0.5\textwidth]{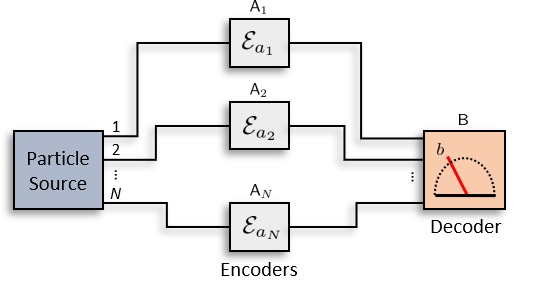}
     \caption{A physical implementation of the MAC $\mbf{p}_{B|\bm{A}}$ using a single particle and a generalized interferometer setup.}
     \label{Fig:N-Local}
\end{figure}

In a general $M$-particle communication experiment, an $M$-particle state $\rho^{\bm{\msf{A}}}$ is prepared and sent along the $N$ possible interferometer paths, with its final destination being the receiver $\msf{B}$.  Here $\rho^{\bm{\msf{A}}}$ has support on space 
\begin{equation}
    \mc{H}^{\bm{\msf{A}}}_{M}:=\text{span}\{\ket{\mbf{x}}:=\ket{x_1}_{\msf{A}_1}\cdots\ket{x_N}_{\msf{A}_N}\;:\sum_{i=1}^Nx_i=M\},
\end{equation}
where $\{\ket{0}_{\msf{A}_i},\ket{1}_{\msf{A}_i},\cdots\}$ are number states for the particle along path $i$.  We assume the particles are of the same species, and an additional restriction holds depending on whether they are fermions or bosons.  In the fermionic case, we must have $|x_i|\leq 1$ while no such constraint holds for bosons.  To encode the message $a_i$, sender $\msf{A}_i$ performs some completely-positive trace-preserving map $\mc{E}^{\msf{A}_i}_{a_i}$ so that the final state obtained by $\msf{B}$ has the form
\begin{equation}
\label{eq:encoded-state}
    \sigma_{a_1\cdots a_N}=\mc{E}^{\msf{A}_1}_{a_1}\otimes\cdots\otimes \mc{E}^{\msf{A}_N}_{a_N}[\rho^{\bm{\msf{A}}}]
\end{equation}
The decoding by $\msf{B}$ is performed using a positive operator-valued measure (POVM) $\{\Pi_{b}\}_{b\in\mc{B}}$ so that the generated MAC has transition probabilities
\begin{equation}
\label{Eq:1-N-quantum}
    p(b|a_1\cdots a_N)=\tr\left[\Pi_b(\sigma_{a_1\cdots a_N})\right].
\end{equation}
Here we do not require that the output set $\mc{B}$ be equivalent to the input set $\bm{\mc{A}}$.  For instance in the task of quantum fingerprinting described above, $\mc{B}=[2]$ and the value $b$ just represents partial information about the collective inputs $(a_1,a_2,\cdots,a_N)$.

The crucial aspect of our approach is that we demand a strict accounting of all particles used in the communication protocol, and we prohibit the use of more particles than what is present in the initial state $\rho$. Hence, the local encoding maps $\mc{E}^{\msf{A}_i}_{a_i}$ need to be restricted so that they cannot increase particle number.  Even more, to make this constraint experimentally feasible, we require that each map be implementable by a process that preserves overall particle number.  This can be modeled by giving party $\msf{A}$ access to a collection of ancilla ports $\msf{E}_1,\msf{E}_2,\cdots,\msf{E}_K$, with each port $j=1,\cdots,K$ having its own set of number states $\{\ket{0}^{\msf{E}_{j}},\ket{1}^{\msf{E}_{j}},\cdots\}$ (see Fig. \ref{Fig:Dilation-free}).  Here we have temporarily omitted the party-labeling subscript on $\msf{A}$ for simplicity.   
A valid encoding map will then have a dilation $U$ on systems $\msf{A},\msf{E}_1,\cdots,\msf{E}_{K}$ that is number-preserving extendible.  More precisely, for every integer $0\leq J\leq M$, the dilation $U$ must satisfy
\begin{equation}
\label{Eq:unitary-N-P}
    U\left[\ket{x}_{\msf{A}}\ket{x_1}_{\msf{E}_{1}}\cdots\ket{x_K}_{\msf{E}_K}\right]\in\mc{H}^{\msf{A}\msf{E}_{1}\cdots \msf{E}_K}_J
\end{equation}
whenever $x+\sum_{i=1}^Kx_i=J$.  Since no additional particles can be introduced along the path to $\msf{B}$, we require that the ancilla systems always begin in the vacuum state $\ket{0}_{\msf{E}}:=\ket{0}_{\msf{E}_1}\cdots\ket{0}_{\msf{E}_K}$.  Then every allowed encoding map on system $\msf{A}$ will have the form
\begin{equation}
    \mc{E}(\rho^{\msf{A}})=\tr_{\msf{E}}\left[U(\rho^{\msf{A}}\otimes\op{0}{0}^{\msf{E}})U^\dagger\right],
\label{eq:map}
\end{equation}
where $U$ has the block-diagonal form of Eq.~\eqref{Eq:unitary-N-P}.  It is natural to assume that each party has access to some local randomness, and so convex combinations of these maps are also allowed. \begin{definition}
A completely positive trace-preserving (CPTP) map is called \textbf{number-preserving extendible} (NPE), if it can be extended as written as a convex combination of channels each having the form of Eq.~\eqref{eq:map}.
\end{definition}  
\begin{figure}[t]
    \includegraphics[width=0.5\textwidth]{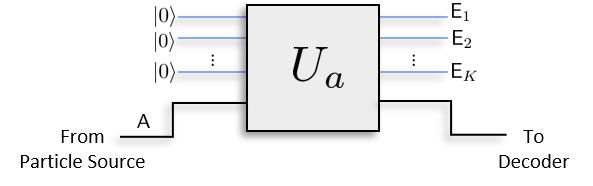}
     \caption{The number-preserving extendible (NPE) encoding maps have a dilation in which overall particle number is preserved.  The effective channel on system $A$ is obtained by tracing out the ancilla systems $E_1,E_2,\cdots,E_k$, and the output is sent to the decoder.}
     \label{Fig:Dilation-free}
\end{figure}
The structure of an NPE map $\mc{E}$ can be characterized in terms of its Kraus operators.  Here we restrict attention to $M$ bosons occupying the input mode of party $\msf{A}$ since the fermionic case is formally equivalent to the $M=1$ bosonic case.  Let $\bm{\msf{E}}:=\msf{E}_1,\cdots,\msf{E}_K$ denote the collection of all ancilla ports for a given NPE operation.  We then write
\begin{equation}
U=\sum_{m=0}^M\op{\psi_{m}}{m,0}^{\msf{A}\bm{\msf{E}}}
\end{equation}
where 
\begin{equation}
\ket{\psi_m}^{\msf{A}\bm{\msf{E}}}=\sum_{k=0}^m\ket{k}\ket{\varphi_{m,m-k}}
\end{equation}
and $\ket{\varphi_{m,m-k}}$ is an $(m-k)$-particle state for systems $\bm{\msf{E}}$.  Let $\{\ket{\mbf{e}_{n,\lambda}}\}_{n,\lambda}$ be an orthonormal basis of $\mc{H}^{\bm{\msf{E}}}$ such that $\ket{\mbf{e}_{n,\lambda}}$ is an $n$-particle state.  Here, $\lambda=1,\cdots,\binom{n+K-1}{n}$ is ranging over the number of ways $n$ particles can be distributed among the $K$ ancilla ports $\bm{\msf{E}}$.  Then a partial contraction with $\ket{\mbf{e}_{n,\lambda}}$ on systems $\bm{\msf{E}}$ will lead to
\begin{equation}
\ip{\mbf{e}_{n,\lambda}}{\psi_m}^{\msf{A}\bm{\msf{E}}}=c_{n,\lambda,m}\ket{m-n}
\end{equation}
with $c_{n,\lambda,m}=0$ if $n>m$.  Thus, the induced map on system $\msf{A}$ will have Kraus operators
\begin{align}
\label{Eq:Kraus-operators-NP}
E_{n,\lambda}=\sum_{m=n}^Mc_{n,\lambda,m}\op{m-n}{m},
\end{align}
with $n=0,\cdots,M$ and $\lambda=1,\cdots,\binom{n+K-1}{n}$.  The range of $\lambda$ can grow unbounded for each value of $n>0$ by taking more ancilla ports; however when $n=0$, there is only one permissible value of $\lambda$ regardless of the size of $K$.  Each $E_{n,\lambda}$ can be seen as a matrix whose only non-negative elements are on the $n^{\text{th}}$ upper diagonal.  Conversely, suppose we have a complete set of Kraus operators having the form of Eq.~\eqref{Eq:Kraus-operators-NP}.  A number-preserving dilation can be formed by defining states $\ket{\varphi_{m,m-k}}=\sum_{\lambda}c_{m-k,\lambda,m}\ket{\mbf{e}_{m-k,\lambda}}$ and then using the above construction.  We summarize in the following Proposition
\begin{proposition}
A CPTP map $\mc{E}$ is NPE if and only if it is a convex combination of CPTP maps, each of which has Kraus operators $\{E_{n,\lambda}\}_{n,\lambda}$ satisfying Eq.~\eqref{Eq:Kraus-operators-NP} for some finite $K$.
\end{proposition}
\noindent We note that these operations represent a subset of $U(1)$-covariant operations \cite{Gour_Spekkens}.  However there is a subtle constraint on the index $\lambda$ when $n=0$ that does not appear in the definition of a $U(1)$-covariant map.
\subsection{MACs with a Single Quantum Particle}
\label{Sect:Macs with Q-particle}
In this work, we focus on MACs that can be generated by a single particle and so henceforth we will take $M=1$ exclusively.  In this case, the structure of a general NPE map has a particularly nice form.  From Eq.~\eqref{Eq:Kraus-operators-NP}, a full set of Kraus operators will have the form $\{A,B_\lambda\}_{\lambda=1}^{K}$ with $A=\left(\begin{smallmatrix}1&0\\0&y\end{smallmatrix}\right)$ and $B_\lambda=\left(\begin{smallmatrix}0&z_\lambda\\0&0\end{smallmatrix}\right)$, and the normalization constraint being $1=|y|^2+\sum_{\lambda}|z_\lambda|^2$.  It is easy to see this map can be equivalently represented by Kraus operators $\{A,B\}$ in which 
\begin{align}
\label{Eq:Kraus-qubit-MAC}
A&=\begin{pmatrix}1&0\\0&y\end{pmatrix}&B&=\begin{pmatrix}0&z\\0&0\end{pmatrix}
\end{align}
with $z=\sum_\lambda|z_\lambda|^2$.  In summary, for $M=1$, an NPE operation is simply a mixture of amplitude damping channels.  Furthermore, to physically implement them, no more than a single ancilla port $\msf{E}_1$ is needed, e.g. phase encoding with $|y|=1$ and path blocking with $y=0$ can be realized with no ancilla port, and an arbitrary damping operation with $0<y<1$ can be implemented by coupling vacuum state with a beam-spiltter.

\begin{definition}
For fixed input/output sets $\bm{\mc{A}}$ and $\mc{B}$, an \textbf{$N$-local quantum} MAC is any channel in which $p(b|a_1\cdots a_N)=\tr\left[\Pi_b(\sigma_{a_1\cdots a_N})\right]$, with $\sigma_{a_1\cdots a_N}$ being a single-particle state encoded by NPE operations.  The collection of all such channels will be denoted as $\mc{Q}_{N}(\bm{\mc{A}};\mc{B})$, or simply $\mc{Q}_{N}$ when the input/output sets are clear.
\end{definition}
\noindent It is also interesting to consider the family of MACs that can be generated by a fixed input state $\rho^{\bm{\msf{A}}}$ under the restriction of NPE encoding.  We will write this set as $\mc{Q}_{N}(\bm{\mc{A}};\mc{B};\rho)$, or simply $\mc{Q}_{N}(\rho)$.  According to the formalism introduced here, we thus have
\begin{equation}
    \mc{Q}_{N}=\bigcup_{\rho\in\mc{B}(\mc{H}_1^{\bm{\msf{A}}})}\mc{Q}_{N}(\rho).
\end{equation}

\subsection{MACs with a Single Classical Particle}
\label{Sect:Macs with C-particle}
The classical case can easily be modeled by taking the initial state $\rho^{\bm{\msf{A}}}$ to be diagonal in the number basis.  The classical NPE operations correspond to $U$ being a permutation in Eq.~\eqref{Eq:unitary-N-P} and the decoding POVM being a projective measurement in the number basis followed by some post-processing of the measurement outcome.

To characterize the generated channels, we work this process out in more detail.  Every single-particle classical state has the form
\begin{equation}\label{Eq:state-classical}
\rho_{\cl}^{\bm{\msf{A}}}=\sum_{i=1}^N p_i \op{\mbf{e}_i}{\mbf{e}_i},
\end{equation}
where $\mbf{e}_i$ is the $i^{\text{th}}$ unit vector and $\ket{\mbf{e}_i}:=\ket{0}_{\msf{A}_1}\cdots\ket{1}_{\msf{A}_i}\cdots\ket{0}_{\msf{A}_N}$.  This can be understood as sending a classical particle (like a tennis ball) along path $i$ with probability $p_i$.  A local NPE operation amounts to stochastically letting the particle continue along its respective path or blocking it from reaching $\msf{B}$.  For each party $\msf{A}_i$, we can model this action by {a collection of }encoding maps $q_i:\mc{A}_i\to \mc{M}_i:=\{0,\mbf{e}_i\}$ with $q_i(0|a_i)$ being the probability that the particle along path $i$ is blocked for input $a_i$ and $q_i(\mbf{e}_i|a_i)$ being the probability that it is transmitted.  Hence for input state $\op{\mbf{e}_i}{\mbf{e}_i}$, the state received by $\msf{B}$ is 
\begin{align}
    \sigma_{a_i}&=\bigotimes_{j\not=i}\op{0}{0}^{\msf{A}_j}\otimes \mc{E}^{\msf{A}_i}_{a_i}(\op{1}{1})\notag\\ 
    &= q_i(\mbf{e}_i|a_i)\op{\mbf{e}_i}{\mbf{e}_i}^{\bm{\msf{A}}}+q_i(0|a_i)\op{0}{0}^{\bm{\msf{A}}}. 
\end{align}
The decoding process of party $B$ is defined by a projective measurement with classical post-processing.  Physically, this amounts to party $B$ examining each path to see if it contains a particle, and then sampling from $\mc{B}$ with probability distribution $d(b|\mbf{e}_i)$ if a particle is received along path $i$ and distribution $d(b|0)$ if no particle is received.  The stochastic mapping $d:\cup_{i=1}^N\mc{M}_i\to\mc{B}$ with transition probabilities $d(b|\mbf{e}_i)$ is referred to as the decoder, and the channel obtained after averaging over all input states is 
\begin{equation}
\label{Eq:1-N-classical}
    p(b|a_1\cdots  a_N)=\sum_{i=1}^Np_i [d(b|0)q_i(0|a_i)+d(b|\mbf{e}_i)q_i(\mbf{e}_i|a_i)]
\end{equation}
\begin{definition}
Any channel admitting a decomposition like Eq.~\eqref{Eq:1-N-classical} will be called an \textbf{$N$-local classical MAC}.  For given input/output sets $\bcA$ and $\mc{B}$, we denote the family of all such channels by $\mc{C}_{N}(\bcA;\mc{B})$, or simply $\mc{C}_{N}$.
\end{definition}

\subsection{Shared Randomness Models}

\label{Sect:SR-model}

\begin{figure}[b]
  \centering
    \includegraphics[width=0.3\textwidth]{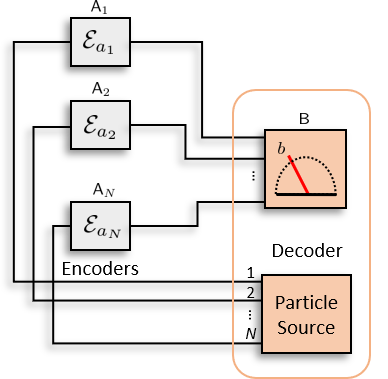}
     \caption{When randomness is shared between the source and the decoder, the decoder knows which path the particle takes, and the generated channels have the form of Eq.~\eqref{Eq:2-N-source}.  These MACs constitute the set $\mc{C}'_N(\bcA,\mc{B})$.}
     \label{Fig:N-Local-SR-source}
\end{figure}

We can modify the communication setup of Fig. \ref{Fig:N-Local} by equipping the senders and receiver with various additional resources. In this section, we introduce three additional communication models that have clear operational meaning. We will see that these three models, along with $\mc{C}_N$ previously defined, are all equivalent when the output is binary (Proposition \ref{Prop:binary-out})
, while they may be distinct when $|\mc{B}|>2$ (Tab. \ref{tab:N-local-comparision}). Let us begin by recalling the general definition of a MAC in $\mc{C}_N(\bcA,\mc{B})$:
\begin{equation}
\label{Eq:2-N-non-convex}
    p(b|a_1\cdots  a_N)=\sum_{i=1}^Np_i\sum_{m=0,\mbf{e}_i} d(b|m)q_i(m|a_i).
\end{equation}

\subsubsection{Which-Path Information}
Our first generalization allows the decoder to share correlations with the particle source (see Fig. \ref{Fig:N-Local-SR-source}). Operationally, this model captures the scenario in which the decoder knows which path the particle takes in each run of the experiment. Crucially in this model, when the decoder receives no particle, he/she knows which party performed the blocking operation, whereas for channels in $\mc{C}_N$ this information is not known. MACs generated in this scenario will have transition probabilities of the form 
\begin{align}
\label{Eq:2-N-source}
p(b|a_1\cdots  a_N)=\sum_{i=1}^Np_i\sum_{m=0,\mbf{e}_i} d_i(b|m)q_i(m|a_i),
\end{align}
and we denote the collection of these channels by $\mc{C}'_N(\bcA,\mc{B})$. Notice that now the decoding function $d_i$ can depend on the path $i$.
The set $\mc{C}'_N(\bcA,\mc{B})$ arises most naturally in the context of nonlocal games where the decoder $\msf{B}$ (often called a ``referee'') both prepares the particle and receives it in the end \cite{Cleve-2004a, Horvat-2019a}.  
\begin{figure}[b]
  \centering
    \includegraphics[width=0.5\textwidth]{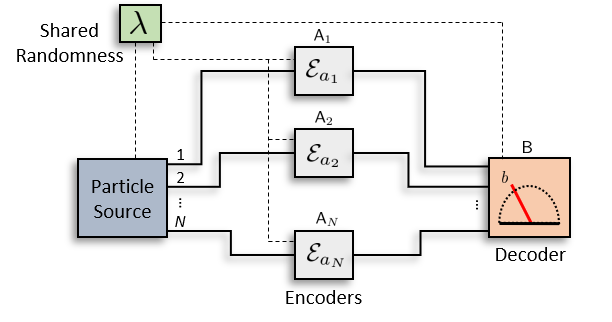}
     \caption{When randomness is shared between the source, the encoders, and the decoder, the generated channels have the form of Eq.~\eqref{Eq:2-N-convex}. These MACs constitute the convex hull of $\mc{C}_N$ and is denoted by $\conv[\mc{C}_N(\bcA,\mc{B})]$.}
     \label{Fig:N-Local-SR}
\end{figure}

\subsubsection{{Global Shared Randomness}}
An even more general model allows for shared randomness between the particle source, the encoders, and the decoder (see Fig. \ref{Fig:N-Local-SR}).  In the classical case, this describes the convex hull of $\mc{C}_N(\bcA,\mc{B})$, denoted by $\conv[\mc{C}_N(\bcA,\mc{B})]$, and a general MAC in $\conv[\mc{C}_N(\bcA,\mc{B})]$ has the form
\begin{align}
\label{Eq:2-N-convex}
     p(b|a_1\cdots  a_N)&=\sum_\lambda t_\lambda\sum_{i=1}^Np_{i|\lambda}\sum_{\mathclap{m=0,\mbf{e}_i}}d(b|m,\lambda)q_i(m|a_i,\lambda)\notag\\
     &=\sum_{i=1}^Np_{i}\sum_\lambda t_{\lambda|i}\sum_{\mathclap{m=0,\mbf{e}_i}}d(b|m,\lambda)q_i(m|a_i,\lambda). 
\end{align}
One way to interpret the difference between Eqns. \eqref{Eq:2-N-convex} and \eqref{Eq:2-N-source} is that Eq.~\eqref{Eq:2-N-source} still allows for some private randomness in the local encoding $q_i(m|a_i)$.  This is no longer the case for MACs in $\conv[\mc{C}_N(\bcA,\mc{B})]$, and without loss of generality it can be assumed that the $q_i(m|a_i,\lambda)$ are deterministic functions in Eq.~\eqref{Eq:2-N-convex} since any randomness can be absorbed into the distributions $t_{\lambda|i}$.  

\subsubsection{Separable MACs}
A final classical model we analyze is simply a convex combination of point-to-point channels between senders $\msf{A}_i$ and receiver $\msf{B}$.  This scenario removes the one-particle constraint entirely while still maintaining locality.  Since we provide a full characterization of these MACs in Section \ref{Sect:N-local-separable}, we highlight them in the following definition.
\begin{definition}
Any $N$-party MAC whose transition probabilities decompose like 
\begin{equation}
\label{Eq:1-N-separable}
    p(b|a_1\cdots  a_N)=\sum_{i=1}^Np_i g_i(b|a_i),
\end{equation}
will be called \textbf{separable}.  We denote the set of separable MACs by $\mc{C}_{N}^{(\text{sep})}(\bcA;\mc{B})$, or simply $\mc{C}_{N}^{(\text{sep})}$ when the input/output sets are clear.
\end{definition}
\noindent Physically, any MAC in $\mc{C}_{N}^{(\text{sep})}(\bcA;\mc{B})$ can be implemented by point-to-point communication between party $\msf{A}_i$ and $\msf{B}$ with probability $p_i$ and using a noisy channel with transition probabilities $g_i(b|a_i)$.  Such communication could be facilitated by introducing more particles to the system or using internal degrees of freedom of the single particle.  

By examining Eqns. \eqref{Eq:2-N-non-convex} -- \eqref{Eq:1-N-separable}, we see that
\begin{equation}
\label{eq:inclusion N-local}
    \mc{C}_{N}(\bcA;\mc{B})\subseteq \mc{C}'_{N}(\bcA;\mc{B})\subseteq\conv[\mc{C}_{N}(\bcA;\mc{B})]\subseteq\mc{C}_{N}^{(\text{sep})}(\bcA;\mc{B}).
\end{equation}
\subsection{(N,K)-local MACs with a single classical particle}
\label{Sect:N-K Macs}
In both the quantum and classical scenarios considered thus far, spatial separation is enforced between the parties and they are unable to communicate with each other.  However, we can consider relaxations to the locality constraint and again compare the powers of classical and quantum MACs \cite{Horvat-2019a}.  The construction is as follows (See also Fig \ref{Fig:N,K-Local}).  Let $S\subseteq\{1,2,\cdots,N\}$ represent a subset of $|S|$ paths with $|S|<N$.  We now suppose that all parties belonging to these paths can coordinate their signal to the receiver.  This means that for each choice of messages, the parties can map a particle traveling along any path in $S$ to any other path in $S$.  Like before, we represent this by a stochastic encoding map $q_S:\bigtimes_{s\in S}\mc{A}_s\to \mc{M}_S$ where $\mc{M}_S=\{0,\mbf{e}_s:s\in S\}$.  
\begin{figure}[t]
  \centering
    \includegraphics[width=0.5\textwidth]{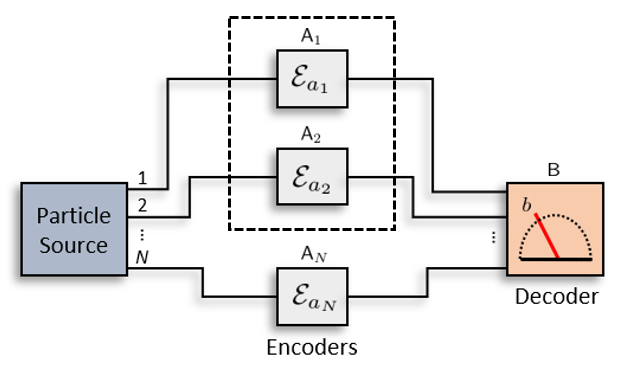}
    \caption{A physical implementation of the $(N,K)$-local MAC $\mbf{p}_{B|\bm{A}}$ using a single particle, where at most $K$ parties are allowed to encode jointly ({the case $K=2$ is depicted}). Classically, the joint encoding will be simply implemented as a permutation between $K$ paths and the vacuum ancilla port. The joint encoding will be represented by stochastic encoding maps $q_S:\bigtimes_{s\in S}\mc{A}_s\to \mc{M}_S$ where $\mc{M}_S=\{0,\mbf{e}_s:s\in S\}$}
     \label{Fig:N,K-Local}
\end{figure}
\begin{definition}
For non-negative integer $K$, a classical MAC will be called \textbf{$(N,K)$-local} if it can be decomposed as
\begin{equation}
\label{Eq:N,K-local-general}
    p(b|a_1\cdots  a_N)=\sum_{\substack{S\\|S|=K}}p_S\sum_{m\in\mc{M}_S}[d(b|m)q_S(m|(a_s)_{s\in S})],
\end{equation}
where the outer sum is over all subsets of $\{1,2,\cdots, N\}$ whose cardinality is $K$, and $p_S$ is the probability that the particle is initially prepared in one of the paths belonging to $S$.  The collection of all such MACs will be denoted by $\mc{C}_{N,K}(\bcA;\mc{B})$, or simply $\mc{C}_{N,K}$.  Note that $\mc{C}_{N}=\mc{C}_{N,1}$.
\end{definition}
For the sets $\mc{C}_{N,K}$ with $K>1$, even more models can be considered under the introduction of shared randomness.  This is due to the fact that the receiver need not know which subset of $K$ parties are jointly encoding, even if the initial particle is known.  A somewhat less complex scenario is when the decoder knows both the initial particle as well as the jointly encoding parties.  MACs generated in this way have the form
\begin{equation}
\label{eq: with rand}
    p(b|a_1\cdots  a_N)=\sum_{\mathclap{\substack{S\\|S|=K}}}p_S\sum_{\mathclap{m\in\mc{M}_S}}[d_S(b|m)q_S(m|(a_s)_{s\in S})],
\end{equation}
with $\mc{M}_s$ being an encoding set consisting of $K+1$ elements $\{0,\mbf{e}_s:s\in S\}$.  We denote the collection of these MACs by $\mc{C}_{N,K}'(\bcA,\mc{B})$.  Compare this with the set $\mc{C}_{N,K}(\bcA,\mc{B})$ in Eq.~\eqref{Eq:N,K-local-general}, which consists of MACs in which the decoder $d(b|m)$ does not depend on knowledge of the encoding parties $S$. We can also consider the convex hull  $\conv[\mc{C}_N(\bcA,\mc{B})]$ and the set of $(N,K)$-separable MACs, $\mc{C}_{N,K}^{(\text{sep})}(\bcA,\mc{B})$.  The latter consists of MACs having the form  
\begin{equation}
\label{Eq:N,K-local-binary-out}
 p(b|a_1\cdots  a_N)=\sum_{\substack{S\\|S|=K}}p_Sg_S(b|(a_s)_{s\in S}).
\end{equation} 
Similar to Eq.~\eqref{eq:inclusion N-local}, we have an inclusion relationship between different $(N,K)$-local MACs as $\mc{C}_{N,K}(\bcA;\mc{B})\subseteq \mc{C}'_{N,K}(\bcA;\mc{B})\subseteq\conv[\mc{C}_{N,K}(\bcA;\mc{B})]\subseteq\mc{C}_{N,K}^{(\text{sep})}(\bcA;\mc{B})$. More discussion on the relationships between these classes will be given in Section \ref{Sect:Beyond-binary}.  However, we close this section by making an important observation for binary-output MACs.
\begin{proposition}
\label{Prop:binary-out}
$\mc{C}_{N,K}(\bcA;[2])=\mc{C}_{N,K}^{(\text{sep})}(\bcA;[2])$ for arbitrary input set $\bcA$ and any $K\geq 1$. 
\end{proposition}
\begin{proof}
Clearly $\mc{C}_{N,K}(\bcA;[2])\subseteq\mc{C}_{N,K}^{(\text{sep})}(\bcA;[2])$, conversely Let $p(b|\mbf{a})=\sum_{s=1}^Np_Sg_S(b|(a_s)_{s\in S})$ be an arbitrary separable MAC with $b\in\{0,1\}$.   Consider the deterministic decoder $d(0|0)=1$ and $d(1|\mbf{e}_i)=1$ for all $i$, and encoders $q_S(0|(a_s)_{s\in S})=g_S(0|(a_s)_{s\in S})$ and $q_S(\mbf{e}_S|(a_s)_{s\in S})=g_S(1|(a_s)_{s\in S})$, where $\mbf{e}_S$ can be any fixed element in $\mc{M}_S$. With these choices, we can construct any point in $\mc{C}_{N,K}^{(\text{sep})}(\bcA;[2])$ with point in $\mc{C}_{N,K}(\bcA;[2])$ 
\end{proof}
{The above proposition can not be generalized to arbitrary classical MACs. For example, with $N=K=1$ and $|\mc{B}|>2$, the MACs in $\mc{C}_{1,1}(\bcA;\mc{B})$ can generate at most one bit of information since there are only two different outputs (one particle/no particle), however, Macs in $\mc{C}_{1,1}^{(\text{sep})}(\bcA;\mc{B})$ doesn't have such limitation. A more general discussion will be given in the next section.}
\section{Classical MACs}
\label{Sect:classical MACs}
Let us now examine in more detail the structure of the classical sets $\mc{C}_{N}(\bcA;\mc{B})$ and $\mc{C}_{N,K}(\bcA;\mc{B})$.  This will allow us to make a comparison with quantum MACs in the next section.  Following the standard approach \cite{Brunner-2014a}, we envision each MAC in $\mc{C}_{N,K}(\bcA;\mc{B})$ as a point $\mbf{p}_{B|\bA}$ in ($|\mc{B}|\times\prod_{k=1}^N|\mc{A}_k|$)-dimensional Euclidean space with coordinates $(p(b|a_1\cdots  a_N))_{b\in\mc{B},a_i\in\mc{A}_i}$.  To proceed with our analysis, the question of whether or not $\mc{C}_{N,K}(\bcA;\mc{B})$ is convex is vitally important.  
As a consequence of Proposition \ref{Prop:binary-out} above, when $|\mc{B}|=2$ the set $\mc{C}_{N,K}(\bcA;[2])$ is a convex polytope whose vertices are deterministic $K$-party MACs, which are those satisfying $p(b|a_1\cdots  a_N)=p(b|a_{i_1},\cdots,a_{i_K})\in\{0,1\}$ for $\{i_1,\cdots,i_K\}\subseteq\{1,\cdots,N\}$.  Considering all possible groupings of $K$ parties with $M=\max\{|\mc{A}_1|,\cdots,|\mc{A}_N|\}$, we see there are no more than $\binom{N}{K}2^{M^K}$ such vertices, denoted by $\mbf{v}_{\lambda}$, and $\mc{C}_{N,K}$ is the polytope contained in their convex hull; i.e. 
\begin{equation}
\label{Eq:polytope-vertex}
\mbf{p}_{B|\bA}\in\mc{C}_{N,K} \quad\Leftrightarrow\quad \mbf{p}_{B|\bA}=\sum_{\lambda} p_\lambda \mbf{v}_{\lambda}
\end{equation}
for a valid probability distribution $(p_\lambda)_{\lambda}$.  

Before moving forward, let us first recall a few general facts about convex and affine sets (see Ref. \cite{Barvinok-2002a} for the details).  An affine subspace $\mc{A}\subset\mbb{R}^d$ is a collection of points that is closed under affine combinations, i.e. linear combinations of the form $\sum_i \lambda_i\mbf{v}_i$ with $\sum_i\lambda_i=1$.  A collection of points $\{\mbf{v}_0,\mbf{v}_1,\cdots\mbf{v}_l\}$ in $\mbb{R}^d$ are called affinely independent if the $l$ vectors $\{\mbf{v}_1-\mbf{v}_0,\cdots,\mbf{v}_l-\mbf{v}_0\}$ are linearly independent.  An affine subspace $\mc{A}$ is said to have dimension $\dim\mc{A}$ if the maximum number of affinely independent points it contains is $\dim\mc{A}+1$.   For an arbitrary collection of points $\mc{P}\subset\mbb{R}^d$, its dimension is the dimension of the smallest affine subspace that contains $\mc{P}$.  If all the vectors in $\mc{P}$ are known to satisfy a system of $k$ linearly independent equations, then Gaussian elimination shows that $\mc{P}$ is contained in an affine subspace of dimension $d-k$.  If, further, $d-k+1$ affinely independent vectors are shown to exist in $\mc{P}$, then $\mc{P}$ has dimension $d-k$. 

Let $\mc{P}$ be a convex polytope.  For a fixed $\mbf{r}\in\mbb{R}^d$ and $s\in\mbb{R}$, we say that an inequality $\mbf{v}\cdot\mbf{r}\leq s$ is ``valid'' for a  $\mc{P}$ if it is satisfied by every $\mbf{v}\in\mc{P}$.  
Valid inequalities of a polytope are useful when trying to certify that some element $\mbf{v}$ is not a member of $\mc{P}$ since it suffices to show that $\mbf{v}$ does not satisfy the particular relation $\mbf{v}\cdot\mbf{r}\leq s$, a simple linear calculation.  On the other hand, if $\mbf{v}$ satisfies the inequality, then in general one cannot conclude that $\mbf{v}\in\mc{P}$; additional conditions are needed to certify membership.  Specifically, the Weyl-Minkowski Theorem states that any $d$-dimensional polytope $\mc{P}$ can be characterized by a finite family of valid and ``tight'' inequalities.  In other words,
\begin{equation}
\label{Eq:polytope-polyhedron}
\mbf{v}\in\mc{P} \Leftrightarrow \mc{P}=\{\mbf{v}\in\mbb{R}^d\;|\:\mbf{v}\cdot\mbf{r}_\lambda\leq s_\lambda,\;\forall \lambda=1,\cdots,n\}.
\end{equation}
Associated with each inequality $\mbf{v}\cdot\mbf{r}_\lambda\leq s_\lambda$ is the hyperplane $H_\lambda=\{\mbf{v}\;|\;\mbf{v}\cdot\mbf{r}_\lambda= s_\lambda\}$, and since $\mc{P}$ is $d$-dimensional, its intersection with $H_\lambda$ forms a $(d-1)$-dimensional affine subspace.  The property $\dim H_\lambda\cap\mc{P}=d-1$ is what it means for the corresponding inequality to be tight, and hence a point $\mbf{v}$ is an element of $\mc{P}$ if and only if it satisfies a finite family of valid and tight inequalities.  In the language of quantum information science, these are known as tight Bell Inequalities \cite{Masanes-2002a, Pironio-2005a}.

Returning to the problem at hand, we aim to compute tight Bell Inequalities for the polytope $\mc{C}_{N,K}(\bcA;\mc{B})$. Moving from the V-representation of $\mc{C}_{N,K}(\bcA;\mc{B})$ given in Eq.~\eqref{Eq:polytope-vertex} to an H-representation in the form of Eq.~\eqref{Eq:polytope-polyhedron} can be a formidable task.  We rely mainly on numerical software such as PORTA, which for small dimensions performs the calculation using the Fourier-Motzkin elimination method in a short amount of time \cite{porta-1997}.  However, we are able to analytically prove that the so-called fingerprinting inequality (see Eq.~\eqref{eq:fpg}) is a tight Bell Inequality for arbitrary $N$ and $K$ when dealing with binary inputs and output.  This is done in Section \ref{Sect:Generalize-fingerprinting} by first calculating the dimension of $\mc{C}_{N,K}([2]^N;[2])$ and then showing that a sufficiently large number of affinely independent points saturate the inequality.  The case of $N$-local MACs (i.e. $K=1$) is special in that its dimension is $N+1$ and the only inequalities are positivity constraints. We turn next to establishing this result.

\subsection{$N$-Local and Separable MACs}

\label{Sect:N-local-separable}

In this section, we aim to characterize the sets of $N$-local MACs $\mc{C}_{N}(\bcA;\mc{B})$.  As indicated above, $\mc{C}_{N}(\bcA;\mc{B})$ is a convex polytope when $|\mc{B}|=2$, whereas this fails to be the case when $|\mc{B}|>2$ (see Section \ref{Sect:Beyond-binary}). To provide a unified treatment based on convexity, we will therefore consider $\conv[\mc{C}_N(\bcA;\mc{B})]$ and the more general $\mc{C}_N^{(\text{sep})}(\bcA;\mc{B})$ when $|\mc{B}|>2$.  
 
 Important quantities in our study of $N$-local polytopes are the so-called second-order interference \cite{Born-1926,Grangier-1986,Jacques-2005,Sorkin-1994a}.  The physical meaning of this name originates from double-slit experiments, as will be described more in Section \ref{Sect:N-local-quantum}.  However, for the moment we simply identify it as a particular linear functional on the set of $N$-partite MACs \cite{Biswas-2017a}.  For two parties with binary inputs, the second-order interference is given by
\begin{equation}
\label{Eq:I2}
    I_2=p(0|0,0)+p(0|1,1)-p(0|0,1)-p(0|1,0).
\end{equation}
As we add more parties and more inputs/outputs, we consider this basic linear combination for certain pairs of two-party inputs. For a given MAC $\mbf{p}_{B|\bA}$, let $\mc{M}^{(i,j)}(\mbf{p}_{B|\bA})$ denote the set of all two-sender MACs obtained by fixing different inputs for $a_k$, $k\not=i,j$.  Then define $ I_2^{(i,j)}(\mbf{p}_{B|\bA})$ as:
\begin{align}
  \max 
    \;&\left|p(b|a_i,a_j)+p(b|a_i',a_j')-p(b|a_i,a_j')-p(b|a_i',a_j)\right|\notag\\
    \text{s.t}&\;\; \mbf{p}_{B|A_iA_j}\in\mc{M}^{(i,j)}(\mbf{p}_{B|\bA})\notag\\
    &\;\;b\in\mc{B},\; \;a_i,a_i'\in\mc{A}_i,\;\;a_j,a_j'\in\mc{A}_j.
\end{align}
If $I_2^{(i,j)}(\mbf{p}_{B|\bA})=0$, then parties $\msf{A}_i$ and $\msf{A}_j$ cannot demonstrate any nonzero second-level interference using $\mbf{p}_{B|\bA}$, regardless of what inputs the other parties choose and what output is considered.  
\begin{definition}
The set of all $N$-sender MACs $\bcA\to\mc{B}$ whose second-order interference vanish for all $i,j$ will be denoted by $\mf{I}_N(\bcA;\mc{B})$; i.e. \begin{equation}
     \mbf{p}_{B|\bA}\in \mf{I}_N(\bcA;\mc{B}) \Leftrightarrow I_2^{(i,j)}(\mbf{p}_{B|\bA})=0\;\;\;\forall (i,j).
 \end{equation}
\end{definition}

The main result of this section is that $\mf{I}_N(\bcA;\mc{B})$ corresponds precisely to the set of $N$ separable MACs. For binary input/output MACs, this fact has also been independently established in the Masters Thesis of Horvat Ref.\cite{Horvat-2019b}. Here we prove the relationship for arbitrary inputs and outputs.
\begin{theorem}
\label{Thm:classical-I2-N}
\begin{equation}
\mc{C}^{(\text{sep})}_N(\bcA;\mc{B})=\mf{I}_N(\bcA;\mc{B})
\end{equation}
for any $\bcA$ and $\mc{B}$.
\end{theorem}
\noindent Since $\mc{C}^{(\text{sep})}_N(\bcA;[2])=\mc{C}_N(\bcA;[2])$ (as shown in Proposition \ref{Prop:binary-out}), this theorem immediately implies the following.
\begin{corollary}
\begin{equation}
\mc{C}_N(\bcA;[2])=\mf{I}_N(\bcA;[2])
\end{equation}
for any $\bcA$.
\end{corollary}
We now turn to the proof of Theorem \ref{Thm:classical-I2-N}.  From a direct evaluation of $I_2^{(i,j)}(\mbf{p}_{B|\bA})$ for any separable $\mbf{p}_{B|\bm{A}}$ given in Eq.~(\ref{Eq:1-N-separable}), it is easy to verify that $\mc{C}^{(\text{sep})}_N(\bm{\mc{A}};\mc{B})\subseteq \mf{I}_N(\bcA;\mc{B})$.  

To prove the converse, the main idea will be to identify the extreme points of $\mf{I}_N(\bcA;\mc{B})$.  Recall that a MAC $\mbf{p}_{B|\bA}\in\mf{I}_N(\bcA;\mc{B})$ is extremal if it cannot be decomposed into a convex combination of other MACs belonging to $\mf{I}_N(\bcA;\mc{B})$, i.e. $\mbf{p}_{B|\bA}\not=\lambda \mbf{p}_{B|\bA}'+(1-\lambda)\mbf{p}_{B|\bA}''$ for distinct $\mbf{p}_{B|\bA}',\mbf{p}_{B|\bA}''\in\mf{I}_N(\bcA;\mc{B})$ and $\lambda\in (0,1)$.  We will show that every extremal MAC in $\mf{I}_N(\bcA;\mc{B})$ is local deterministic.  This means that it has the form
\begin{equation}
\label{Eq:local-deterministic-MAC}
    p(b|a_1\cdots  a_N)=\delta_{bf(a_i)}
\end{equation}
for some party $\msf{A}_i$ and function $f:\mc{A}_i\to\mc{B}$.  Since these MACs belong to $\mc{C}^{(\text{sep})}_N(\bcA;\mc{B})$, it follows that $\mf{I}_N(\bcA;\mc{B})\subseteq \mc{C}^{(\text{sep})}_N(\bcA;\mc{B})$.  We first prove the case when $\bcA=[2]\times[2]$.
\begin{lemma}
\label{Lem:extremal-binary-input}
Every extremal MAC in $\mf{I}_2([2]\times[2];\mc{B})$ is local deterministic, and therefore $\mf{I}_N([2]\times[2];\mc{B})\subseteq \mc{C}^{(\text{sep})}_N([2]\times[2];\mc{B})$.
\end{lemma}
\begin{proof}
We begin by expressing $\mf{I}_2([2]\times[2];\mc{B})=\{p_{B|\bA}\;|\;\text{such that Eqns. \eqref{Eq:I2-cons-1} -- \eqref{Eq:I2-cons-3} hold}\}$:
\begin{subequations}
\begin{align}
    1&=\sum_{b\in\mc{B}} p(b|a_1,a_2),\forall (a_1,a_2)\ne (1,1)\label{Eq:I2-cons-1}\\
    0&=p(b|0,0)+p(b|1,1)-p(b|0,1)-p(b|1,0)\label{Eq:I2-cons-2}\\
    0&\leq p(b|a_1,a_2)\label{Eq:I2-cons-3}.
\end{align}
\end{subequations}
Note that $1=\sum_bp(b|1,1)$ is also implied when Eqns. \eqref{Eq:I2-cons-1} and \eqref{Eq:I2-cons-2} hold.  If we consider MACs in $\mf{I}_2([2]\times[2];\mc{B})$ as elements of $\mbb{R}^{4|\mc{B}|}$ (with components $p(b|a_1,a_2)$), then $\hat{p}_{B|\bA}$ is extremal if and only if $4|\mc{B}|$ linearly independent constraints are binding (i.e. tight) among those listed in Eqns. \eqref{Eq:I2-cons-1}--\eqref{Eq:I2-cons-3}.  We see that Eqns. \eqref{Eq:I2-cons-1} and \eqref{Eq:I2-cons-2} represent $3+|\mc{B}|$ equality constraints. Thus, the remaining binding constraints must come from \eqref{Eq:I2-cons-3}.  Hence if $\hat{p}_{B|\bA}$ is extremal, then it has at least $4|\mc{B}|-(3+|\mc{B}|)=3|\mc{B}|-3$ vanishing probabilities $\hat{p}(b|a_1,a_2)$.

For the case when $|\mc{B}|=2$, an extremal $\hat{p}_{B|\bA}$ must have at least $3$ vanishing probabilities.  Equations \eqref{Eq:I2-cons-1} and \eqref{Eq:I2-cons-2} then imply that the only possibilities for $(\hat{p}(0|0,0),\hat{p}(0|0,1),\hat{p}(0|1,0),\hat{p}(0|1,1))$ are $(0,0,0,0)$, $(1,1,1,1)$, $(0,1,0,1)$, $(1,0,1,0)$, $(0,0,1,1)$, and $(1,1,0,0)$.  Each of these corresponds to a deterministic local MAC (i.e having the form of Eq.~\eqref{Eq:local-deterministic-MAC}).  

For the case when $|\mc{B}|=3$, an extremal $\hat{p}_{B|\bA}$ must have at least $6$ vanishing probabilities, and let us consider these for each output $b\in\{0,1,2\}$.  If there is some $b'$ such that $\hat{p}(b'|a_i,a_j)=0$ for three (or more) distinct input pairs, then Eq.~\eqref{Eq:I2-cons-2} implies that it also vanishes for the fourth input pair.  In this event, outcome $b'$ occurs with zero probability and $\hat{p}_{B|\bA}$ reduces to an extremal MAC with $|\mc{B}|=2$; i.e. it is a local deterministic MAC.  The only other alternative is that for each $b\in\{0,1,2\}$, the probabilities $\hat{p}(b|a_i,a_j)$ vanish for exactly two input pairs $(a_i,a_j)$.  Hence only two nonzero terms can appear on the right-hand side of Eq.~\eqref{Eq:I2-cons-2} for each $b\in\{0,1,2\}$.  This means, up to relabeling, we must have relationships of the form $1=\hat{p}(0|0,0)=\hat{p}(0|0,1)$, $\hat{p}(1|1,1)=\hat{p}(1|1,0)>0$, and $\hat{p}(2|1,1)=\hat{p}(2|1,0)>0$.  However, Eq.~\eqref{Eq:I2-cons-1} then does not represent three additional linearly independent constraints, and so $\hat{p}_{B|\bA}$ cannot be extremal.

For $|\mc{B}|>3$, the lemma is easily proven by induction.  Indeed, suppose that every extremal MAC is local deterministic for $|\mc{B}|=N\geq 3$.  Then consider an extremal MAC $\hat{p}_{B|\bA}$ in $\mf{I}_2([2]\times[2];[N+1])$.  By the above observation, there must be at least $3|\mc{B}|-3$ vanishing probabilities.  Consequently, there will be at least one outcome $b'$ such that $\hat{p}(b'|a_i,a_j)=0$ for three (or more) distinct input pairs provided $3|\mc{B}|-3>2|\mc{B}|$, which is true since $|\mc{B}|=N+1\geq 4$.  Thus, Eq.~\eqref{Eq:I2-cons-2} requires that $\hat{p}(b'|a_i,a_j)$ vanishes for the fourth input pair as well.  Like before, this means that outcome $b'$ occurs with zero probability.  Hence $\hat{p}_{B|\bA}$ is extremal in $\mf{I}_2([2]\times[2];[N])$, which by our inductive assumption means that $\hat{p}_{B|\bA}$ is local deterministic.

\end{proof}

We next consider the set $\mf{I}_2(\mc{A}_1\times\mc{A}_2;\mc{B})$ for arbitrary $\mc{A}_1$ and $\mc{A}_2$.  It follows as a corollary of Lemma \ref{Lem:extremal-binary-input}  that $\mf{I}_2(\mc{A}_1\times\mc{A}_2;\mc{B})\subseteq \mc{C}^{(\text{sep})}_2(\mc{A}_1\times\mc{A}_2;\mc{B})$.  This can be seen by considering binary representations of the inputs $a_1=a_{11}\cdots a_{1m}$ and $a_2=a_{21}\cdots a_{2n}$, where $a_{1i},a_{2j}\in [2]$.  For simplicity, suppose that $m=n=2$, but the general case follows by the same reasoning.  Suppose that $p(b|a_{11}a_{12},a_{21}a_{22})$ are transition probabilities for a MAC in $\mf{I}_2([4]\times[4];[2])$.  The key observation is that for each fixed values of $a_{12}$ and $a_{22}$, the probabilities $p(b|a_{11}a_{12},a_{21}a_{22})$ describe a MAC in $\mf{I}_2([2]\times[2];\mc{B})$ with inputs $a_{11}$ and $a_{21}$.  By Lemma \ref{Lem:extremal-binary-input}, we have a separable decomposition of $p(b|a_{11}a_{12},a_{21}a_{22})$ as:
\begin{align*}
\lambda p(b|a_{11}a_{12},a_{22})+(1-\lambda)p(b|a_{12},a_{21}a_{22}).
\end{align*}
Next, we let only $a_{11}$ and $a_{22}$ vary.  In doing so we observe that $p(b|a_{11}a_{12},a_{22})\in\mf{I}_2([2]\times[2];\mc{B})$ 
so that
\begin{align*}
   p(b|a_{11}a_{12},a_{22})=\lambda_2 p(b|a_{11}a_{12})+(1-\lambda_2)p(b|a_{12},a_{22}).
\end{align*}
Similar reasoning for a varying $a_{12}$ and $a_{21}$ allows us to decompose
\begin{align*}
    p(b|a_{12},a_{21}a_{22})&=\lambda_3 p(b|a_{12},a_{22})+(1-\lambda_3)p(b|a_{21}a_{22}).
\end{align*}
Thus $p(b|a_{11}a_{12},a_{21}a_{22})$ becomes
\begin{align*}
\mu_1p(b|a_{12},a_{22}) +\mu_2 p(b|a_{11}a_{12})+\mu_3p(b|a_{21}a_{22})
\end{align*}
for some probabilities $\mu_i$.  Finally, by letting just the inputs $a_{12}$ and $a_{22}$ vary, we obtain the decomposition
\begin{align*}
     p(b|a_{11}a_{12},a_{21}a_{22})&=\mu_1(\nu p(b|a_{12})+(1-\nu)p(b|a_{22}))\notag\\ &+\mu_2 p(b|a_{11}a_{12}) +\mu_3p(b|a_{21}a_{22}). 
\end{align*}
Clearly this is an element of $\mc{C}^{(\text{sep})}_2(\mc{A}_1\times\mc{A}_2;\mc{B})$.  In general, for inputs $a_1=a_{11}\cdots a_{1m}$ and $a_2=a_{21}\cdots a_{2n}$, one considers varying each pair of binary inputs $(a_{1i},a_{2j})$ to construct a separable decomposition.

This binary splitting technique also works to scale up the number of parties.  Specifically if $p(b|\mbf{a})$ are transition probabilities for an $N$-sender MAC in $\mf{I}_N(\bcA;\mc{B})$, then one fixes the inputs for all but two senders.  What remains is a two-party MAC in $\mf{I}_2(\mc{A}_1\times\mc{A}_2;\mc{B})$ that can be separated between the two non-fixed senders.  On each of the remaining branches, again fix all but two of the senders, and a further separation can be performed.  Reiterating this procedure, one arrives at an $N$-party separable MAC.  These arguments thus establish $\mf{I}_N(\bcA;\mc{B})\subseteq\mc{C}^{(\text{sep})}_N(\bcA;\mc{B})$.
\\
\subsection{Binary $(N,K)$-Local Channels}
\label{Sect:N,K binary}

We next turn to $(N,K)$-local MACs.  In this section we restrict attention to binary inputs and output (i.e. $\mc{A}_i=\mc{B}=[2]$),  and to simplify the notation, we write $\mc{C}_{N,K}=\mc{C}_{N,K}([2]^N;[2])$. A more general discussion on MACs $\mc{C}_N(\bcA;\mc{B})$ with arbitrary inputs and outputs will be given in Appendix~\ref{sec:A2 dimension}.  In light of Proposition \ref{Prop:binary-out}, every MAC in $\mc{C}_{N,K}$ is separable and therefore has a decomposition like Eq.~(\ref{Eq:N,K-local-binary-out}).

Much of our analysis will involve generalizing the $I_2$ quantity given in Eq.~\eqref{Eq:I2}.  This can be done by considering the expression
\begin{align}
\label{eq:IK equality}
    I_{K}=\sum_{a_1,\cdots,a_{K}\in\{0,1\}}\prod_{i=1}^{K}(-1)^{a_i}p(0|a_1\cdots a_{K}),
\end{align}
which is the $K^{\text{th}}$-order interference \cite{Sorkin-1994a}.  For $N$ parties, take $S\subseteq\{1,\cdots,N\}$ with $|S|=K$, and let $\mc{M}^{(S)}(\mbf{p}_{B|\bA})$ be the collection of $K$-sender MACs obtained by fixing different inputs for $a_i\in\{0,1\}$ with $i\not\in S$.  Then define the quantity
\begin{widetext}
\begin{align}
\label{eq:I3 equality}
I_{K}^{(S)}(\mbf{p}_{B|\bA}):=\max\left\{\sum_{a_1,\cdots,a_{K}\in\{0,1\}}\prod_{i=1}^{K}(-1)^{a_i}p(0|a_1\cdots a_{K})\;\bigg|\; \mbf{p}_{B|A_1,\cdots,A_{K}}\in \mc{M}^{(S)}(\mbf{p}_{B|\bA})\right\}.
\end{align}
\end{widetext}
The condition $I_{K}^{(S)}(\mbf{p}_{B|\bA})=0$ means that parties $S$ have no $K^{\text{th}}$-order interference in the MAC $\mbf{p}_{B|\bA}$.  We let $\mf{I}_{N,K}$ denote the collection of MACs such that $I^{(S)}_{K+1}(\mbf{p}_{B|\bA})=0$ for all $S\subseteq\{1,\cdots,N\}$ with $|S|=K+1$. One can simply observe that $\mf{I}_{N,1}\subset\mf{I}_{N,2}\subset\cdots\subset\mf{I}_{N,N}$ by definition. Notice that $\mf{I}_{N,1}=\mf{I}_N([2]^N;[2])$, with the latter introduced in the previous section.  

It is easy to see that $\mc{C}_{N,K}\subseteq\mf{I}_{N,K}$.  Indeed by convexity it suffices to observe that 
$\sum_{a_{K+1}\in\{0,1\}}(-1)^{a_{K+1}}p(0|a_1,\cdots,a_K)=0$ 
for all choices of $a_1,\cdots, a_K$.  However, unlike the case for $\mc{C}_{N,1}$, the converse inclusion is not true for $K>1$ and more constraints are needed to characterize $\mc{C}_{N,K}$ than just vanishing $(K+1)^{\text{th}}$-order interference. Nevertheless, $\mc{C}_{N,K}$ and $\mf{I}_{N,K}$ have the same dimension, as we demonstrate next.  

\subsubsection{The Dimension of $\mc{C}_{N,K}$}
\label{Sect:dim-binary-N,K}

As explained above, $\mc{C}_{N,K}$ is a polytope in $\mbb{R}^{2^{N+1}}$.  To compute the dimension of $\mc{C}_{N,K}$, we will first show that its elements satisfy $2^N+\sum_{k=K+1}^N\binom{N}{k}$ linearly independent equations.  We then show that $2^{N+1}-2^N-\sum_{k=K+1}^N\binom{N}{k}+1$ affinely independent points belong to $\mc{C}_{N,K}$.  From this, we obtain the following theorem
\begin{theorem}
\label{thm:Rank Ck}
 $\dim\mathcal{C}_{N,K}=\sum_{k=0}^K\binom{N}{k}$.
\end{theorem}
\begin{proof}

Each $\mbf{p}_{B|\bA}\in\mc{C}_{N,K}$ satisfies $2^N$ normalization conditions $1=p(0|a_1 \cdots a_N)+p(1|a_1,\cdots,a_N)$, with $a_i\in\{0,1\}$. To facilitate our analysis, instead of working in the previous transition coordinates $p(0|a_1 \cdots a_N)$, we introduce interference coordinates, whose definition is based on the $I_K$ equalities in Eq.~(\ref{eq:IK equality}):
\begin{widetext}
\begin{align}
\label{eq:full coordinates}
\bigcup_{k\in\{0,\cdots,N\}}\bigcup_{\substack{S\subseteq\{1,\cdots,N\}\\|S|=k}}\left\{q(S):=(-1)^k\sum_{\substack{a_s\in\{0,1\}\\s\in S}}\prod_{i\in S}(-1)^{a_i}p(0|a_1 \cdots a_N) \;\bigg|\;\text{$a_j=0$ for $j\not\in S$}\right\}.
\end{align}
\end{widetext}
There are $\sum_{k=0}^N{{N}\choose{k}}=2^N$ elements in this set, which is the same as the original set of transition coordinates $p(0|a_1 \cdots a_N)$. Now we show that they are also linearly independent.  For $|S|=k$, consider the specific interference coordinate with $S=\{1,\cdots,k\}$
\begin{equation}
q(S)=(-1)^k\sum_{\substack{a_s\in\{0,1\}\\ s\in S}} \prod_{i\in S}(-1)^{a_i}p(0|a_1\cdots a_{k},\overbrace{0\cdots0}^{N-k}).
\end{equation}
Focus on the term $p(0|\overbrace{1\cdots1}^{k},\overbrace{0\cdots0}^{N-k})$ which appears in this term.  Different interference coordinates $q(S')$ with $|S'|=k$ can be obtained by permuting the $0$'s and $1$'s, yet, for each permutation there will be exactly one and only one distinct probability term in the sum that is conditioned on $k$ values of $1$.  Hence, the different permutations generate $\binom{N}{k}$ linearly independent coordinates.  By the same reasoning, as we consider different $k'\in \{0,\cdots,N\}$, the coordinates $q(S')$ with $|S'|=k'$ will be linearly independent from the coordinates $q(S)$ with $|S|=
k<k'$.  In total, we obtain $\sum_{k=0}^N{{N}\choose{k}}=2^N$ linearly independent coordinates.\par
We can thus represent each point $\mbf{p}_{B|\bA}\in \mc{C}_{N,K}$ in the interference coordinates.  In doing so, we use the above observation that $I_{k}^{(S)}(\mbf{p}_{B|\bA})=0$ for all $k=|S|\ge K+1$. This indicates that $\mbf{p}_{B|\bA}$ has coordinates $q(S)=0$ for all $|S|\ge K+1$, and therefore the dimension of $(N,K)$-local MAC $\mc{C}_{N,K}$ can be upper bounded by:
\begin{equation}
\dim\mc{C}_{N,K}\leq \sum_{k=0}^K \binom{N}{k}.
\end{equation}

To compute a lower bound on $\dim\mc{C}_{N,K}$, we are going to form a family of affinely independent points.  Consider first the constant MAC that has transition probabilities
\begin{align}
\label{Eq:constant-1}
    p(0|a_1 \cdots a_N)&=0\qquad \forall a_1,\cdots,a_N
\end{align}
which has interference coordinates $q(S)=0$ for all $S\subseteq\{1,\cdots,N\}$.  Next, for each $S\subseteq \{1,\cdots, N\}$ with $1\leq |S|\leq K$, consider the $(N,K)$-local MAC $\mbf{p}_{B|\bA}^{(S)}$ defined by transition probabilities
\begin{align}
\label{Eq: affine points}
    p(0|a_1 \cdots a_N)=g_S(0|(a_s)_{s\in S})=\prod_{s\in S}\delta_{1,a_s}.
\end{align} 
When $S=\emptyset$, we define $\mbf{p}_{B|\bA}^{(\emptyset)}$ as the constant MAC with 
\begin{equation}
\label{Eq:constant-0} 
       p(0|a_1 \cdots a_N)=1\qquad \forall a_1,\cdots,a_N.
\end{equation}
It is straightforward to verify that $\mbf{p}_{B|\bA}^{(S)}$ has interference coordinates $q(S')=1$ if $S'=S$ and $q(S')=0$ if $S'\not=S$.  In other words, in interference coordinates, the MACs we have constructed in Eqns. \eqref{Eq: affine points} and \eqref{Eq:constant-0} correspond precisely to the standard unit vectors, while the MAC in Eq. \eqref{Eq:constant-1} corresponds to the all zero vector.  Clearly, this generates a set of $\sum_{k=0}^K\binom{N}{k}+1$ affinely independent points belonging to $\mc{C}_{N,K}$.  In summary, the dimension of the $\mc{C}_{N,K}$ polytope is $\sum_{k=0}^K\binom{N}{k}$ for any $1\le K\le N$.\par 
Note that we can use the same coordinates and affinely independent points (Eqs. \eqref{eq:full coordinates} and \eqref{Eq: affine points}) to arrive at $\dim\mf{I}_{N,K}=\sum_{k=0}^K\binom{N}{k}$.
\end{proof}

\begin{remark}
For $K=1$, Theorem \ref{thm:Rank Ck} says that the dimension of $\mc{C}_{N}([2]^N;[2])$ is $N+1$ and the interference coordinates we introduced in Eq.~(\ref{eq:full coordinates}) are:
\begin{align}
\label{Eq:I2-coords}
\{p(b|\mbf{e}_0), p(b|\mbf{e}_1)-p(b|\mbf{e}_0),\cdots, p(b|\mbf{e}_N)-p(b|\mbf{e}_0)\}.
\end{align}
with $p(b|\mbf{e}_i)=p(b|0,\cdots,1_i,\cdots,0)$ for $i\ne 0$ and $p(b|\mbf{e}_0)=p(b|0,\cdots,0)$.  In this case, it is more natural to work in the transition coordinates
\begin{align}
\label{Eq:I2-coords-transition}
\{p(b|\mbf{e}_0), p(b|\mbf{e}_1),\cdots, p(b|\mbf{e}_N)\}.
\end{align}
For any $L\in\{1,\cdots,N\}$ and any permutation on the input set, one can show that
\begin{align}
\label{Eq:I2-prob-relationships}
    p(b|\overbrace{1\cdots1}^{L},\overbrace{0\cdots0}^{N-L})=-(L-1)p(b|\mbf{e}_0)+\sum_{i=1}^L p(b|\mbf{e}_i)
\end{align}
for every $b\in\mc{B}$.  Hence according to Theorem \ref{Thm:classical-I2-N}, the only constraints on these coordinates for them to define a MAC in $\mc{C}_{N}([2]^N;[2])$ is positivity:
\begin{align}
0\leq -(|S|-1)p(b|\mbf{e}_0)+\sum_{s\in S} p(b|\mbf{e}_s)\leq 1
\end{align}
for all $S\subseteq\{1,\cdots,N\}$.

The proof of Eq.~\eqref{Eq:I2-prob-relationships} follows by induction on $N$.  Clearly it is true when $N=1$, so suppose it holds for arbitrary $N\geq 1$.  Consider an $(N+1)$-sender MAC $\mbf{p}_{B|A_1,\cdots A_{N+1}}\in\mf{I}_{N+1}([2]^{N+1};\mc{B})$.  If we fix the input value of any party to be zero, then what remains is an $N$-sender MAC.  Hence by inductive assumption, for any $L=1,\cdots,N$ and any permutation, it holds that
\begin{align}
\label{Eq:lem-induction}
    p(b|\overbrace{1,\cdots,1,}^{L}\overbrace{0,\cdots,0}^{N+1-L})&=-(L-1)p(b|\mbf{e}_0)+\sum_{i=1}^L p(b|\mbf{e}_i).
\end{align}
It just remains to prove the case when $L=N+1$.  Since $\mbf{p}_{B|A_1,\cdots,A_{N+1}}\in\mf{I}_{N+1}([2]^{N+1};\mc{B})$, we have
\begin{align}
    p(b|\overbrace{1\cdots1}^{N+1})=&p(b|\overbrace{1\cdots1}^{N-1},0,1)+p(b|\overbrace{1\cdots1}^{N-1},1,0)\notag\\
    &-p(b|\overbrace{1\cdots1}^{N-1},0,0)\notag\\
    =&-(L-1)p(b|\mbf{e}_0)+\sum_{i=1}^{N+1}p(b|\mbf{e}_i),
\end{align}
where the second equation follows from Eq.~\eqref{Eq:lem-induction} and its parties' permutation.  This proves the claim.

\end{remark}

\subsubsection{$(3,2)$-Local MACs}
The $\mathcal{C}_{3,2}$ polytope is formed by the set of correlations generated according to Eq.~\eqref{Eq:N,K-local-binary-out} with $N=3$ and $K=2$. From Theorem~\ref{thm:Rank Ck}, our analysis involves a seven-dimensional affine subspace with 38 vertices.  Upon running the PORTA software, we find there are 96 facet inequalities characterizing the polytope $\mathcal{C}_{3,2}$, with 16 of them being positivity constraints. By removing equivalent inequalities that can be obtained by relabelling the inputs and output, we find the following three nontrivial inequalities (see also \cite{Horvat-2019b}): 
    \begin{align}
    \label{eq:fp}
    &p(0|000)+p(1|001)+p(1|010)+p(1|100)\le3\\
    &p(0|000)+p(1|001)+p(1|010)+p(0|101)\le3 \label{Eq:ineq2}\\
    &p(0|000)+p(1|001)+p(1|010)+p(1|011)\notag\\
    &\textcolor{white}{p(0|000)+p(1|010)+p(1|011)}+p(0|111)\le4.\label{Eq:ineq3}
    \end{align}
Among these, only the first one is symmetric among the senders.  It is of particular interest in our investigation, and we will refer to it as the fingerprinting inequality in what follows. The reference to ``fingerprinting"  is due to the fact that Eq. \eqref{eq:fp} can be interpreted in terms of a task that loosely resembles the two-party fingerprinting task described in the introduction. {To see that, when relabeling all the inputs in Eq. \eqref{eq:fp}, we will obtain an inequality of the form:
$$p(0|111)+p(1|110)+p(1|101)+p(1|001)\le3.$$
If we assume the input bit of each party is given uniformly at random, Eqs. \eqref{eq:fp} and its relabled counterpart imply that the probability the decoder can decide whether $a_1=a_2=a_3$ is at most $3/4$, classically.}


Moving beyond $\mathcal{C}_{3,2}$ to general $\mathcal{C}_{N,K}$ polytopes can be done by working in a larger affine subspace.   However, it becomes increasingly more resource-intensive to perform the numerical analysis since the number of vertices increases exponentially.  Therefore, instead of trying to fully characterize the polytopes $\mc{C}_{N,K}$, we will consider just a few useful facets for each of them.\par

\begin{table*}
  \begin{tabular}{ l | c| c}
  \textrm{$N$-local MACs $\mc{C}_N{(\bcA,\mc{B})}$} &  \textrm{Inclusion relations} & \textrm{Propositions}\\
  \hline
    Arbitrary inputs and binary output & $\mc{C}_N=\mc{C}_N'=\conv[\mc{C}_N]=\mc{C}_{N}^{(\text{sep})}$ & Prop. \ref{Prop:binary-out}  \\ 
    Binary inputs and arbitrary output & $\mc{C}_N\subset^{\dagger} \mc{C}_N'=\conv[\mc{C}_N]=\mc{C}_{N}^{(\text{sep})}$ & Prop. \ref{Prop:non-convex}, Prop. \ref{Prop:binary-int}\\
     Arbitrary inputs and arbitrary output & $\mc{C}_N\subset^{\dagger} \mc{C}'_N\subset\conv[\mc{C}_N]\subset \mc{C}_{N}^{(\text{sep})}$ & Prop. \ref{Prop:non-convex2} \\
  \end{tabular}
  \caption{A comparison of $N$-party classical MACs and the corresponding propositions.  The inclusion $\subset$ is proper for all $N\geq 1$.  The inclusion $\subset^{\dagger}$ is proper only for $N>1$ whereas it becomes a trivial equivalence when $N=1$.}
  \label{tab:N-local-comparision}
\end{table*}
\subsubsection{The Generalized Fingerprinting Inequality}

\label{Sect:Generalize-fingerprinting}

To identify facet inequalities for higher-dimensional $\mathcal{C}_{N,K}$, we invoke the idea of ``lifting'' a Bell inequality, which has been previously used in the study of nonlocality \cite{Svetlichny-1987a}\cite{Pironio-2005a}.  The basic idea is to identify a facet inequality for $\mc{C}_{N-1,K}$ and then generalize its structure by the addition of one more party.  To verify that this constructed inequality indeed defines a valid facet of $\mc{C}_{N,K}$, it suffices to show to that it is satisfied by $\dim\mathcal{C}_{N,K}$ affinely independent points \cite{Pironio-2005a}.  

Having computed the dimension of $\dim\mathcal{C}_{N,K}$ in Section \ref{thm:Rank Ck}, we will specifically apply this procedure to lift inequality \eqref{eq:fp} to higher-dimensional polytopes. Let us consider the $(N,K)$-local generalized fingerprinting inequalities, which, up to a relabeling of inputs and permutation of parties, have the form
\begin{equation}
p(0|\overbrace{0\cdots0}^N)+\sum_{i=1 }^{K+1}p(1|\overbrace{0\cdots,1_i,\cdots0}^{K+1},\overbrace{0\cdots0}^{N-K-1})\le K+1.
\label{eq:fpg}
\end{equation}
Here, $1_i$ means that party $\msf{A}_i$ has input $1$. When $N=3$ and $K=2$, we obtain the fingerprinting inequality studied above; when $N=2$, and $K=1$, the corresponding inequality can be derived from $I_2$ equality Eq.~\ref{Eq:I2} along with a positivity constrains $p(0|1,1)\ge 0$. These inequalities has been previously studied in Ref. \cite{Horvat-2019a}, although its tightness as a facet inequality of $\mc{C}_{N,K}$ was not considered.  One can easily check that these inequalities cannot be violated by any correlation with the decomposition Eq.~(\ref{Eq:N,K-local-binary-out}) for a given $N$ and $K$. Furthermore based on Theorem~\ref{thm:Rank Ck}, the dimension of $\mathcal{C}_{N,K}$ is $\sum_{k=0}^K{N\choose{k}}$. Therefore, to show that Eq.~(\ref{eq:fpg}) is a valid facet of the $\mathcal{C}_{N,K}$ polytope, it is sufficient to find $\sum_{k=0}^K{N\choose{k}}$ affinely independent MACs that saturate it. \par

\begin{proposition}
The $(N,K)$-party generalized fingerprinting inequalities are valid inequalities of the $\mathcal{C}_{N,K}$ polytope.
\label{prop:gfp}
\end{proposition}
\begin{proof}

The list of affinely independent points saturating the inequality in Eq. \eqref{eq:fpg} is not unique, and here we construct one using MACs similar to those in the proof of Theorem \ref{thm:Rank Ck}.   Among the $\sum_{k=0}^K{N\choose{k}}+1$ affinely independent points constructed in Theorem \ref{thm:Rank Ck} (Eqns. \eqref{Eq:constant-1} -- \eqref{Eq:constant-0}), a total of $\sum_{k=0}^K{N\choose{k}}-K-1$ will saturate the inequality in Eq.~\eqref{eq:fpg}.  In particular, it will not be saturated by the constant MAC, with $p(0|a_1 \cdots a_N)=1$, as well as the MACs with $p(0|a_1 \cdots a_N)=\prod_{s\in S}\delta_{1,a_s}$ for $|S|=1$ and $S\subset\{1,\cdots,K+1\}$.  We replace $K+1$ of these points by MACs defined by
\begin{equation}
    p(0|a_1 \cdots a_N)=\prod_{\substack{j=1\\j\not=i}}^{K+1}\delta_{0,a_j},
\end{equation}
for each $i\in\{1,\cdots,K+1\}$.
Each of these is an $(N,K)$-local MAC since it only depends on the inputs of $K$ parties, and it can be readily seen to saturate the inequality in Eq.~\eqref{eq:fpg}.  They are also affinely independent from each other and from the $\sum_{k=0}^K{N\choose{k}}-K-1$ previous points since they have probability values $p(0|a_1 \cdots a_N)$ equaling $1$ when only one or fewer parties have an input equaling $1$.  In summary, we have identified $\sum_{k=0}^K{N\choose{k}}$ affinely independent points saturating Inequality \eqref{eq:fpg}. This inequality and all its permutations are thus valid facets of the $\mathcal{C}_{N,K}$ polytope.
\end{proof}

\subsection{Beyond Binary $(N,K)$-Local MACs}
\label{Sect:Beyond-binary}
We now draw some general conclusions about $(N,K)$-local MACs.
\subsubsection{Polytope Dimension and Generalized Fingerprinting Inequalities}
The dimension of $\mc{C}_{N,K}(\bcA,\mc{B})$ for general input/output sets can be computed using the same argumentation as in Section \ref{Sect:dim-binary-N,K}.  For simplicity we assume that $|\mc{A}_i|=|\mc{A}|$ for all parties $\msf{A}_i$.  Then one finds 
\begin{theorem}
\label{thm:Rank Ck general}
 \begin{equation}
 \label{eq:rank CN,K general}
 \dim\mathcal{C}_{N,K}(\bs{\mc{A}},\mc{B})=(|\mc{B}|-1)\sum_{k=0}^K(|\mc{A}|-1)^k\binom{N}{k}
 \end{equation}
\end{theorem}
While in general $\mc{C}_{N,K}(\bs{\mc{A}};\mc{B})$ is a proper subset of $\mf{I}_{N,K}(\bs{\mc{A}};\mc{B})$, the two sets have the same dimension:
\begin{align}
&\dim \mf{I}_{N,K}(\bs{\mc{A}};\mc{B})=\dim{\mc{C}_{N,K}(\bs{\mc{A}};\mc{B})}=\dim{\mc{C}'_{N,K}(\bs{\mc{A}};\mc{B})}\notag\\&=\dim{\conv[\mc{C}_{N,K}(\bs{\mc{A}};\mc{B})]}=\dim{\mc{C}_{N,K}^{(\text{sep})}(\bs{\mc{A}};\mc{B})}  .
\end{align}
Details on its proof is provided in Appendix~\ref{sec:A1 separation}.  We can likewise extend the generalized fingerprinting inequality to more outputs.  That is, we have
\begin{equation}
-p(b|\overbrace{0,\cdots,0}^N)+\sum_{i=1 }^{K+1}p(b|\overbrace{0,\cdots,1_i,\cdots,0}^{K+1},\overbrace{0,\cdots,0}^{N-K-1})\le K,
\label{eq:arbitary output fpg}
\end{equation}
for all $b\in\mc{B}$ and all other inequalities obtained by relabeling inputs and permuting parties.  These are valid facets for $\mc{C}_{N,K}^{(\text{sep})}(\bs{\mc{A}};\mc{B})$, as can be shown by following the procedure taken in Proposition \ref{prop:gfp} .
\begin{table*}
  \begin{tabular}{ l | c| c}
  \textrm{$(N,K)$-local MACs $\mc{C}_{N,K}{(\bcA,\mc{B})}$} &  \textrm{Inclusion relations} & \textrm{Propositions}\\
  \hline
$|\mc{B}|=2$& $\mc{C}_{N,K}=\mc{C}_{N,K}'=\conv[\mc{C}_{N,K}]=\mc{C}_{N,K}^{(\text{sep})}$ & Prop. \ref{Prop:binary-out} \\ 
     $K+1\ge|\mc{B}|>2$ & $\mc{C}_{N,K}\subset^{\dagger}\mc{C}_{N,K}'=\conv[\mc{C}_{N,K}]=\mc{C}_{N,K}^{(\text{sep})}$ & Prop. \ref{Prop:non-convex3}, Prop. \ref{prop:N,K convexity}\\
     $|\mc{B}|>K+1$ & $\mc{C}_{N,K}\subset^{\dagger}\mc{C}_{N,K}'\subset\conv[\mc{C}_{N,K}]\subset\mc{C}_{N,K}^{(\text{sep})}$ & Prop. \ref{Prop:non-convex2} 
  \end{tabular}
  \caption{A comparison of $(N,K)$-party classical MACs for $K\ge 2$.  The inclusion notation is the same as in Table \ref{tab:N-local-comparision}.}
  \label{tab:N,K macs}
\end{table*}
\subsubsection{Non-Convexity and Separating Different Classes of $(N,K)$-Local MACs}
In Section \ref{Sect:SR-model}, we introduced four classes of MACs,  $\mc{C}_{N,K}(\bcA,\mc{B})$, $\mc{C}'_{N,K}(\bcA,\mc{B})$,  $\conv[\mc{C}_{N,K}(\bcA,\mc{B})]$, and $\linebreak\mc{C}^{(\text{sep})}_{N,K}(\bcA,\mc{B})$.  We have seen in Proposition \ref{Prop:binary-out} that they are all equivalent when $|\mc{B}|=2$.  This is a special case, however, as shown in Proposition~\ref{Prop:non-convex}(whose proof is given in Appendix~\ref{sec:A1 separation}) they are not equivalent or even convex in general. The inclusion and convexity structure will be discussed in this section.  The results are summarized in Table \ref{tab:N-local-comparision} for $K=1$ and Table \ref{tab:N,K macs} for $K>1$ (with addition information given in Appendix~\ref{sec:A1 separation}).  Here we assume the cardinality of the input sets $\mc{A}_i$ and the output $\mc{B}$ are all greater than one.  
\par

\begin{proposition}
\label{Prop:non-convex}
For $|\mc{B}|>2$ and $N\geq 2$, the set $\mc{C}_{N}(\bcA,\mc{B})$ is non-convex; hence $\mc{C}_{N}(\bcA,\mc{B})\not=\mc{C}'_{N}(\bcA,\mc{B})$.
\end{proposition}
\begin{remark}

The convexity of $\mc{C}_N(\bcA;[2])$ is established in Proposition \ref{Prop:binary-out} essentially by absorbing the decoding function $d:\cup_{i=1}^N\mc{M}_i\to \{0,1\}$ into the encoding functions $q_i:\mc{A}_i\to\mc{M}_i$, where $\mc{M}_i=\{0,\mbf{e}_i\}$.  The is possible because $|\mc{B}|\leq|\mc{M}_i|$ in the binary-output case.  The non-convexity shown in Proposition \ref{Prop:non-convex} follows from the fact that $|\mc{B}|=3$ and so $|\mc{B}|>|\mc{M}_i|$; hence it is not possible to absorb the decoder into the encoders.

\end{remark}

While $\mc{C}_N([2]^N;\mc{B})$ is non-convex for $|\mc{B}|>2$, if shared randomness between the particle source and receiver is allowed, then convexity can be restored by absorbing the encoding function into the decoding function.  
\begin{proposition}
\label{Prop:binary-int}
 $\mc{C}'_{N}([2]^N;\mc{B})=\mc{C}_{N}^{(\text{sep})}([2]^N;\mc{B})$ for arbitrary output set $\mc{B}$.
\end{proposition}
\begin{proof}
For $p(b|\mbf{a})=\sum_{i=1}^Np_ig_i(b|a_i)$ we define decoders with $d_i(b|0)=g_i(b|0)$ and $d_i(b|\mbf{e}_i)=g_i(b|1)$, as well as deterministic encoders $q_i(0|0)=1$ and $q_i(\mbf{e}_i|1)=1$.  With these choices, $\sum_{i=1}^Np_ig_i(b|a_i)$ has the form of Eq.~\eqref{Eq:2-N-source}.
\end{proof}
However, as shown in Proposition~\ref{Prop:non-convex2} and the following remark (whose proof is given in Appendix~\ref{sec:A1 separation}), this no longer holds when the input sets are not binary, i.e. $\mc{C}_N'(\bcA,\mc{B})$ with $|\mc{B}|>2$ and $|\mc{A}_i|> 2$ will no longer be convex.  
\begin{proposition}
\label{Prop:non-convex2}
For $|\mc{B}|>K+1$ and $|\mc{A}_i|>2$ for some party $\msf{A}_i$, the set $\mc{C}_{N,K}'(\bcA,\mc{B})$ is non-convex; hence $\mc{C}_{N,K}'(\bcA,\mc{B})\not=\conv[\mc{C}_{N,K}(\bcA,\mc{B})]$.
\end{proposition}
\begin{remark}
When $|\mc{B}|>K+1$ and $|\mc{A}_i|>2$, one can also have $\conv[\mc{C}_{N,K}(\bcA,\mc{B})]\ne \mc{C}_{N,K}^{(\text{sep})}(\bcA,\mc{B})$\cite{supp}
\end{remark}
Similarly propositions~\ref{Prop:non-convex3} and \ref{prop:N,K convexity} discussing the relation specifically between different $(N,K)$-local MACs with $K\ge2$ can be found below while their proofs and other discussions are also given in detail in Appendix~\ref{sec:A1 separation} and the results are summarized in Table \ref{tab:N,K macs}.
\begin{proposition}
\label{Prop:non-convex3}
For $|\mc{B}|>2$, the set $\mc{C}_{N,K}(\bcA,\mc{B})$ is non-convex; hence $\mc{C}_{N,K}(\bcA,\mc{B})\not=\mc{C}'_{N,K}(\bcA,\mc{B})$.
\end{proposition}
 \begin{proposition}
\label{prop:N,K convexity}
$\mc{C}'_{N,K}(\bcA;\mc{B})=\mc{C}_{N,K}^{(\text{sep})}(\bcA;\mc{B})$ for arbitrary input sets $\bs{\mc{A}}$ and output set $|\mc{B}|\le K+1$.
\end{proposition}
\section{Quantum MACs}
\label{Sect:quantum MACs}
\label{Sect:N-local-quantum}
\subsection{Quantum Violation of $N$-Local Classical MAC}
According to Theorem \ref{Thm:classical-I2-N}, an $N$-partite binary MAC can be generated by a single classical particle if and only if the interference term $I_2$ vanishes for every pair of parties.  In contrast, quantum mechanics allows for single particles to demonstrate nonzero interference, a fact demonstrated most conspicuously in a double-slit experiment \cite{Grangier-1986,Jacques-2005,Merzbacher-1998a}.  Let us here describe a simplified version of this effect.  Suppose that a particle is prepared in a superposition of two paths $\ket{\psi}=x\ket{0}_{\msf{A}_1}\ket{1}_{\msf{A}_2}+y\ket{1}_{\msf{A}_1}\ket{0}_{\msf{A}_2}$.  The particle is then subjected to interference described by the path transformation
\begin{align}
\frac{1}{\sqrt{2}}(\ket{0}_{\msf{A}_1}\ket{1}_{\msf{A}_2}+\ket{1}_{\msf{A}_1}\ket{0}_{\msf{A}_2})&\to\ket{1}_{\msf{B}_1}\ket{0}_{\msf{B}_2}\notag\\
\frac{1}{\sqrt{2}}(\ket{0}_{\msf{A}_1}\ket{1}_{\msf{A}_2}-\ket{1}_{\msf{A}_1}\ket{0}_{\msf{A}_2})&\to\ket{0}_{\msf{B}_1}\ket{1}_{\msf{B}_2}.
\end{align}
In the end, if the particle is found in path $\msf{B}_1$ we say this is ``outcome $0$'' and if it is found in path $\msf{B}_2$ it is ``outcome $1$.''  According to Born's Rule, the probability of outcome $0$ is 
\begin{equation}
\label{Eq:Born-I2}
\frac{1}{2}|x+y|^2=\frac{1}{2}|x|^2+\frac{1}{2}|y|^2+ Re(xy^*).
\end{equation}
The term $\frac{1}{2}|x|^2$ we can recognize as the probability of outcome $0$ if path $\msf{A}_1$ is blocked prior to the interference, and similarly, $\frac{1}{2}|y|^2$ is the probability of outcome $0$ if path $\msf{A}_2$ blocked.  If we let $0/1$ denote the event of blocked/unblocked path, then we can rewrite Eq.~\eqref{Eq:Born-I2} as
\begin{align}
Re(xy^*)=p(0|1,1)-p(0|0,1)-p(0|1,0).
\end{align}
Since $p(0|0,0)=0$, as both paths are being blocked, we can add this to the right-hand side.  A comparison with Eq.~\eqref{Eq:I2} shows that $I_2=|Re(xy^*)|$ in this simple path-blocking interference setup.  It is referred to as second-order interference since it is quadratic in the wave-function amplitudes, and it is nonzero if and only if both amplitudes are nonzero.

The setup just described falls under the communication scenario of Fig. \ref{Fig:N-Local}.  Namely the ``path blocking'' action of $\msf{A}_1$ or $\msf{A}_2$ is a valid NPE operation, and thus the fact that $I_2=Re(xy^*)$ implies $\mc{C}_2$ is a strict subset of $\mc{Q}_2$.  Even stronger, every non-classical pure state (i.e. one with a nonzero superposition of paths) is capable of generating a MAC outside the set $\mc{C}_2$.  A natural question is whether the same holds for mixed states, and whether an even higher value of $I_2$ can be achieved by using an encoding scheme other than path blocking.  To answer this question, let us define for an arbitrary state $\rho\in\mc{H}^{\msf{A}_1\cdots\msf{A}_N}_1$ the quantities
\begin{equation}
I_{2}^{(i,j)}(\rho):=\max\{I^{(i,j)}_2(\mbf{p}_{B|\bA})\;:\;\mbf{p}_{B|\bA}\in\mc{Q}_N(\rho)\}. 
\end{equation}
In other words, $I_{2}^{(i,j)}(\rho)$ is the largest value of $I_2$ that can be achieved between parties $i$ and $j$ using state $\rho$ and NPE encoding.  The following {proposition} shows these values directly correspond to the off-diagonal matrix elements of $\rho$.\begin{proposition}
\label{Prop:N-Local-quantum-I2}
For $\rho\in\mc{H}^{\bsA}_1$ let $\rho_{ij}=\bra{\mbf{e}_i}\rho\ket{\mbf{e}_j}$.  Then  $I_{2}^{(i,j)}(\rho)=4|\rho_{ij}|$.  Moreover, this value can be attained by parties $i$ and $j$ performing a $\{0,\pi\}$ encoding, i.e. $\mc{E}_{a}(X)=\sigma_z^a(X)\sigma_z^a$ for $a\in\{0,1\}$ and $\sigma_z=\op{0}{0}+e^{i\pi}\op{1}{1}$.
\end{proposition}
\begin{proof}
{As discussed in Eq.~\eqref{Eq:Kraus-qubit-MAC}, every NPE operation $\mc{E}_{a}^{\msf{A}_i}$ on qubit system $\msf{A}_i$ is characterized by Kraus operators $\left\{\left(\begin{smallmatrix}1&0\\0&y_{a_i}^{i}\end{smallmatrix}\right),\left(\begin{smallmatrix}0&z_{a_i}^{i}\\0&0\end{smallmatrix}\right)\right\}$ such that $|y_{a_i}^i|^2+|z_{a_i}^i|^2=1$.  Notice that
\begin{align}
\label{Eq:NP-output-operators}
\mc{E}_{a_i}^{\msf{A}_i}\otimes\mc{E}_{a_j}^{\msf{A}_j}(\op{00}{00})&=\op{00}{00}\notag\\
\mc{E}_{a_i}^{\msf{A}_i}\otimes\mc{E}_{a_j}^{\msf{A}_j}(\op{10}{10})&=\lambda_{a_i}^i\op{10}{10}+(1-\lambda_{a_i}^i)\op{00}{00}\notag\\
\mc{E}_{a_i}^{\msf{A}_i}\otimes\mc{E}_{a_j}^{\msf{A}_j}(\op{01}{01})&=\lambda_{a_j}^j\op{01}{01}+(1-\lambda_{a_j}^j)\op{00}{00}\notag\\
\mc{E}_{a_i}^{\msf{A}_i}\otimes\mc{E}_{a_j}^{\msf{A}_j}(\op{01}{10})&=\kappa_{a_i,a_j}\op{01}{10}  
\end{align}
in which $\kappa_{a_i,a_j}=y^j_{a_j}y^{i*}_{a_i}$ and $\lambda_a^i=|y^i_{a}|^2$.  If $\mc{E}^{\overline{\msf{A}_i\msf{A}_j}}$ is any local NPE map collectively performed by all parties other than $\msf{A}_i$ and $\msf{A}_j$, it can then be easily seen that}
\begin{align}
(\mc{E}_{0}^{\msf{A}_i}-\mc{E}_{1}^{\msf{A}_i})&\otimes (\mc{E}_{0}^{\msf{A}_j}-\mc{E}_{1}^{\msf{A}_j})\otimes \mc{E}^{\overline{\msf{A}_i\msf{A}_j}}(\rho)\notag \\
&= \kappa\rho_{ij}\op{\mbf{e}_i}{\mbf{e}_j} + \kappa^*\rho_{ji}^*\op{\mbf{e}_j}{\mbf{e}_i}, 
\end{align}
where
\begin{equation}
\label{eq:kappa}
\kappa=\kappa_{0,0}+\kappa_{1,1}-\kappa_{0,1}-\kappa_{1,0}.
\end{equation}
For any decoding POVM $\{\Pi_0,\Pi_1\}$ we have
\begin{align}
I_{2}^{(i,j)}(\rho)
&= \tr\left[\Pi_0 \left(\kappa\rho_{ij}\op{\mbf{e}_i}{\mbf{e}_j} + \kappa^*\rho_{ji}^*\op{\mbf{e}_j}{\mbf{e}_i}\right) \right]\notag \\
&\leq |\kappa\rho_{ij}|\leq 4|\rho_{ij}|,
\end{align}
where the last inequality follows from the triangle inequality and the fact that $\Pi_0\leq \mbb{I}$. The inequalities are saturated when choosing $\{0,\pi\}$ phase encoding and $\Pi_0 = \ket{\psi_{ij}}\bra{\psi_{ij}}$, where $\displaystyle\ket{\psi_{ij}} = \frac{1}{\sqrt{2}}(\ket{\mbf{e}_i}+\frac{\rho_{ij}^*}{|\rho_{ij}|}\ket{\mbf{e}_j})$.
\end{proof}

\subsection{Quantum Violation of $(N,K)$-Local Classical MACs}
\label{sec:quantum violation}

For $K>1$, the polytope $\mc{C}_{N,K}$ is constrained by the vanishing of all $I_{K+1}$ expressions (see Eq.~\eqref{eq:IK equality}) as well as other non-trivial facets.  We just showed that quantum states can violate the $I_2$ equality using NPE operations.  However, for $K>1$ this is no longer possible.  To see this explicitly, consider the $I_3$ expression for parties $S=\{i,j,k\}$.  Using the notation of Proposition \ref{Prop:N-Local-quantum-I2}, for any $\rho\in\mc{H}_1^{\bsA}$ the relevant encoding for $I_3$ is 
\begin{align}
&(\mc{E}_{0}^{\msf{A}_k}-\mc{E}_{1}^{\msf{A}_k})\otimes(\mc{E}_{0}^{\msf{A}_i}-\mc{E}_{1}^{\msf{A}_i})\otimes (\mc{E}_{0}^{\msf{A}_j}-\mc{E}_{1}^{\msf{A}_j})\otimes \mc{E}^{\overline{\msf{A}_i\msf{A}_j\msf{A}_k}}(\rho) \notag\\&=(\mc{E}_{0}^{\msf{A}_k}-\mc{E}_{1}^{\msf{A}_k})( \kappa\rho_{ij}^*\op{\mbf{e}_i}{\mbf{e}_j} + \kappa^*\rho_{ji}\op{\mbf{e}_j}{\mbf{e}_i})=0, 
\label{eq:vanishing-i3}
\end{align}
where $\kappa$ is defined in Eq.~\eqref{eq:kappa}.  The second equality holds because $\mc{E}_{a_k}^{\msf{A}_k}$ is an NPE operation and so $\mc{E}_{a_k}^{\msf{A}_k}(\op{\mbf{e}_i}{\mbf{e}_j})=\op{\mbf{e}_i}{\mbf{e}_j}$.  One can easily extend this argument to show that $I_{K+1}=0$ for all $K>1$ when the quantum encodings are required to be $N$-local, NPE operations.  

Ultimately, the vanishing of $I_3$ for quantum MACs is a consequence of Born's Rule \cite{Sorkin-1994a,Daki-2014}, and it has been experimentally verified \cite{Sinha-2010a}.  Recently, it was shown by Rozema \textit{et al.} that $I_3$ and all higher-order interference must vanish in any multi-path interferometer experiment using path-blocking operations \cite{Rozema-2020}.  Equation \eqref{eq:vanishing-i3} offers a slight generalization in that the encoding maps $\mc{E}_{a_k}^{\msf{A}_k}$ need not just be a path-blocking operation but rather any NP-operation.  Nevertheless, the spirit of the trace argument presented in Ref. \cite{Rozema-2020} is implicit here since Eq. \eqref{eq:vanishing-i3} is a traceless operator.  


While quantum mechanics cannot violate the $I_{K+1}=0$ constraints of $\mc{C}_{N,K}$ for $K>1$, it is possible for quantum systems to violate some of the other facet inequalities.  For example, in the case of $(3,2)$-local MACs, all channels in $\mc{C}_{3,2}([2]^3;[2])$ must satisfy Eqs.~\eqref{eq:fp}--\eqref{Eq:ineq3}, namely
\begin{align*}
&p(0|000)+p(1|001)+p(1|010)+p(1|100)\le3\\
&p(0|000)+p(1|001)+p(1|010)+p(0|101)\le3\\
&p(0|000)+p(1|001)+p(1|010)+p(1|011)\\
&\textcolor{white}{p(0|000)+p(1|010)+p(1|011)}+p(0|111)\le4.
\end{align*}
However, by using the equal superposition state $\sqrt{1/3}(\ket{100}+\ket{010}+\ket{001})$, it is possible to generate a MAC in $\mc{Q}_{3,1}([2]^3;[2])$ for which the left-hand side attains a value of $3.66$, $3.15$ and $4.66$ correspondingly. \par

In this section we will generalize this example by demonstrating $\mc{Q}_{N,1}([2]^N;[2])\not\subset \mc{C}_{N,N-1}([2]^N;[2])$ for every $N$.  In other words, $N$ quantum parties restricted to local encodings can generate a MAC that cannot be simulated using $N$ classical parties with the locality restriction relaxed on all but one party.  This result was previously shown in Ref. \cite{Horvat-2019a} and given the interpretation that quantum mechanics allows for a greater ``information speed'' with respect to communicating information shared among spatially separated senders.  As in Ref. \cite{Horvat-2019a}, we show this result by finding a violation of the $(N,N-1)$-party generalized fingerprinting inequality,
\begin{equation}
p(0|\overbrace{0,\cdots,0}^N)+\sum_{i=1 }^Np(1|\overbrace{0,\cdots,1_i,\cdots,0}^N)\le N,
\label{eq:gfp}
\end{equation}
which we have proven to be a valid facet of $\mathcal{C}_{N,N-1}$ in Proposition \ref{prop:gfp}.  Here we provide a slightly improved quantum strategy than the one presented in Ref. \cite{Horvat-2019a}, and we are thus able to obtain a larger violation in many cases.  

To show the quantum violation, we suppose the quantum particle is prepared in an equal superposition among $N$ paths, that is, $\ket{\Psi_N}=\frac{1}{\sqrt{N}}\sum_{i=1}^N\ket{\mbf{e}_i}$ with $\ket{\mbf{e}_i}=\ket{0}_{A_1}\cdots\ket{1}_{A_i}\cdots\ket{0}_{A_N}$. Consider the local phase encoding map for party $\msf{A}_i$ characterized by two angles $(\theta_i,\phi_i)$: $\mc{E}^{\msf{A}_i}_{0}(\rho)=U(\theta_i)\rho U(-\theta_i)$ and $\mc{E}^{\msf{A}_i}_{1}(\rho)=U(\phi_i+\theta_i)\rho U(-\phi_i-\theta_i)$, where $U(\theta)=\op{0}{0}+e^{i\theta}\op{1}{1}$.  With this encoding, Eq.~\eqref{eq:gfp} becomes
\begin{equation}
\tr[\Pi_0(\sigma_{0\cdots0})]+\sum_{i=1}^N\tr[\Pi_1(\sigma_{0\cdots1_i\cdots0})]\le N,
\label{eq:Helstrom}
\end{equation}
where $\sigma_{a_1\cdots a_N}=\mc{E}^{\msf{A}_1}_{a_1}\otimes\cdots\otimes \mc{E}^{\msf{A}_N}_{a_N}(\op{\Psi_N}{\Psi_N})$ and $\{\Pi_0,\Pi_1\}$ is the decoding POVM. From Helstrom's Theorem \cite{helstrom-1969}, the maximum value of the left-hand side is $\frac{1}{2}(N+1+\norm{M}_1)$, where $\norm{M}_1=\tr{\sqrt{M^{\dagger}M}}$ and
\begin{equation}
\label{Eq: M}
M=\sum_{i=1}^N\sigma_{0 \cdots1_i\cdots0}-\sigma_{0\cdots0}.
\end{equation}
A violation of the generalized fingerprinting inequality can then be cast as the condition $\delta>0$ in which
\begin{equation}
\label{Eq: Violation}
\delta=\frac{1}{2}\left(\norm{M}_1-N+1\right).
\end{equation} 
Notice that $\mc{E}^{\msf{A}_1}_{0}\otimes\cdots\otimes \mc{E}^{\msf{A}_N}_{0}$ is a global unitary, and since the trace norm is invariant under a global unitary, it can be factored out.  Therefore, $M$ can be simply characterized as an $N\times N$ matrix of the form
\begin{widetext}
\begin{align}
M=\frac{1}{N}\begin{pmatrix}
N-1&N-3+e^{i\phi_1}+e^{-i\phi_2}&\cdots&N-3+e^{i\phi_1}+e^{-i\phi_N}\\
N-3+e^{i\phi_2}+e^{-i\phi_1}&N-1&\cdots&N-3+e^{i\phi_2}+e^{-i\phi_N}\\
\vdots&\vdots&\vdots&\vdots\\
N-3+e^{i\phi_N}+e^{-i\phi_1}& N-3+e^{i\phi_N}+e^{-i\phi_2}&\cdots&N-1
\end{pmatrix}.
\label{eq:violation}
\end{align}
\end{widetext}
\par
This matrix depends on $N$ parameters and there appears to be no easy way to compute its trace norm.  Nevertheless, after performing exhaustive numerical searches for small $N$, we find that a maximal violation can be obtained with some simple encoding strategies (see Table \ref{tab:violation}).  This strategy consists of setting $\phi_1=\phi$, $\phi_2=-\phi$, $\phi_i=\pi,~\forall i\ne 1,2$. In this case $M$ will be simplified as:
\begin{equation}
M=\frac{1}{N}(4\mathbb{I}+(N-3)\op{\Psi_N}{\Psi_N}-\op{\Psi_N^{'}}{\Psi_N^{'}}-\op{\Psi_N^{''}}{\Psi_N^{''}}),
\label{eq: M}
\end{equation}
where $\ket{\Psi_N^{'}}=\frac{1}{\sqrt{N}}\left(e^{i\phi}\ket{\mbf{e}_1}+\sum_{i\ne 1}^N\ket{\mbf{e}_i}\right)$ and  $\ket{\Psi_N^{''}}=\frac{1}{\sqrt{N}}\left(\ket{\mbf{e}_2}+e^{i\phi}\sum_{i\ne 2}^N\ket{\mbf{e}_i}\right)$. With these simplifications, we can set a new basis $\ket{\mbf{e}_1'}=\ket{\mbf{e}_1},\ket{\mbf{e}_2'}=\ket{\mbf{e}_2}$, and $\ket{\mbf{e}_3'}=\frac{1}{\sqrt{N-2}}\sum_{i=3}^N\ket{\mbf{e}_i}$ (for $N=2$, there is no need for the third basis). Hence the matrix $M$ can be expressed as:
\begin{equation}
M=M_3\oplus\frac{4}{N}\bar{\mbb{I}},
\end{equation}
with $M_3$ written in the subspace $\{\ket{\mbf{e}'_1},\ket{\mbf{e}'_2},\ket{\mbf{e}'_3}\}$ and $\bar{\mbb{I}}$ is identity in the orthogonal subspace:
\begin{widetext}
\begin{align}
M_3=\frac{1}{N}\begin{pmatrix}
N-1&N-3+2e^{i\phi}&\sqrt{N-2}(N-4+e^{i\phi})\\
N-3+2e^{-i\phi}&N-1&\sqrt{N-2}(N-4+e^{-i\phi})\\
\sqrt{N-2}(N-4+e^{-i\phi})&\sqrt{N-2}(N-4+e^{i\phi})&(N-2)(N-5)+4
\end{pmatrix},
\label{eq: M'}
\end{align}
\end{widetext}

\begin{table}
    \begin{tabular}{c|c}
    Number of parties $N$ & Quantum violation $\delta$   \\
    \hline
    2 & 1  \\

    3 &0.6667 \\

    4 &0.1250 \\

    5 &0.0333\\

    6 &0.0139 \\
    \end{tabular}
 
    \caption{Quantum violation $\delta$ for different generalized fingerprinting inequality based on full numerical search over $N$ phase parameters.}
    \label{tab:violation}
\end{table}

Since $\tr{M}=N-1$, a necessary condition for $\delta>0$ is that $\lambda_{\min}<0$, where $\lambda_{\min}$ is the smallest eigenvalue $M_3$.  Moreover, from the Cauchy Interlace Theorem, $M_3$ will have at most one negative eigenvalue since its first and second leading principal minors are positive.  Therefore, the quantum violation can be expressed as
\begin{equation}
\delta=\max\{0,-\lambda_{\min}\}.
\end{equation} 
The smallest eigenvalue of a $3\times 3$ matrix is analytically calculable. Two useful criteria can be obtained as follows:\par 
(1) Quantum violation $\delta>0$ if and only if
\begin{alignat}{2}
        &\cos{\phi}\ne 1  &&\text{if}~N=2,3\\ &\cos{\phi}\ge\frac{(N-2)(N-3)-4}{(N-2)(N-3)}\quad &&\text{if}~N>3
\end{alignat}\par
(2) Maximal quantum violation $\delta$ is given as:
\begin{alignat}{2}
\label{Eq: maximal violation}
&\delta=\frac{2}{N} && \text{if}~N=2,3\\
&\delta=\frac{1}{N(N-2)(N-3)}\quad &&\text{if}~N>3
\label{Eq:Maximal violation}
\end{alignat}
which is achieved at $\phi=\pi$ for $N=2,3$ and $\cos\phi=\frac{53 - 105N + 71 N^2 - 20 N^3 + 2 N^4}{2(N-2)^2(N-3)^2}$ for $N>3$. Surprisingly, these violation saturate the maximal violation in Table \ref{tab:violation} where we were considering all phase encoding strategies.
\begin{remark}
The optimal encoding strategy for achieving the maximal quantum violation Eq.~\eqref{Eq:Maximal violation} is not unique.  We can get the same violation with other encoding strategies: e.g. for even $N$, $\phi_i=\phi~~\forall i\le\frac{N}{2}$ and  $\phi_i=-\phi$ otherwise \cite{Horvat-2019a}; for odd $N$,  $\phi_{N}=\pi$, $\phi_i=\phi~~\forall i\le\lfloor\frac{N}{2}\rfloor$ and $\phi_i=-\phi$ otherwise.
\end{remark}

\subsection{Beyond $N$-Local Quantum Channels}
\label{Sect:quantum-B>2}
 In the previous section, we showed that $N$-local quantum MACs have a much richer structure than $N$-local classical MACs, and they cannot even be simulated using $(N,K)$-local classical MACs for any $K<N$.  Here we consider other structural properties of quantum MACs. Similar to our consideration of $(N,K)$-local classical MACs, we can relax the locality constraint in $N$-local quantum MACs by allowing $K$ parties to collaborate and perform a joint NPE encoding across all their subsystems.  For some collection of parties $S$, we will write $\mc{E}^{(\msf{A}_s)_s}_{(a_s)_s}$ to denote the joint map performed on systems $\bigotimes_{s\in S}\mc{H}^{\msf{A}_s}$ for collective input $(a_s)_s\in \bigtimes_{s\in S}\mc{A}_s$.  Channels that are built in this fashion will be called $(N,K)$-local quantum MACs, denoted as $\mc{Q}_{N,K}$. We will first show that MACs in $\mc{Q}_{N,K}$ can demonstrate higher-order intereference; more specifically, the $I_{2K}$ equality can be violated. We then calculate the dimension of binary $(N,K)$-local quantum MACs $\mathcal{Q}_{N,K}([2]^N;[2])$ and generalize this result to arbitrary inputs and output.
\begin{proposition}
\label{prop:2K-coherence-if-nonzero-off-diag}
For $K\le \lfloor\frac{N}{2}\rfloor$, a state $\rho\in\mc{H}_1^{\bsA}$ can produce a quantum MAC in $\mc{Q}_{N,K}$ that exhibits non-vanishing $2K$-order interference if and only if it has some off-diagonal term $\rho_{ij}\ne 0$ .
\end{proposition}
\par
\begin{proof}
Clearly, for an incoherent state $\rho$ with $\rho_{ij}=0$ for all $i\ne j$, the state is incoherent and it cannot even violate any $I_{K+1}$ equality. Conversely, suppose that $\rho_{ij}\not=0$.  To get a quantum violation for the $I_{2K}$ equality, let us  consider two $K$-collaborating groups $S_i$ and $S_j$ with $i\in S_i$, $j\in S_j$, and $S_i\cap S_j=\emptyset$. Each group performs a $\{0,\pi\}$ phase encoding on path $i$ and $j$ conditioned on the parity of their inputs in $S_i$ and $S_j$ respectively, i.e. $\mc{E}^{(\msf{A}_s)_{s\in S_i}}_{(a_s)_{s\in S_i}}(X)=\bigotimes_{s\in S_i\setminus \{i\}}\text{id}^{\msf{A}_s}\otimes (\sigma_z^{\oplus_{s\in S}a_s})^{\msf{A}_i}(X)(\sigma_z^{\oplus_{s\in S}a_s})^{\msf{A}_i}$ where $\sigma_z=\op{0}{0}+e^{i\pi}\op{1}{1}$. The decoder then measures with POVM $\left\{\Pi_0=\ket{\psi_{ij}}\bra{\psi_{ij}}^{\bsA},\;\Pi_1=\mbb{I}-\Pi_0\right\}$, where $\displaystyle\ket{\psi_{ij}} = \frac{1}{\sqrt{2}}\left(\ket{\mbf{e}_i}+\frac{\rho_{ij}^*}{|\rho_{ij}|}\ket{\mbf{e}_j}\right)$. With this scheme, we can achieve $I^{(S_i\cup S_j)}_{2K}=2^{2K}|\rho_{ij}|$.
\end{proof}
\begin{remark}
Similarly, for any $k\le 2K$, we can obtain nonzero $k^{\text{th}}$-order interference $I^{(S_i\cup S_j)}_{k}(\rho)=2^{k}|\rho_{ij}|$ by choosing disjoint $S_i,S_j\in \{1,\cdots,N\}$ such that $|S_i|+|S_j|=k$. In fact, for each $k\le 2K$, we can achieve maximal $k^{\text{th}}$-order interference $I^{(S_i\cup S_j)}_{k}(\rho)=2^{k-1}$ by a maximally coherence state on parties $\msf{A}_i$ and $\msf{A_j}$ with $|\rho_{ij}|=1/2$. However, for a general quantum state $\rho^{\bsA}$ shared between $N$ senders, it remains open what strategy extracts maximal higher-order interference.
\end{remark}

In Theorem \ref{thm:Rank Ck}, we demonstrated that $\mc{C}_{N,K}$ violate $I_K$ equality but has vanishing $(K+1)$-order interference, thus has at most $K$-order interference, analogously, we will show here that $I_{2K+1}$ can not be violated by any quantum MAC in $\mc{Q}_{N,K}$.
\begin{proposition}
\label{prop:QN,K 2K+1}
A quantum MAC in $\mc{Q}_{N,K}$ can have at most $2K$-order interference. 
\end{proposition}
\begin{proof}
As discussed in Section \ref{sec:quantum violation}, MACs in $\mc{Q}_N$ can exhibit second order interference, whereas they never violate the $I_3$ equality. Now, we will to generalize this result and demonstrate that any given MAC $\mbf{p}_{B|\bA}\in \mc{Q}_{N,K}$ always satisfies $I_{2K+1}=0$. For a general quantum state $\rho=\sum_{ij}\rho_{ij}\op{\mbf{e}_i}{\mbf{e}_j} \in \mc{H}^{\bsA}$, we have
\begin{align}
\sigma_{a_1\cdots a_N}&=\mc{E}^{(\msf{A}_s)_{s\in S_1}}_{(a_s)_{s\in S_1}}\otimes\cdots\otimes\mc{E}^{(\msf{A}_s)_{s\in S_G}}_{(a_s)_{s\in S_G}}[\rho]\notag\\&=\sum_{ij}\mc{E}^{(\msf{A}_s)_{s\in S_I}}_{(a_s)_{s\in S_I}}\otimes \mc{E}^{(\msf{A}_s)_{s\in S_J}}_{(a_s)_{s\in S_J}}[\rho_{ij}\op{\mbf{e}_i}{\mbf{e}_j}], 
\end{align}
where we have $G$ different groups with $|S_i|\le K$, $S_i\cap S_j=\emptyset$ for any $i,j\in \{1,\cdots,G\}$ ($i\ne j$) and $\bigcup_{i=1}^G S_i=\{1,\cdots,N\}$. To get the last equality, we assume $i\in S_I$ and $j\in S_J$, and thus only $\mc{E}^{(\msf{A}_s)_{s\in S_I}}_{(a_s)_{s\in S_I}}$ and $\mc{E}^{(\msf{A}_s)_{s\in S_J}}_{(a_s)_{s\in S_J}}$ act non-trivially on the operator $\rho_{ij}\op{\mbf{e}_i}{\mbf{e}_j}$ since these maps are NPE operations; i.e.  $\mc{E}^{(\msf{A}_s)_{s\in S_K}}_{(a_s)_{s\in S_K}}[\op{\mbf{e}_i}{\mbf{e}_j}]=\op{\mbf{e}_i}{\mbf{e}_j}$ for any $S_I\cap S_K=\emptyset$ and $S_
J\cap S_K=\emptyset$. Hence, the $I^{(S)}_{2K+1}(\rho)$ quantity for a given $S$ with $|S|=2K+1$ becomes
\begin{align}
I^{(S)}_{2K+1}(\rho)=\tr{[\Pi_0\sum_{\mathclap{\substack{a_s\in\{0,1\}\\ s\in S}}}(-1)^{\oplus_{i\in S}a_i}\sigma_{a_1\cdots a_N}]}=0, 
\end{align}
where the final equality holds because each map on $\rho_{ij}\op{\mbf{e}_i}{\mbf{e}_j}$ depends on at most $|S_I\cup S_J|\le 2K$ inputs, but the summation is over $2K+1$ parameters. 
\end{proof}
This observation is intimately connected to the conclusion drawn in \cite{Horvat-2020a}, where the authors noted that classical and quantum theories exhibit different levels of interference depending on the number of particles used in the experiment. Here we observe a similar result except in a different operational setting: a single particle is still used, but a different numbers of parties are allowed to collaborate.\par
Finally in this section, we study the dimension of the set $\mc{Q}_{N,K}$. For simplicity, here we restrict our attention to the binary channel and defer the proof of the more general $\mc{Q}_{N,K}(\bcA;\mc{B})$ to Proposition S4 in \cite{supp}.
\begin{proposition}\label{prop:Rank Qk}
$\dim \mathcal{Q}_{N,K}([2]^N;[2]) = \sum_{k=0}^{2K} \binom{N}{k}$ for $K\le \lfloor\frac{N}{2}\rfloor$, $\dim \mathcal{Q}_{N,K}([2]^N;[2]) = 2^N$ for $\lfloor\frac{N}{2}\rfloor<K\le N$
\end{proposition}
\begin{proof}
Let us first consider the case when $K\le\lfloor\frac{N}{2}\rfloor$. The set $\mathcal{Q}_{N,K}([2]^N;[2])$ lives in a $(2\cdot2^N)$-dimensional space, and there are $2^N$ normalization condition for each $\mbf{p}_{B|\bA}\in\mc{Q}_{N}([2]^N;[2])$. Additionally, Proposition \ref{prop:QN,K 2K+1} shows that any $\mbf{p}_{B|\bA} \in \mc{Q}_{N,K}$ is constrained to have vanishing $I_{2K+1}$ and any higher-order interference. This means that there are a total of  $2^N+\sum^{N}_{k=2K+1} \binom{N}{k}$ constraints. Hence, we have
\begin{equation}
\dim\mathcal{Q}_{N,K}([2]^N;[2]) \le \sum_{k=0}^{2K} \binom{N}{k}.
\end{equation}
To lower bound $\dim\mathcal{Q}_{N,K}([2]^N;[2])$, we use a subset of the interference coordinates, originally defined in Eq.~\eqref{eq:full coordinates} by keeping $k\in \{0,\cdots,2K\}$ only:
\begin{widetext}
\begin{align}
\label{eq:I_0-to-I_2K}
\bigcup_{k\in\{0,\cdots,2K\}}\bigcup_{\substack{S\subseteq\{1,\cdots,N\}\\|S|=k}}\left\{q(S)\coloneqq(-1)^k\sum_{\substack{a_s\in\{0,1\}\\s\in S}}\prod_{i\in S}(-1)^{a_i}p(0|a_1 \cdots a_N) \;\bigg|\;\text{$a_j=0$ for $j\not\in S$}\right\}.
\end{align} 
\end{widetext}
We consider two classes of points: \\
(i) Classical MACs: This is the set of $\sum_{k=0}^K{N\choose{k}}+1$ affinely independent points in $\mc{C}_{N,K}$ as we discussed in Theorem \ref{thm:Rank Ck} (Eq. \ref{Eq: affine points}). Since these MACs are in $\mc{C}_{N,K}$, they are also in $\mc{Q}_{N,K}$. \\
(ii) Quantum MACs: we consider quantum state $\ket{\psi}^{\msf{A}_i\msf{A}_j}=\frac{1}{\sqrt{2}}(\ket{\mbf{e}_i}+\ket{\mbf{e}_j})$ (with $|\rho_{ij}|=\frac{1}{2}$) shared between parties $i\in S_i$, $j\in S_j$ and disjoint groups $ S_i\cap S_j=\emptyset$,  then we can construct quantum MAC $\mbf{p}^{(S)}_{B|\bA}$ as in Proposition \ref{prop:2K-coherence-if-nonzero-off-diag} and the remark below. This quantum MAC can demonstrate maximal $k^{\text{th}}$-order interference among and only among senders $S=S_i\cup S_j$ where $|S|=k$. The above strategy applies for different groups $S$ with $K<|S|\le2K$. Clearly, there are in total $\sum_{k=K+1}^{2K}\binom{N}{k}$ such quantum MACs. \par
Points in class (i) do not violate $I_k$ equalities for any $k>K$ hence have $q(S)=0$ for any $|S|>K$, while each point considered in class (ii) violates one and only one $I_k$ equality for integer $k=|S|>K$. Specifically, it has $q(S)=2^{k-1}$ but $q(S')=0$ for any $S'\ne S$ and $|S'|\ge k$ (For any $S'$ we can find one element $s'\in S'$ but $s'\notin S$, and therefore $q(S')$ will vanish since the quantum MAC $\mbf{p}^{(S)}_{B|\bA}$ doesn't depend on input $a_{s'}$). Hence these points form a upper-triangular matrix and thus are affinely independent. e.g. the upper-traingular matrix for $\mc{Q}_{3,1}$ is:
\begin{equation}
\label{eq:Qexample}
\begin{blockarray}{c(rrrrrrrr)}
    q(\emptyset)~~&0&1&0&0&0&1&1&1\\
    q(\{1\})~~&0&0&1&0&0&-1&-1&0\\
    q(\{2\})~~&0&0&0&1&0&-1&0&-1\\
    q(\{3\})~~&0&0&0&0&1&0&-1&-1\\
    q(\{1,2\})~~&0&0&0&0&0&2&0&0\\
    q(\{1,3\})~~&0&0&0&0&0&0&2&0\\
    q(\{2,3\})~~&0&0&0&0&0&0&0&2\\
\end{blockarray},
\end{equation}
where the first five columns correspond to the classical MACs, and the sixth to final columns are $\mbf{p}^{(S)}_{B|\bA}$ introduced above with $|S|=2$.\par
In total then, we obtain a collection of $\sum_{k=0}^{2K}\binom{N}{k}+1$ affinely independent points belonging to $\mc{Q}_{N,K}$. Hence we have:
\begin{equation}
\dim{\mc{Q}_{N,K}}\ge \sum_{k=0}^{2K}\binom{N}{k}.
\end{equation}
\par
In summary, the dimension of the set $\mc{Q}_{N,K}$ is $\sum_{k=0}^{2K}\binom{N}{k}$ for any $1\le K\le \lfloor\frac{N}{2}\rfloor$. Following the same line of reasoning, we can also deduce that $\dim\mc{Q}_{N,K}=\sum_{k=0}^{N}\binom{N}{k}=2^N$ for $\lfloor\frac{N}{2}\rfloor<K\le N$
\end{proof}
\begin{proposition}\label{prop:Rank Qk general}
for $K\le \lfloor\frac{N}{2}\rfloor$ and $|\mc{A}_i|=|\mc{A}|$ , we have
$$\dim \mathcal{Q}_{N,K}(\bs{\mc{A}},\mc{B}) = (|\mc{B}|-1)\sum_{k=0}^{2K}(|\mc{A}|-1)^k\binom{N}{k}.$$
\end{proposition}
The proof is given in Appendix~\ref{sec:A2 dimension}

\begin{remark}
The dimension of $(N,K)$-local quantum MAC $\mc{Q}_{N,K}$ coincides with the dimension of $(N,2K)$-local classical MACs $\mc{C}_{N,2K}$. However, the two sets are distinct, that is, $\mc{Q}_{N,K}\ne \mc{C}_{N,2K}$ in general. For instance, for $N=3$ and $K=1$, the separation can be directly seen from the quantum violation of fingerprinting inequality in Section \ref{sec:quantum violation}. For the simplest case with $N=2$ and $K=1$ however, this separation is less obvious. We detail a proof for $\mc{Q}_{2,1}\ne\mc{C}_{2,2}$ in Appendix~\ref{sec:A3 seperation} and conclude it as the following proposition .
\end{remark}
\begin{proposition}
$\mc{Q}_{2,1}([2]^2;[2])\ne\mc{C}_{2,2}([2]^2;[2])$.
\end{proposition}

\section{Connections with Multi-Level Coherence }

\label{Sect:Multi-level-coherence}

\subsection{The Channel Discrimination Task and Multi-Level Coherence of Pure States}
In quantum resource theories, channel discrimination tasks have been used as a way to witness different types of resourceful states \cite{Takagi-2019}. Here we would now like to relate our framework to the notion of multi-level coherence, which has recently been studied in the resource theory of coherence \cite{Levi-2014a, Ringbauer-2018a}.  We begin by reviewing some elements of this theory \cite{Aberg-2006a, Baumgratz-2014a}.  For a $d$-dimensional quantum system, one begins by fixing an orthonormal basis $\{\ket{i}\}_{i=1}^d$ referred to as the incoherent basis.  Any state $\rho\in\mc{B}(\mbb{C}^{d})$ that is diagonal in the incoherent basis is called incoherent, and we denote the collection of all such states as $\mc{I}$.  Then all states not belonging to $\mc{I}$, i.e. all states with non-vanishing off-diagonal terms, are considered to be a resource.

In what sense is a state $\rho\not \in\mc{I}$ a resource?  One answer can be found in the context of phase discrimination games \cite{Coles-2016a, Biswas-2017a, Napoli-2016a, Piani-2016a}, a special case of channel discrimination games.  In its most basic form, a phase discrimination game is a two-party communication task in which Alice first encodes one of many possible phases $\{\theta_k\}_k$, with respective probabilities $\{p_k\}_k$, in a $d$-dimensional quantum state $\rho\in\mc{B}(\mbb{C}^d)$ by applying the unitary $U(\theta_k)=\sum_{j=0}^{d-1}e^{ij\theta_k}\op{j}{j}$ {where $\ket{j}$ is the fixed incoherent basis }.  Bob receives the state $\rho_k:=U(\theta_k)\rho U(-\theta_k)$ and attempts to determine the encoded variable $k$. As a figure of merit for how well the state $\rho$ can support the transmission of a phase ensemble $\Theta=\{(p_k,\theta_k)\}_k$, one could consider the maximum average success probability 
\begin{equation}
    p_\Theta^*(\rho)=\max_{\{\Pi_k\}_k}\sum_k p_k\tr\left[\rho_k\Pi_k\right],
\end{equation}
where the maximization is taken over all POVMs $\{\Pi_k\}_k$ on Bob's end.  Notice that if Alice attempts to encode in an incoherent state, the information of $k$ is lost: $U(\theta_k)\delta U(-\theta_k)=\delta$ for all $\theta_k$ and all $\delta\in\mc{I}$.  From this we see
\begin{equation}
    p_\Theta^{*}(\mc{I}):=p_\Theta^{*}(\delta)=\max_{\{\Pi_k\}_k}\sum_kp_k\tr\left[\delta\Pi_k\right]\leq \max_k p_k~\forall \delta\in\mc{I}.
\end{equation}
One of the main results established in Refs. \cite{Napoli-2016a, Piani-2016a} is that for any state $\rho\not\in\mc{I}$, 
\begin{equation}
\label{Eq:phase-robusteness}
    \max_{\Theta=\{(p_k,\theta_k)\}}\frac{p^*_{\Theta}(\rho)}{p^*_\Theta(\mc{I})}>1.
\end{equation}
Thus every coherent state offers an advantage in some phase discrimination game.

The bound of Eq.~\eqref{Eq:phase-robusteness} has been tightened in Ref. \cite{Ringbauer-2018a} by considering multi-level quantum coherence.  To explain this, let us say that the coherence rank of a pure state $\ket{\psi}\in\mbb{C}^d$ is the number of nonzero coherent amplitudes it possesses \cite{Levi-2014a}.  That is, given the expansion $\ket{\psi}=\sum_{i=1}^d c_i\ket{i}$, the coherence rank of $\ket{\psi}$ is $\crk(\ket{\psi})=\left|\{c_i\;|\;c_i\not=0\}\right|$.  
Beyond pure states, the coherence rank of a density matrix $\rho$ can be defined as
\begin{equation}
\crk(\rho)=\min_{\{p_i,\ket{\psi_i}\}}\max_i \crk(\ket{\psi_i}),
\end{equation}
where the minimization is taken over all ensembles $\{p_i,\ket{\psi_i}\}$ such that $\rho=\sum_i p_i\op{\psi_i}{\psi_i}$.  A state is said to have ($K+1$)-level coherence if its coherence rank is $K+1$, and the set of all states with coherence rank no more than $K$ will be denoted by $\mc{I}_{K}$.  Clearly 
\[\mc{I}=\mc{I}_1\subset\mc{I}_2\subset\cdots\subset \mc{I}_{d}=\mc{B}(\mbb{C}^{d}).\]
Consider the special case of noisy maximally coherent states $\rho=(1-\lambda)\mbb{I}/d+\lambda\op{\Psi_d}{\Psi_d}$, {where $\ket{\Psi_d}=1/\sqrt{d}\sum^d_{i=1}\ket{i}$. For the above state with coherence rank $K+1$}, Ref. \cite{Ringbauer-2018a} has strengthened Eq.~\eqref{Eq:phase-robusteness} to read
\begin{equation}
\label{Eq:phase-k-robusteness}
\max_{\Theta=\{(p_k,\theta_k)\}}\frac{p_\Theta^*(\rho)}{p_\Theta^*(\sigma)}>1 \qquad\forall \sigma\in\mc{I}_{K}.
\end{equation}
Whether such an inequality holds for general states of coherence rank $K+1$ remains an open problem.  However in Theorem \ref{thm:pure-N} and the following remark, we show that it is true for pure states after generalizing the phase discrimination game into a channel discrimination game.  That is, if $\op{\psi}{\psi}$ has coherence rank $K+1$, then
\begin{equation}
\label{Eq:channel-k-coherence}
\max_{\mc{E}=\{(p_k,\mc{E}_k)\}}\frac{p_\mc{E}^*(\rho)}{p_\mc{E}^*(\sigma)}>1 \qquad\forall \sigma\in\mc{I}_{K},
\end{equation}
where similarly the maximum average success probability is $ p_\mc{E}^*(\rho)=\max_{\{\Pi_k\}_k}\sum_k p_k\tr\left[\mc{E}_k(\rho)\Pi_k\right]$.  This is a consequence of the following theorem.
    \begin{theorem}
    \label{thm:pure-N}
    For $\ket{\psi}\in\mc{H}^{\msf{A}_1,\cdots,\msf{A}_N}_1$, $\crk(\ket{\psi})=K+1$ if and only if it violates at least one $(N,K)$-party generalized fingerprinting inequality.
    \end{theorem}
\begin{proof}
By definition, an $N$-level state without $(K+1)$-level coherence can be written as $\rho=\sum_{i}p_{i}\op{\psi_i}{\psi_i}$ where $\ket{\psi_i}$ has coherent rank no more than $K$ and thus has support on a subspace $\mc{H}_1^{\msf{A}_{i_1}\cdots\msf{A}_{i_K}}$. States of this form will only admit MACs having transition probabilities
\begin{align}
p(0|a_1,\cdots ,a_N)&=\tr\left\{\Pi_0\left[\mc{E}^{\msf{A}_1}_{a_1}\otimes\cdots\otimes \mc{E}^{\msf{A}_N}_{a_N}\left(\rho^{\bm{\msf{A}}}\right)\right]\right\}\notag\\&=\sum_i p_ig_i(0|a_{i_1},\cdots,a_{i_K}),
\end{align}
where $\tr\left\{\Pi_0\left[\mc{E}^{\msf{A}_{i_1}}_{a_{i_1}}\otimes\cdots\otimes \mc{E}^{\msf{A}_{i_K}}_{a_{i_K}}\left(\op{\psi_i}{\psi_i}\right)\right]\right\}=g_i(0|a_{i_1},\cdots,a_{i_K})$ with $\mc{E}^{\msf{A}_1}_{a_1}\otimes\cdots\otimes \mc{E}^{\msf{A}_N}_{a_N}$ being arbitrary NP encoding operation, and $\{\Pi_0,\mbb{I}-\Pi_0\}$ being any decoding POVM. If the encoding maps are required to be NPE operations, these MACs will belong to $\mc{C}_{N,K}$, and thus they cannot violate any $(N,K)$-generalized fingerprinting inequality. \par
To show that the converse is also true for pure state, let us consider an arbitrary pure state $\ket{\psi}=\sum_{i=1}^N c_i\ket{\mbf{e}_i}$ in $\mc{H}_1^{\bm{\msf{A}}}$ with coherent rank $\crk{\ket{\psi}}=K+1$. Without loss of generality, we assume the state has non-zero amplitude $c_i\ne 0$ for $i\in\{1,\cdots,K+1\}$ and consider the $(N,K)$-party generalized fingerprinting inequality among the first $K+1$ parties of the form of Eq.~(\ref{eq:fpg}). Since $\ket{\psi}$ only has support on a subspace $\mc{H}_1^{{\msf{A}_1}\cdots{\msf{A}_{K+1}}}$, we can treat the $(N,K)$-party generalized fingerprinting inequality as a $(K+1,K)$-party generalized fingerprinting inequality among the first $K+1$ parties with all the other parties having fixed inputs. And then we will show the violation of this inequality based on the violation we obtained for maximally coherent state in section \ref{sec:quantum violation}\par 
To begin with, let us rewrite the quantum state $\ket{\psi}=\sum_{i=1}^{K+1}c_i\ket{\mbf{e}_i}$ in terms of the $(K+1)$-party maximally coherent state  $\ket{\Psi_{K+1}}$ as: $\ket{\psi}=C\ket{\Psi_{K+1}}$, where $C$ is the coefficient matrix $C=\sqrt{K+1}\text{diag}[c_1,\dots,c_{K+1}]$. Since the map $\mc{E}_{a_i}^{\msf{A}_i}$ we used in defining $M$ in Eq. \eqref{Eq: M} consists of diagonal unitaries, we can define a similar matrix $M(\op{\psi}{\psi})$ for any pure state $\op{\psi}{\psi}$ as:
\begin{align}
M(\op{\psi}{\psi})&=M(C^{\dagger}\op{\Psi_{K+1}
}{\Psi_{K+1}}C)\notag\\&=C^{\dagger}M(\op{\Psi_{K+1}
}{\Psi_{K+1}})C=C^{\dagger}MC. 
\end{align}
In Section \ref{sec:quantum violation}, we have shown that $\tr{M}=K$ but $\norm{M}_1>K$ for some specific encoding strategies, which implies that there exists some normalized vector $\ket{\alpha}=\sum_{i=1}^{K+1} \alpha_i\ket{\mbf{e}_i}$ such that $\bra{\alpha}M\ket{\alpha}<0$. Based on this result, for an pure state with $(K+1)$-coherent rank, we can define a normalized vector $\ket{\alpha'}=\frac{1}{\sqrt{N_C}}C^{-1}\ket{\alpha}$, where $N_C$ is the normalization factor. Hence we can easily observe that:
\begin{align}
\bra{\alpha'}M(\op{\psi}{\psi})\ket{\alpha'}&=\frac{1}{N_C}\bra{\alpha}C^{\dagger-1}C^{\dagger}MCC^{-1}\ket{\alpha}\notag\\&=\frac{1}{N_C}\bra{\alpha}M\ket{\alpha}<0. 
\end{align}
Thus,  $M(\op{\psi}{\psi})$ is not positive semi-definite. Since $\tr{M(\op{\psi}{\psi})}=K$ still holds, we have $\norm{M(\op{\psi}{\psi})}_1>K$, which implies quantum violation $\delta>0$ in Eq.~\eqref{Eq: Violation}. \par 
To summarize, by using the same encoding strategy we have found for the maximally coherent state $\ket{\Psi_{K+1}}$ in section \ref{sec:quantum violation}, we can always obtain a quantum violation of a $(N,K)$-party generalized fingerprinting inequality for any pure state as long as it has $(K+1)$-level coherence. We can apply the same strategy above for any $N-1\ge K\ge 1$, the violation of different $(N,K)$-party generalized fingerprinting inequality will then herald the existence of different level of coherence for a given pure state. 
\end{proof}
\begin{remark}
If $\rho=\op{\psi}{\psi}$ is a pure state with $\crk(\rho)=K+1$, we can take $\mc{E}=\{(p_k,\mc{E}_k)\}_{k=0}^1$ being the channel ensemble with $p_0=\frac{1}{K+1}$,  $p_1=\frac{K}{K+1}$ and $\mc{E}_k$ given by:
\begin{equation}
\begin{split}
&\mc{E}_0=\mc{E}^{\msf{A}_1}_{0}\otimes\cdots\otimes \mc{E}^{\msf{A}_K}_{0}\\
&\mc{E}_1=\frac{1}{K}\sum_{i=1}^K\mc{E}^{\msf{A}_1}_{0}\otimes\cdots\otimes\mc{E}^{\msf{A}_i}_{1}\otimes\cdots\otimes \mc{E}^{\msf{A}_K}_{0},
\label{eq:channel ensemble}
\end{split}
\end{equation}
where $\mc{E}^{\msf{A}}_i$ can be any NPE operation. The maximum average success probability will then be expressed as:
\begin{align}
p_\mc{E}^*(\sigma)=\max_{\{\Pi_k\}_k}\sum_k p_k\tr\left[\mc{E}_k(\sigma)\Pi_k\right]=\max[I_{\text{F}}]
\end{align}
with: \[I_{\text{F}}=\frac{1}{K+1}\left(p(0|\overbrace{0\cdots0}^K)+\sum_{i=1 }^Kp(1|\overbrace{0\cdots,1_i,\cdots0}^K)\right) \]
From Theorem \ref{thm:pure-N}, we have that 
\begin{equation}
p_\mc{E}^*(\rho)>p_\mc{E}^*(\sigma)=\frac{K}{K+1}\quad \forall \sigma \in\mc{I}_K
\end{equation}
\end{remark}
An important consequence of this corollary is that a state must have $(K+1)$-level coherence whenever it can win with probability greater than $\frac{K}{K+1}$ the channel discrimination game given by the channel ensemble in Eq.~\eqref{eq:channel ensemble}. There is no information one needs to assume about the measurement device \textit{a priori} nor the channels themselves besides their inability to increase particle number.  Hence we conclude that the channel discrimination task here, or the $(N,K)$-party generalized fingerprinting inequality Eq.~\eqref{eq:fpg}, is a semi-device-independent witness of multilevel
coherence, and it is capable of measuring all pure multi-level coherent states. {The witness is device-independent on the decoder part, but still have some constrains on the encoder, i.e, local and NPE operation}\par

\subsection{Three-Level Coherence Witness}

For $K+1=3$, we have a stronger inequality-based criterion for three-level coherence, and the violation of this inequality is both necessary and sufficient for a large class of three-level quantum state.
 \begin{lemma}
\label{lemma:asso}
A quantum state has coherence three-level coherence if and only if $\norm{{\tilde{M}}(\rho)}_1>1$, where ${\tilde{M}}(\rho)$ is the associate comparison matrix of $\rho$:
\end{lemma}
As showed in Theorem 1 of \cite{Ringbauer-2018a}, a quantum state $\rho$ has three-level coherence if and only if its associated comparison matrix ${\tilde{M}}(\rho)$ is not positive-semidefinite or equivalently $\norm{{\tilde{M}}(\rho)}_1>1$, where $\norm{A}_1$ is the trace norm of A: $\norm{A}_1=\tr{\sqrt{A^{\dagger}A}}$ and the associated comparison matrix is defined as:
\begin{equation}
\label{Eq:Asso Matrix}
{\tilde{M}}(A)_{ij}=\left\{
                \begin{array}{ll}
                  &|A|_{ii}\quad \text{if}~~ i=j \\
                  -&|A|_{ij}\quad \text{if}~~ i\ne j
                \end{array}
              \right.
\end{equation}.  

\begin{proposition}
\label{prop:pure-3}
A three-level quantum state $\rho$ has multi-level coherence if it violates fingerprinting inequality Eq.~\eqref{eq:fpg}. The converse is also true if $U^{\dagger}\rho U$ is a non-negative matrix for some unitary $U$.
\end{proposition}
\begin{proof}
From Theorem \ref{thm:pure-N}, if a state $\rho$ has rank at most 2, it does not violate the fingerprinting inequality. \par
Conversely, consider a state $\rho\in\mc{B}(\mc{H}_1^{\msf{A}_1\msf{A}_2\msf{A}_3})$. As has been discussed in section \ref{sec:quantum violation}, the violation of fingerprinting inequality is equivalent to
$\norm{\mc{E}_{M}(\rho)}_1> 2$. If we take the $\{0,\pi\}$ phase encoding strategy with $\mc{E}_{a}(X)=\sigma_z^a(X)\sigma_z^a$ for $a\in\{0,1\}$ and $\sigma_z=\op{0}{0}+e^{i\pi}\op{1}{1}$, it is easy to see that:
\begin{equation}
{M}(A)_{ij}=\left\{
                \begin{array}{ll}
                  &2A_{ii}\quad \text{if}~~ i=j \\
                  -&2A_{ij}\quad \text{if}~~ i\ne j
                \end{array}
              \right.
              ,
\end{equation}
where $M(A)$ is defined similar as in Eq.~\eqref{Eq: M}.

For a three-level coherent state $\rho$ in $\mc{H}_1^{\msf{A}_1\msf{A}_2\msf{A}_3}$, ${{M}}(\rho)$ has a similar form as the associate comparison matrix ${\tilde{M}}(\rho)$. More precisely, {if there exists some unitary $U$ such that $U^{\dagger}\rho U$ is a non-negative matrix, ${\tilde{M}}(\rho)$ and $M(\rho)$ are only different by a global unitary and a factor of 2, i.e $2{\tilde{M}}(\rho^{\msf{A}_1\msf{A}_2\msf{A}_3})=U^{\dagger}{M}(\rho^{\msf{A}_1\msf{A}_2\msf{A}_3})U$ }, hence from Eq.\eqref{Eq: Violation}:
\begin{equation}
\delta=\frac{1}{2}(4+\norm{{M}(\rho)}_1)=\frac{1}{2}(4+2\norm{{\tilde{M}}(\rho)}_1)>3.
\end{equation}
\end{proof}
\begin{remark}
However, the conclusions drawn in Theorem \ref{thm:pure-N} and Proposition \ref{prop:pure-3} are not true for some three-level coherent mixed states with complex matrix elements. One counterexample can be found based on the following state:
\begin{equation*}
\rho=\begin{pmatrix}
\frac{1}{8}&\frac{1}{12}&\frac{1}{6}\\
\frac{1}{12}&\frac{1}{8}&\frac{i}{8}\\
\frac{1}{6}&-\frac{i}{8}&\frac{3}{4}. 
\end{pmatrix}
\end{equation*}
One can check $\norm{{\tilde{M}}(\rho)}>1$, which indicates that it has multi-level coherence from Lemma \ref{lemma:asso}; however, it can not violate fingerprinting inequality with any phase encoding and decoding strategy. Actually, it can not violate any inequality in $\mc{C}_{3,2}$ polytope (Eqns.~(\ref{eq:fp}), (\ref{Eq:ineq2}) and (\ref{Eq:ineq3})) with phase encoding strategy. Whether there is a quantum violation for a more general encoding strategy is yet unknown. 
\end{remark}

\section{Discussion}
The study of multiple-access channels in this paper provides an operational comparison between classical and quantum information processing in the context of multi-party communication. We have focused on a fine-grained analysis of the MACs that can be generated using only a single classical/quantum particle without touching any of its internal degrees of freedom, combined with a restricted type of encoding. The main results we have established focus on characterizing different MACs and showing the separation between classical MACs and quantum MACs in an operational way. Standard quantum fingerprinting can be seen as just one of many different enhancements that emerge when using quantum MACs. We have also analyzed the situation in which the locality constraints on the encoders are partially relaxed and a richer structure of these MACs was revealed.   \par 

In particular, we have identified the generalized fingerprinting inequalities as valid facets for different $(N,K)$-local classical MACs. This provides one route for quantifying the quantum advantage of a single quantum particle in multi-party communication scenarios.  Further, $N$-local quantum MACs can outperform $N$-party classical MACs, even when $N-1$ of the parties are allowed to collaborate. Finally, we have highlighted the connection of our framework to the resource theory of multilevel coherence and provided a semi-device independent approach to witnessing multi-level coherence for a quantum state. \par
Our operational framework for analyzing the multiple-access channel is important, not only for showing the fundamental differences between a classical particle and a quantum particle with spatial superposition, but also for quantum enhanced multi-party communication and potential applications in sensing of complex quantum systems due to the connection to multi-level coherence \cite{Ringbauer-2018a,Giovannetti-2011}.\par 
There are many questions remaining from our investigation.  We have identified the generalized fingerprinting inequality as a valid facet for all $\mc{C}_{N,K}$ polytopes.  It would be interesting to identify other facets, something which is infeasible to do using the standard numerical approach.  A related question is to decide which of the $\mc{C}_{N,K}$ polytopes can be violated by a given $N$-partite quantum state and how large such a violation can be.  For the case of $\mc{Q}_{N,K}$, we showed that the $I_{2K}$ equality can be violated for all non-classical states, yet the maximal violation for a given state is still unknown.  Additionally, in Section \ref{sec:quantum violation} we found a quantum violation of the $(N,N-1)$-party generalized fingerprinting inequality that scales like $N^{-3}$, which is consistent with the previous results of Ref. \cite{Horvat-2019a}.  It would be good to know whether an $O(N^{-3})$ violation is optimal, for not only the generalized fingerprinting inequality but also for all other facets of $\mc{C}_{N,N-1}([2]^N;[2])$.  It would also be interesting to further pursue the connection between general multi-level coherent states and the $(N,K)$-local facet inequalities.  Finally, the current results we have obtained are based on using just a one single particle. {Though, the $C_{N,K}^{(sep)}$ already encapsulate MACs attainable with K classical particles and arbitrary internal degrees of freedom}, a natural next question is how the classical and quantum MACs compare when more than one particle or multiple internal degrees of freedom are utilized. 

\section{Acknowledgments}

This work was supported by NSF Award 1839177.  The authors thank Virginia Lorenz, Paul Kwiat, Andreas Winter and Kai Shinbrough for helpful discussions during the preparation of this manuscript.

\bibliographystyle{plain}

\begin{thebibliography}{60}
\providecommand{\natexlab}[1]{#1}
\providecommand{\url}[1]{\texttt{#1}}
\expandafter\ifx\csname urlstyle\endcsname\relax
  \providecommand{\doi}[1]{doi: #1}\else
  \providecommand{\doi}{doi: \begingroup \urlstyle{rm}\Url}\fi

\bibitem[Cover and Thomas(2006)]{Cover-2006a}
Thomas~M. Cover and Joy~A. Thomas.
\newblock \emph{Elements of Information Theory (Wiley Series in
  Telecommunications and Signal Processing)}.
\newblock Wiley-Interscience, USA, 2006.
\newblock ISBN 0471241954.
\newblock \doi{10.5555/1146355}.
\newblock URL \url{https://dl.acm.org/doi/10.5555/1146355}.

\bibitem[Tse and Viswanath(2005)]{Tse-2005a}
David Tse and Pramod Viswanath.
\newblock \emph{Fundamentals of Wireless Communication}.
\newblock Cambridge University Press, Cambridge, 2005.
\newblock ISBN 9780521845274.
\newblock \doi{10.1017/CBO9780511807213}.
\newblock URL \url{https://doi.org/10.1017/CBO9780511807213}.

\bibitem[Kristj\'ansson et~al.(2021)Kristj\'ansson, Mao, and
  Chiribella]{Kristjansson-2020a}
Hl\'er Kristj\'ansson, Wenxu Mao, and Giulio Chiribella.
\newblock Witnessing latent time correlations with a single quantum particle.
\newblock \emph{Phys. Rev. Research}, 3:\penalty0 043147, Nov 2021.
\newblock \doi{10.1103/PhysRevResearch.3.043147}.
\newblock URL \url{https://link.aps.org/doi/10.1103/PhysRevResearch.3.043147}.

\bibitem[Ebler et~al.(2018)Ebler, Salek, and Chiribella]{Ebler-2018a}
Daniel Ebler, Sina Salek, and Giulio Chiribella.
\newblock Enhanced communication with the assistance of indefinite causal
  order.
\newblock \emph{Phys. Rev. Lett.}, 120:\penalty0 120502, Mar 2018.
\newblock \doi{10.1103/PhysRevLett.120.120502}.
\newblock URL \url{https://link.aps.org/doi/10.1103/PhysRevLett.120.120502}.

\bibitem[Goswami et~al.(2020)Goswami, Cao, Paz-Silva, Romero, and
  White]{Goswami-2020}
K.~Goswami, Y.~Cao, G.~A. Paz-Silva, J.~Romero, and A.~G. White.
\newblock Increasing communication capacity via superposition of order.
\newblock \emph{Phys. Rev. Research}, 2:\penalty0 033292, Aug 2020.
\newblock \doi{10.1103/PhysRevResearch.2.033292}.
\newblock URL \url{https://link.aps.org/doi/10.1103/PhysRevResearch.2.033292}.

\bibitem[Salek et~al.(2018)Salek, Ebler, and Chiribella]{salek-2018}
Sina Salek, Daniel Ebler, and Giulio Chiribella.
\newblock Quantum communication in a superposition of causal orders, 2018.
\newblock URL \url{https://arxiv.org/abs/1809.06655}.

\bibitem[Procopio et~al.(2019)Procopio, Delgado, Enríquez, Belabas, and
  Levenson]{Procopio-2019}
Lorenzo~M. Procopio, Francisco Delgado, Marco Enríquez, Nadia Belabas, and
  Juan~Ariel Levenson.
\newblock Communication enhancement through quantum coherent control of n
  channels in an indefinite causal-order scenario.
\newblock \emph{Entropy}, 21\penalty0 (10), 2019.
\newblock ISSN 1099-4300.
\newblock \doi{10.3390/e21101012}.
\newblock URL \url{https://www.mdpi.com/1099-4300/21/10/1012}.

\bibitem[Procopio et~al.(2020)Procopio, Delgado, Enr\'{\i}quez, Belabas, and
  Levenson]{Procopio-2020}
Lorenzo~M. Procopio, Francisco Delgado, Marco Enr\'{\i}quez, Nadia Belabas, and
  Juan~Ariel Levenson.
\newblock Sending classical information via three noisy channels in
  superposition of causal orders.
\newblock \emph{Phys. Rev. A}, 101:\penalty0 012346, Jan 2020.
\newblock \doi{10.1103/PhysRevA.101.012346}.
\newblock URL \url{https://link.aps.org/doi/10.1103/PhysRevA.101.012346}.

\bibitem[Chiribella et~al.(2021)Chiribella, Wilson, and Chau]{Chiribella-2021}
Giulio Chiribella, Matt Wilson, and H.~F. Chau.
\newblock Quantum and classical data transmission through completely
  depolarizing channels in a superposition of cyclic orders.
\newblock \emph{Phys. Rev. Lett.}, 127:\penalty0 190502, Nov 2021.
\newblock \doi{10.1103/PhysRevLett.127.190502}.
\newblock URL \url{https://link.aps.org/doi/10.1103/PhysRevLett.127.190502}.

\bibitem[Abbott et~al.(2020)Abbott, Wechs, Horsman, Mhalla, and
  Branciard]{Alastair-2018a}
Alastair~A. Abbott, Julian Wechs, Dominic Horsman, Mehdi Mhalla, and Cyril
  Branciard.
\newblock Communication through coherent control of quantum channels.
\newblock \emph{{Quantum}}, 4:\penalty0 333, September 2020.
\newblock ISSN 2521-327X.
\newblock \doi{10.22331/q-2020-09-24-333}.
\newblock URL \url{https://doi.org/10.22331/q-2020-09-24-333}.

\bibitem[Gisin et~al.(2005)Gisin, Linden, Massar, and Popescu]{Gisin-2005}
N.~Gisin, N.~Linden, S.~Massar, and S.~Popescu.
\newblock Error filtration and entanglement purification for quantum
  communication.
\newblock \emph{Phys. Rev. A}, 72:\penalty0 012338, Jul 2005.
\newblock \doi{10.1103/PhysRevA.72.012338}.
\newblock URL \url{https://link.aps.org/doi/10.1103/PhysRevA.72.012338}.

\bibitem[Chiribella and Kristjánsson(2019)]{Chiribella-2019a}
Giulio Chiribella and Hlér Kristjánsson.
\newblock Quantum shannon theory with superpositions of trajectories.
\newblock \emph{Proceedings of the Royal Society A: Mathematical, Physical and
  Engineering Sciences}, 475\penalty0 (2225):\penalty0 20180903, 2019.
\newblock \doi{10.1098/rspa.2018.0903}.
\newblock URL \url{https://doi.org/10.1098/rspa.2018.0903}.

\bibitem[Kristj{\'{a}}nsson et~al.(2020)Kristj{\'{a}}nsson, Chiribella, Salek,
  Ebler, and Wilson]{Kristjansson-2020b}
Hl{\'{e}}r Kristj{\'{a}}nsson, Giulio Chiribella, Sina Salek, Daniel Ebler, and
  Matthew Wilson.
\newblock Resource theories of communication.
\newblock \emph{New Journal of Physics}, 22\penalty0 (7):\penalty0 073014, jul
  2020.
\newblock \doi{10.1088/1367-2630/ab8ef7}.

\bibitem[Buhrman et~al.(2001)Buhrman, Cleve, Watrous, and
  de~Wolf]{Buhrman-2001a}
Harry Buhrman, Richard Cleve, John Watrous, and Ronald de~Wolf.
\newblock Quantum fingerprinting.
\newblock \emph{Phys. Rev. Lett.}, 87:\penalty0 167902, Sep 2001.
\newblock \doi{10.1103/PhysRevLett.87.167902}.
\newblock URL \url{https://link.aps.org/doi/10.1103/PhysRevLett.87.167902}.

\bibitem[Horn et~al.(2005{\natexlab{a}})Horn, Scott, Walgate, Cleve, Lvovsky,
  and Sanders]{Horn-2005a}
Rolf~T. Horn, A.~J. Scott, Jonathan Walgate, Richard Cleve, A.~I. Lvovsky, and
  Barry~C. Sanders.
\newblock Classical and quantum fingerprinting with shared randomness and
  one-sided error.
\newblock \emph{Quantum Inf. Comput.}, 5, 2005{\natexlab{a}}.
\newblock URL \url{https://dl.acm.org/doi/10.5555/2011637.2011643}.

\bibitem[Massar(2005)]{Massar-2005a}
S.~Massar.
\newblock Quantum fingerprinting with a single particle.
\newblock \emph{Phys. Rev. A}, 71:\penalty0 012310, Jan 2005.
\newblock \doi{10.1103/PhysRevA.71.012310}.
\newblock URL \url{https://link.aps.org/doi/10.1103/PhysRevA.71.012310}.

\bibitem[Horn et~al.(2005{\natexlab{b}})Horn, Babichev, Marzlin, Lvovsky, and
  Sanders]{Horn-2005b}
Rolf~T. Horn, S.~A. Babichev, Karl-Peter Marzlin, A.~I. Lvovsky, and Barry~C.
  Sanders.
\newblock Single-qubit optical quantum fingerprinting.
\newblock \emph{Phys. Rev. Lett.}, 95:\penalty0 150502, Oct 2005{\natexlab{b}}.
\newblock \doi{10.1103/PhysRevLett.95.150502}.
\newblock URL \url{https://link.aps.org/doi/10.1103/PhysRevLett.95.150502}.

\bibitem[Del~Santo and Daki\ifmmode~\acute{c}\else
  \'{c}\fi{}(2018)]{DelSanto-2018a}
Flavio Del~Santo and Borivoje Daki\ifmmode~\acute{c}\else \'{c}\fi{}.
\newblock Two-way communication with a single quantum particle.
\newblock \emph{Phys. Rev. Lett.}, 120:\penalty0 060503, Feb 2018.
\newblock \doi{10.1103/PhysRevLett.120.060503}.
\newblock URL \url{https://link.aps.org/doi/10.1103/PhysRevLett.120.060503}.

\bibitem[Hsu et~al.(2020)Hsu, Lai, Chang, Wu, and Lee]{Hsu-2020a}
Li-Yi Hsu, Ching-Yi Lai, You-Chia Chang, Chien-Ming Wu, and Ray-Kuang Lee.
\newblock Carrying an arbitrarily large amount of information using a single
  quantum particle.
\newblock \emph{Phys. Rev. A}, 102:\penalty0 022620, Aug 2020.
\newblock \doi{10.1103/PhysRevA.102.022620}.
\newblock URL \url{https://link.aps.org/doi/10.1103/PhysRevA.102.022620}.

\bibitem[Del~Santo and Daki\ifmmode~\acute{c}\else
  \'{c}\fi{}(2020)]{DelSanto-2019a}
Flavio Del~Santo and Borivoje Daki\ifmmode~\acute{c}\else \'{c}\fi{}.
\newblock Coherence equality and communication in a quantum superposition.
\newblock \emph{Phys. Rev. Lett.}, 124:\penalty0 190501, May 2020.
\newblock \doi{10.1103/PhysRevLett.124.190501}.
\newblock URL \url{https://link.aps.org/doi/10.1103/PhysRevLett.124.190501}.

\bibitem[Horvat and Daki{\'{c}}(2021{\natexlab{a}})]{Horvat-2019a}
Sebastian Horvat and Borivoje Daki{\'{c}}.
\newblock Quantum enhancement to information acquisition speed.
\newblock \emph{New Journal of Physics}, 23\penalty0 (3):\penalty0 033008, mar
  2021{\natexlab{a}}.
\newblock \doi{10.1088/1367-2630/abe9d4}.
\newblock URL \url{https://doi.org/10.1088/1367-2630/abe9d4}.

\bibitem[Garcia-Escartin and Chamorro-Posada(2013)]{Garcia-Escartin-2013a}
Juan~Carlos Garcia-Escartin and Pedro Chamorro-Posada.
\newblock swap test and hong-ou-mandel effect are equivalent.
\newblock \emph{Phys. Rev. A}, 87:\penalty0 052330, May 2013.
\newblock \doi{10.1103/PhysRevA.87.052330}.
\newblock URL \url{https://link.aps.org/doi/10.1103/PhysRevA.87.052330}.

\bibitem[Arrazola and L\"utkenhaus(2014)]{Arrazola-2014a}
Juan~Miguel Arrazola and Norbert L\"utkenhaus.
\newblock Quantum communication with coherent states and linear optics.
\newblock \emph{Phys. Rev. A}, 90:\penalty0 042335, Oct 2014.
\newblock \doi{10.1103/PhysRevA.90.042335}.
\newblock URL \url{https://link.aps.org/doi/10.1103/PhysRevA.90.042335}.

\bibitem[Horodecki and Oppenheim(2013)]{Horodecki-2013a}
Micha\l{} Horodecki and Jonathan Oppenheim.
\newblock (quantumness in the context of) resource theories.
\newblock \emph{International Journal of Modern Physics B}, 27\penalty0
  (01n03):\penalty0 1345019, 2013.
\newblock \doi{10.1142/S0217979213450197}.
\newblock URL \url{https://doi.org/10.1142/S0217979213450197}.

\bibitem[Coecke et~al.(2016)Coecke, Fritz, and Spekkens]{Coecke-2016a}
Bob Coecke, Tobias Fritz, and Robert~W. Spekkens.
\newblock A mathematical theory of resources.
\newblock \emph{Information and Computation}, 250:\penalty0 59 -- 86, 2016.
\newblock ISSN 0890-5401.
\newblock \doi{https://doi.org/10.1016/j.ic.2016.02.008}.
\newblock URL
  \url{http://www.sciencedirect.com/science/article/pii/S0890540116000353}.

\bibitem[Chitambar and Gour(2019)]{Chitambar-2019a}
Eric Chitambar and Gilad Gour.
\newblock Quantum resource theories.
\newblock \emph{Rev. Mod. Phys.}, 91:\penalty0 025001, Apr 2019.
\newblock \doi{10.1103/RevModPhys.91.025001}.
\newblock URL \url{https://link.aps.org/doi/10.1103/RevModPhys.91.025001}.

\bibitem[{Winter}(2001)]{Winter-2001a}
A.~{Winter}.
\newblock The capacity of the quantum multiple-access channel.
\newblock \emph{IEEE Transactions on Information Theory}, 47\penalty0
  (7):\penalty0 3059--3065, 2001.
\newblock \doi{10.1109/18.959287}.
\newblock URL \url{https://doi.org/10.1109/18.959287}.

\bibitem[{Hsieh} et~al.(2008){Hsieh}, {Devetak}, and {Winter}]{Hsieh-2008a}
M.~{Hsieh}, I.~{Devetak}, and A.~{Winter}.
\newblock Entanglement-assisted capacity of quantum multiple-access channels.
\newblock \emph{IEEE Transactions on Information Theory}, 54\penalty0
  (7):\penalty0 3078--3090, 2008.
\newblock \doi{10.1109/TIT.2008.924726}.
\newblock URL \url{https://doi.org/10.1109/TIT.2008.924726}.

\bibitem[Leditzky et~al.(2020)Leditzky, Alhejji, Levin, and
  Smith]{Leditzky-2020a}
Felix Leditzky, Mohammad~A. Alhejji, Joshua Levin, and Graeme Smith.
\newblock Playing games with multiple access channels.
\newblock \emph{Nature Communications}, 11\penalty0 (1):\penalty0 1497, Mar
  2020.
\newblock ISSN 2041-1723.
\newblock \doi{10.1038/s41467-020-15240-w}.
\newblock URL \url{https://doi.org/10.1038/s41467-020-15240-w}.

\bibitem[Rozema et~al.(2021)Rozema, Zhuo, Paterek, and
  Daki\ifmmode~\acute{c}\else \'{c}\fi{}]{Rozema-2020}
Lee~A. Rozema, Zhao Zhuo, Tomasz Paterek, and Borivoje
  Daki\ifmmode~\acute{c}\else \'{c}\fi{}.
\newblock Higher-order interference between multiple quantum particles
  interacting nonlinearly.
\newblock \emph{Phys. Rev. A}, 103:\penalty0 052204, May 2021.
\newblock \doi{10.1103/PhysRevA.103.052204}.
\newblock URL \url{https://link.aps.org/doi/10.1103/PhysRevA.103.052204}.

\bibitem[Grangier et~al.(1986)Grangier, Roger, and Aspect]{Grangier-1986}
P~Grangier, G~Roger, and A~Aspect.
\newblock Experimental evidence for a photon anticorrelation effect on a beam
  splitter: A new light on single-photon interferences.
\newblock \emph{Europhysics Letters ({EPL})}, 1\penalty0 (4):\penalty0
  173--179, feb 1986.
\newblock \doi{10.1209/0295-5075/1/4/004}.
\newblock URL \url{https://doi.org/10.1209%2F0295-5075%2F1%2F4%2F004}.

\bibitem[Jacques et~al.(2005)Jacques, Wu, Toury, Treussart, Aspect, Grangier,
  and Roch]{Jacques-2005}
V.~Jacques, E.~Wu, T.~Toury, F.~Treussart, A.~Aspect, P.~Grangier, and J.-F.
  Roch.
\newblock Single-photon wavefront-splitting interference.
\newblock \emph{The European Physical Journal D - Atomic, Molecular, Optical
  and Plasma Physics}, 35\penalty0 (3):\penalty0 561--565, Sep 2005.
\newblock ISSN 1434-6079.
\newblock \doi{10.1140/epjd/e2005-00201-y}.
\newblock URL \url{https://doi.org/10.1140/epjd/e2005-00201-y}.

\bibitem[Merzbacher(1998)]{Merzbacher-1998a}
E.~Merzbacher.
\newblock \emph{Quantum Mechanics}.
\newblock Wiley, 1998.
\newblock ISBN 9780471887027.
\newblock URL
  \url{https://www.worldcat.org/title/quantum-mechanics/oclc/246969310}.

\bibitem[Horvat(2019)]{Horvat-2019b}
Sebastian Horvat.
\newblock Quantum superposition as a resource for quantum communication.
\newblock Master's thesis, University of Zagreb, Croatia, 2019.
\newblock URL \url{https://zir.nsk.hr/en/islandora/object/pmf%3A7648}.

\bibitem[Horvat and Daki{\'{c}}(2021{\natexlab{b}})]{Horvat-2020a}
Sebastian Horvat and Borivoje Daki{\'{c}}.
\newblock Interference as an information-theoretic game.
\newblock \emph{{Quantum}}, 5:\penalty0 404, March 2021{\natexlab{b}}.
\newblock ISSN 2521-327X.
\newblock \doi{10.22331/q-2021-03-08-404}.
\newblock URL \url{https://doi.org/10.22331/q-2021-03-08-404}.

\bibitem[Sorkin(1994)]{Sorkin-1994a}
Rafael~D. Sorkin.
\newblock Quantum mechanics as quantum measure theory.
\newblock \emph{Modern Physics Letters A}, 09\penalty0 (33):\penalty0
  3119--3127, 1994.
\newblock \doi{10.1142/S021773239400294X}.
\newblock URL \url{https://doi.org/10.1142/S021773239400294X}.

\bibitem[Sinha et~al.(2010)Sinha, Couteau, Jennewein, Laflamme, and
  Weihs]{Sinha-2010a}
Urbasi Sinha, Christophe Couteau, Thomas Jennewein, Raymond Laflamme, and
  Gregor Weihs.
\newblock Ruling out multi-order interference in quantum mechanics.
\newblock \emph{Science}, 329\penalty0 (5990):\penalty0 418--421, 2010.
\newblock ISSN 0036-8075.
\newblock \doi{10.1126/science.1190545}.
\newblock URL \url{https://science.sciencemag.org/content/329/5990/418}.

\bibitem[Ududec et~al.(2011)Ududec, Barnum, and Emerson]{Ududec-2010a}
Cozmin Ududec, Howard Barnum, and Joseph Emerson.
\newblock Three slit experiments and the structure of quantum theory.
\newblock \emph{Foundations of Physics}, 41\penalty0 (3):\penalty0 396--405,
  Mar 2011.
\newblock ISSN 1572-9516.
\newblock \doi{10.1007/s10701-010-9429-z}.
\newblock URL \url{https://doi.org/10.1007/s10701-010-9429-z}.

\bibitem[Lee and Selby(2017)]{Lee-2016a}
Ciar{\'a}n~M. Lee and John~H. Selby.
\newblock Higher-order interference in extensions of quantum theory.
\newblock \emph{Foundations of Physics}, 47\penalty0 (1):\penalty0 89--112, Jan
  2017.
\newblock ISSN 1572-9516.
\newblock \doi{10.1007/s10701-016-0045-4}.
\newblock URL \url{https://doi.org/10.1007/s10701-016-0045-4}.

\bibitem[Daki{\'{c}} et~al.(2014)Daki{\'{c}}, Paterek, and Brukner]{Daki-2014}
B~Daki{\'{c}}, T~Paterek, and {\v{C}}~Brukner.
\newblock Density cubes and higher-order interference theories.
\newblock \emph{New Journal of Physics}, 16\penalty0 (2):\penalty0 023028, feb
  2014.
\newblock \doi{10.1088/1367-2630/16/2/023028}.
\newblock URL \url{https://doi.org/10.1088%2F1367-2630%2F16%2F2%2F023028}.

\bibitem[Gour and Spekkens(2008)]{Gour_Spekkens}
Gilad Gour and Robert~W Spekkens.
\newblock The resource theory of quantum reference frames: manipulations and
  monotones.
\newblock \emph{New Journal of Physics}, 10\penalty0 (3):\penalty0 033023, mar
  2008.
\newblock \doi{10.1088/1367-2630/10/3/033023}.
\newblock URL \url{https://doi.org/10.1088%2F1367-2630%2F10%2F3%2F033023}.

\bibitem[{Cleve} et~al.(2004){Cleve}, {Hoyer}, {Toner}, and
  {Watrous}]{Cleve-2004a}
R.~{Cleve}, P.~{Hoyer}, B.~{Toner}, and J.~{Watrous}.
\newblock Consequences and limits of nonlocal strategies.
\newblock In \emph{Proceedings. 19th IEEE Annual Conference on Computational
  Complexity, 2004.}, pages 236--249, 2004.
\newblock \doi{10.1109/CCC.2004.1313847}.
\newblock URL \url{https://doi.org/10.1109/CCC.2004.1313847}.

\bibitem[Brunner et~al.(2014)Brunner, Cavalcanti, Pironio, Scarani, and
  Wehner]{Brunner-2014a}
Nicolas Brunner, Daniel Cavalcanti, Stefano Pironio, Valerio Scarani, and
  Stephanie Wehner.
\newblock Bell nonlocality.
\newblock \emph{Rev. Mod. Phys.}, 86:\penalty0 419--478, Apr 2014.
\newblock \doi{10.1103/RevModPhys.86.419}.
\newblock URL \url{https://link.aps.org/doi/10.1103/RevModPhys.86.419}.

\bibitem[Barvinok(2002)]{Barvinok-2002a}
A.~Barvinok.
\newblock \emph{A Course in Convexity}.
\newblock Graduate studies in mathematics. American Mathematical Society, 2002.
\newblock URL \url{https://bookstore.ams.org/gsm-54/}.

\bibitem[Masanes(2003)]{Masanes-2002a}
Ll~Masanes.
\newblock Tight bell inequality for d-outcome measurements correlations.
\newblock \emph{Quantum Info. Comput.}, 3\penalty0 (4):\penalty0 345–358,
  July 2003.
\newblock ISSN 1533-7146.
\newblock \doi{10.5555/2011528.2011532}.
\newblock URL \url{https://dl.acm.org/doi/10.5555/2011528.2011532}.

\bibitem[Pironio(2005)]{Pironio-2005a}
Stefano Pironio.
\newblock Lifting bell inequalities.
\newblock \emph{Journal of Mathematical Physics}, 46\penalty0 (6):\penalty0
  062112, 2005.
\newblock \doi{10.1063/1.1928727}.
\newblock URL \url{https://doi.org/10.1063/1.1928727}.

\bibitem[Christof and L{\"o}bel(1997)]{porta-1997}
T.~Christof and A.~L{\"o}bel.
\newblock \emph{porta, URL http://porta.zib.de}, 1997.
\newblock URL \url{http://porta.zib.de/}.

\bibitem[Born(1926)]{Born-1926}
Max Born.
\newblock Zur quantenmechanik der sto{\ss}vorg{\"a}nge.
\newblock \emph{Zeitschrift f{\"u}r Physik}, 37\penalty0 (12):\penalty0
  863--867, Dec 1926.
\newblock ISSN 0044-3328.
\newblock \doi{10.1007/BF01397477}.
\newblock URL \url{https://doi.org/10.1007/BF01397477}.

\bibitem[Biswas et~al.(2017)Biswas, Garc\'{i}a~Díaz, and Winter]{Biswas-2017a}
Tanmoy Biswas, María Garc\'{i}a~Díaz, and Andreas Winter.
\newblock Interferometric visibility and coherence.
\newblock \emph{Proceedings of the Royal Society A: Mathematical, Physical and
  Engineering Sciences}, 473\penalty0 (2203):\penalty0 20170170, 2017.
\newblock \doi{10.1098/rspa.2017.0170}.
\newblock URL \url{https://doi.org/10.1098/rspa.2017.0170}.


\bibitem[Svetlichny(1987)]{Svetlichny-1987a}
George Svetlichny.
\newblock Distinguishing three-body from two-body nonseparability by a
  bell-type inequality.
\newblock \emph{Phys. Rev. D}, 35:\penalty0 3066--3069, May 1987.
\newblock \doi{10.1103/PhysRevD.35.3066}.
\newblock URL \url{https://link.aps.org/doi/10.1103/PhysRevD.35.3066}.

\bibitem[Helstrom(1969)]{helstrom-1969}
Carl~W. Helstrom.
\newblock Quantum detection and estimation theory.
\newblock \emph{Journal of Statistical Physics}, 1\penalty0 (2):\penalty0
  231--252, Jun 1969.
\newblock ISSN 1572-9613.
\newblock \doi{10.1007/BF01007479}.
\newblock URL \url{https://doi.org/10.1007/BF01007479}.

\bibitem[Takagi and Regula(2019)]{Takagi-2019}
Ryuji Takagi and Bartosz Regula.
\newblock General resource theories in quantum mechanics and beyond:
  Operational characterization via discrimination tasks.
\newblock \emph{Phys. Rev. X}, 9:\penalty0 031053, Sep 2019.
\newblock \doi{10.1103/PhysRevX.9.031053}.
\newblock URL \url{https://link.aps.org/doi/10.1103/PhysRevX.9.031053}.

\bibitem[Levi and Mintert(2014)]{Levi-2014a}
Federico Levi and Florian Mintert.
\newblock A quantitative theory of coherent delocalization.
\newblock \emph{New Journal of Physics}, 16\penalty0 (3):\penalty0 033007, mar
  2014.
\newblock \doi{10.1088/1367-2630/16/3/033007}.
\newblock URL \url{https://doi.org/10.1088%2F1367-2630%2F16%2F3%2F033007}.

\bibitem[Ringbauer et~al.(2018)Ringbauer, Bromley, Cianciaruso, Lami, Lau,
  Adesso, White, Fedrizzi, and Piani]{Ringbauer-2018a}
Martin Ringbauer, Thomas~R. Bromley, Marco Cianciaruso, Ludovico Lami,
  W.~Y.~Sarah Lau, Gerardo Adesso, Andrew~G. White, Alessandro Fedrizzi, and
  Marco Piani.
\newblock Certification and quantification of multilevel quantum coherence.
\newblock \emph{Phys. Rev. X}, 8:\penalty0 041007, Oct 2018.
\newblock \doi{10.1103/PhysRevX.8.041007}.
\newblock URL \url{https://link.aps.org/doi/10.1103/PhysRevX.8.041007}.

\bibitem[\"{A}berg(2006)]{Aberg-2006a}
Johan \"{A}berg.
\newblock Quantifying superposition, 2006.

\bibitem[Baumgratz et~al.(2014)Baumgratz, Cramer, and Plenio]{Baumgratz-2014a}
T.~Baumgratz, M.~Cramer, and M.~B. Plenio.
\newblock Quantifying coherence.
\newblock \emph{Phys. Rev. Lett.}, 113:\penalty0 140401, Sep 2014.
\newblock \doi{10.1103/PhysRevLett.113.140401}.
\newblock URL \url{https://link.aps.org/doi/10.1103/PhysRevLett.113.140401}.

\bibitem[Coles(2016)]{Coles-2016a}
Patrick~J. Coles.
\newblock Entropic framework for wave-particle duality in multipath
  interferometers.
\newblock \emph{Phys. Rev. A}, 93:\penalty0 062111, Jun 2016.
\newblock \doi{10.1103/PhysRevA.93.062111}.
\newblock URL \url{https://link.aps.org/doi/10.1103/PhysRevA.93.062111}.

\bibitem[Napoli et~al.(2016)Napoli, Bromley, Cianciaruso, Piani, Johnston, and
  Adesso]{Napoli-2016a}
Carmine Napoli, Thomas~R. Bromley, Marco Cianciaruso, Marco Piani, Nathaniel
  Johnston, and Gerardo Adesso.
\newblock Robustness of coherence: An operational and observable measure of
  quantum coherence.
\newblock \emph{Phys. Rev. Lett.}, 116:\penalty0 150502, Apr 2016.
\newblock \doi{10.1103/PhysRevLett.116.150502}.
\newblock URL \url{https://link.aps.org/doi/10.1103/PhysRevLett.116.150502}.

\bibitem[Piani et~al.(2016)Piani, Cianciaruso, Bromley, Napoli, Johnston, and
  Adesso]{Piani-2016a}
Marco Piani, Marco Cianciaruso, Thomas~R. Bromley, Carmine Napoli, Nathaniel
  Johnston, and Gerardo Adesso.
\newblock Robustness of asymmetry and coherence of quantum states.
\newblock \emph{Phys. Rev. A}, 93:\penalty0 042107, Apr 2016.
\newblock \doi{10.1103/PhysRevA.93.042107}.
\newblock URL \url{https://link.aps.org/doi/10.1103/PhysRevA.93.042107}.

\bibitem[Giovannetti et~al.(2011)Giovannetti, Lloyd, and
  Maccone]{Giovannetti-2011}
Vittorio Giovannetti, Seth Lloyd, and Lorenzo Maccone.
\newblock Advances in quantum metrology.
\newblock \emph{Nature Photonics}, 5\penalty0 (4):\penalty0 222--229, Apr 2011.
\newblock ISSN 1749-4893.
\newblock \doi{10.1038/nphoton.2011.35}.
\newblock URL \url{https://doi.org/10.1038/nphoton.2011.35}.

\end{thebibliography}

\onecolumn\newpage
\appendix

Here we provide detailed proofs of proposition \ref{Prop:non-convex}, \ref{Prop:non-convex2}, \ref{Prop:non-convex3} and \ref{prop:N,K convexity} for the inclusion relation between different classical MACs in section \ref{sec:A1 separation}; proposition \ref{prop:Rank Qk general} and theorem \ref{thm:Rank Ck general} for the dimension of classical and quantum MACs for arbitary inputs and output in section \ref{sec:A2 dimension}; and proposition \ref{prop: seperation}, conjecture \ref{prop: nonpolytope}  for the separation between the most simple classical MACs $\mc{C}_{2,2}([2]^2,[2])$ and quantum MACs $\mc{Q}_{2,1}([2]^2,[2])$ in section \ref{sec:A3 seperation}.
\section{Separation and convexity of $(N,K)$-local Classical MACs}
\label{sec:A1 separation}
\setcounter{proposition}{3}
\begin{proposition}
For $|\mc{B}|>2$ and $N\geq 2$, the set $\mc{C}_{N}(\bcA,\mc{B})$ is non-convex; hence $\mc{C}_{N}(\bcA,\mc{B})\not=\mc{C}'_{N}(\bcA,\mc{B})$.
\end{proposition}
\begin{proof}
It suffices to prove the statement for $\mc{C}_{2}([2]^2;[3])$.  Consider $\mbf{p}'_{B|A_1A_2},\mbf{p}''_{B|A_1A_2}\in \mc{C}_{2}([2]^2;[3])$ having respective coordinates $p'(b|a_1,a_2)=g_1(b|a_1)$, with $g_1(0|0)=1$ and $g_1(1|1)=1$, and $p''(b|a_1,a_2)=g_2(b|a_2)$, with $g_2(0|0)=g_2(1|0)=1/2$ and $g_2(2|1)=1$.  We will show their mixture $\mbf{p}_{B|A_1,A_2}=\lambda\mbf{p}'_{B|A_1,A_2}+(1-\lambda)\mbf{p}''_{B|A_1,A_2}$ does not belong to $\mc{C}_{2}([2]^2;[3])$ for $\lambda\in(0,1)$.  For if it did, then we could write
\begin{align}
\lambda g_1(b|a_1)+(1-\lambda)g_2(b|a_2)=\sum_{i=1}^2p_i\sum_{\mathclap{m=0,\mbf{e}_i}}d(b|m)q_i(m|a_i). 
\end{align}
Define $g_i'(b|a_i)=\sum_{m=0,\mbf{e}_i}d(b|m)q_i(m|a_i)$ \[\lambda g_1(b|a_1)+(1-\lambda)g_2(b|a_2)=p_1 g_1'(b|a_1)+p_2g_2'(b|a_2)
. \]  Then $p(0|1,1)=p(1|0,1)=p(2|a_1,0)=0$ for $a_1=0,1$ implies $g_1'(0|1)=g_2'(0|1)=0$, $g_1'(1|0)=g_2'(1|1)=0$, and $g_1'(2|a_1)=0$ for $a_1=0,1$.  Hence
\begin{align}
p(0|0,1)&=\lambda g_1(0|0)=p_1 g'_1(0|0)+p_2g'_2(0|1)=p_1g'_1(0|0).
\end{align}
Hence we have $g_1'(b|a_1)\propto g_1(b|a_1)$ for all $b,a_1$.  By normalization, they must, in fact be equal.  A similar argument shows $g_2'(b|a_2)=g_2(b|a_2)$.
\begin{equation*}
\begin{split}
g_1(b|a_1)&=\sum_{m=0,\mbf{e}_1}d(b|m)q_1(m|a_1)\\
g_2(b|a_2)&=\sum_{m=0,\mbf{e}_2}d(b|m)q_2(m|a_2).
\end{split}
\end{equation*}
With $g_1$ being a deterministic MAC, we must have that $d(b|0)\in\{0,1\}$ and $0=d(2|0)$.  By considering $1=g_2(2|1)=d(2|0)q_2(0|1)+d(2|\mbf{e}_2)q_2(\mbf{e}_2|1)=d(2|\mbf{e}_2)q_2(\mbf{e}_2|1)$, we then have $d(2|\mbf{e}_2)=1$.  But since $d(b|0)\in\{0,1\}$, it follows that $g_2$ cannot output values $\{0,1,2\}$ each with nonzero probability. As this is a contradiction, we conclude that $\mc{C}_{2}([2]^2;[3])$ is non-convex.  On the other hand, since $\mc{C}'_2([2]^2;[3])$ is convex by the Proposition 5, it follows that $\mc{C}_N(\bcA,\mc{B})\ne\mc{C}'_N(\bcA,\mc{B})$.

\end{proof}
\setcounter{proposition}{5}
\begin{proposition}
\label{Prop:non-convex2}
With $|\mc{B}|>K+1$ and $|\mc{A}_i|>2$ for some party $\msf{A}_i$, the set $\mc{C}_{N,K}'(\bcA,\mc{B})$ is non-convex; hence $\mc{C}_{N,K}'(\bcA,\mc{B})\not=\conv[\mc{C}_{N,K}(\bcA,\mc{B})]$.
\end{proposition}
\begin{proof}
We consider the extreme case of one party with $N=1$ and $K=1$; i.e. one-local MACs $\mc{C}'_{1}([3];[3])$.  Note that in this case we always have $\mc{C}'_{1}(\bcA;\mc{B})=\mc{C}_{1}(\bcA;\mc{B})$.  Let $\mbf{p}'_{B|A}\in\mc{C}'_{1}([3];[3])$ have coordinates $p'(b|0)=\delta_{b,0}$ and $p'(b|1)=p'(b|2)=\delta_{b,2}$, and let $\mbf{p}''_{B|A}\in \mc{C}'_{1}([3];[3])$ have coordinates $p''(b|1)=\delta_{b,1}$ and $p''(b|0)= p''(b|2)=\delta_{b,2}$.  We will show that their mixture $\mbf{p}_{B|A}=\lambda\mbf{p}'_{B|A}+(1-\lambda)\mbf{p}''_{B|A}$ does not belong to $\mc{C}'_{1}([3];[3])$ for any $\lambda\in (0,1)$. For if it does, we could write
\begin{align*}
p(2|2)=\lambda p'(2|2)+(1-\lambda) p''(2|2)=\sum_{m\in\{0,\mbf{e}_1\}}d(2|m)q(m|2)=1.
\end{align*}
If follows $\exists m\in \{0,\mbf{e}_1\}$  such that $d(2|m)=1$.  Without loss of generality, we let $d(2|0)=1$, and so $d(1|0)=d(0|0)=0$.  This means that
\begin{align*}
p(0|a_1)&=d(0|\mbf{e}_1)q(\mbf{e}_1|a_1)\\
p(1|a_1)&=d(1|\mbf{e}_1)q(\mbf{e}_1|a_1).
\end{align*}
However, it easily leads to a contradiction. The mixture $\mbf{p}_{B|A}$ has coordinates $\{p(0|0),p(0|1)\}=\{\lambda,0\}$ which requires $q(\mbf{e}_1|0)\ne 0$ and $q(\mbf{e}_1|1)= 0$.  Likewise, $\mbf{p}_{B|A}$ has coordinates $\{p(1|0),p(1|1)\}=\{0,1-\lambda\}$ which can only be obtained when $q(\mbf{e}_1|0)= 0$ and $q(\mbf{e}_1|1)\ne 0$. Therefore, $\mbf{p}_{B|A}$ cannot belong to $\mc{C}'_{1}([3];[3])$.
\end{proof}
\begin{remark}
In fact, a distinction between $\conv[\mc{C}_N]$ and $\mc{C}^{(\text{sep})}_N$ can also be seen when $|\mc{B}|>2$ and $|\mc{A}_i|>2$.  As a simple example, consider the one-local MACs $\conv[\mc{C}_1([3];[3])]$ and $\mc{C}^{(\text{sep})}_1([3];[3])$.  The latter consists of all channels $[3]\to[3]$ (including the identity); in contrast, the former must be built by the exchange of just a single particle, and therefore all such channels will have a bounded capacity of one bit. 
\end{remark}
\bigskip

For $K>1$, the problem is even more intricate, from proposition \ref{Prop:binary-out}, we know that the four classes of $(N,K)$-local MACs are equivalent when $|\mc{B}|=2$; the counterexample in proposition \ref{Prop:non-convex2} can also be applied to $(N,K)$-local MACs if $|\mc{B}|>\mc{M}_s=K+1$ ($\prod_{i\in S}|\mc{A}_i|>\mc{M}_s$ always holds), thus, we can conclude that $\mc{C}'_{N,K}(\bcA,\mc{B})$, $\conv[\mc{C}_{N,K}(\bcA,\mc{B})]$, and $\mc{C}^{(\text{sep})}_{N,K}(\bcA,\mc{B})$ are distinct when $|\mc{B}|>K+1$. For the remaining cases, their relation are provided by the following two propositions. 
\begin{proposition}
\label{Prop:non-convex3}
For $|\mc{B}|>2$ and $N>K\geq 2$, the set $\mc{C}_{N,K}(\bcA,\mc{B})$ is non-convex; hence $\mc{C}_{N,K}(\bcA,\mc{B})\not=\mc{C}'_{N,K}(\bcA,\mc{B})$.
\end{proposition}
\begin{proof}
 
We consider $\mbf{p}'_{B|A_1A_2A_3},\mbf{p}''_{B|A_1A_2A_3}~\text{and}~\mbf{p}'''_{B|A_1A_2A_3}\in\mc{C}_{3,2}([2]^3;[3])$ having coordinates $p'(b|a_1,a_2,a_3)=g_{12}(b|a_1,a_2)$,  $p''(b|a_1,a_2,a_3)=g_{13}(b|a_1,a_3)$ and $p''(b|a_1,a_2,a_3)=g_{23}(b|a_2,a_3)$ with  $g_S(0|0,0)=1$, $g_S(1|0,1)=1$, $g_S(2|1,0)=1$ and $g_S(2|1,1)=1$ for any $S\in\{12,13,23\}$.  We will show that their mixture $\mbf{p}_{B|A_1A_2A_3}=\lambda_1\mbf{p}'_{B|A_1A_2A_3}+\lambda_2\mbf{p}'_{B|A_1A_2A_3}+(1-\lambda_1-\lambda_2)\mbf{p}''_{B|A_1A_2A_3}$ does not belong to $\mc{C}_{3,2}([2]^3;[3])$ for $\lambda_1\in(0,1)$ and $\lambda_2\in(0,1-\lambda_1)$.  For if it did, we could write $p(b|a_1,
a_2,a_3)$ as:
\begin{align}
\lambda_1 g_{12}(b|a_1,a_2)+\lambda_1 g_{13}(b|a_1,a_3)+(1-\lambda_1-\lambda_2)g_{23}(b|a_2,a_3)=\sum_{S}p_S\sum_{m\in \mc{M}_S}d(b|m)q_S(m|(a_s)_{s\in S}).
\end{align}
Since all $g_S$ are extreme, one can verify that $p(b|a_1,a_2,a_3)$ has unique decomposition similarly as what we did in proposition 4, thus the previous equation implies that:
\begin{equation*}
\begin{split}
g_{12}(b|a_1,a_2)&=\sum_{m=0,\mbf{e}_1,\mbf{e}_2}d(b|m)q_{12}(m|a_1,a_2)\\
g_{13}(b|a_1,a_3)&=\sum_{m=0,\mbf{e}_1,\mbf{e}_3}d(b|m)q_{13}(m|a_1,a_3)\\
g_{23}(b|a_2,a_3)&=\sum_{m=0,\mbf{e}_2,\mbf{e}_3}d(b|m)q_{23}(m|a_2,a_3).
\end{split}
\end{equation*}
Because $g_S$ are deterministic MAC with three distinct outputs conditional on different inputs, we should have $d(b|0)=d(b'|\mbf{e}_1)=d(b''|\mbf{e}_2)= d(b'''|\mbf{e}_3)=1$ for some $b\ne b'\ne b''\ne b'''$. This contradicts with the fact $|\mc{B}|=3$ hence $\mbf{p}_{B|A_1A_2A_3}$ is not in $\mc{C}_{3,2}([2]^3;[3])$ and it follows that $\mc{C}_{3,2}([2]^3;[3])$ is non-convex. On the other hand, as has been shown in Proposition \ref{prop:N,K convexity}, $\mc{C}'_{3,2}([2]^3;[3])$ is convex, we have in general $\mc{C}_{N,K}(\bs{\mc{A}};\mc{B})\ne\mc{C}'_{N,K}(\bs{\mc{A}};\mc{B})$.
\end{proof}
While $\mc{C}_{N,K}(\bs{\mc{A}};\mc{B})$ is non-convex for $|\mc{B}|>2$, if shared randomness between the particle source and receiver is allowed, then convexity can be restored if $|\mc{B}|\le K+1$. 
\begin{proposition}
\label{prop:N,K convexity}
$\mc{C}'_{N,K}(\bcA;\mc{B})=\mc{C}_{N,K}^{(\text{sep})}(\bcA;\mc{B})$ for arbitrary input sets $\bs{\mc{A}}$ and output set $|\mc{B}|\le K+1$.
\end{proposition}
\begin{proof}
Clearly $\mc{C}'_{N,K}(\bcA;\mc{B})\subset\mc{C}_{N,K}^{(\text{sep})}(\bcA;\mc{B})$.  Conversely, we can construct a deterministic decoder by $d_S(0|0)=1$ and $d_S(b|\mbf{\tilde{e}}_b)=1$; encoders by $q_S(0|(a_s)_{s\in S})=g_S(0|(a_s)_{s\in S})$ and $q_S(\mbf{\tilde{e}}_b|(a_s)_{s\in S})=g_S(b|(a_s)_{s\in S})$, where $\mbf{\tilde{e}}_b$ is the $b^{\text{th}}$ non-zero element in $\mc{M}_S$. Thus, any MAC having form Eq.~(\ref{Eq:N,K-local-binary-out}). can be written as Eq. (21). 
\end{proof}
Different from proposition \ref{Prop:binary-out}, the shared randomness in the $(N,K)$-local setting is crucial. Without knowing the grouping $S$, the decoder will not be able to correctly associate the encoding $\mbf{e}_i$ to the group of parties.  

\section{Dimensions of $(N,K)$-local Classical MACs and quantum MACs }
\label{sec:A2 dimension}
In this section, we extend Theorem 2 and Proposition 9 and discuss the multiple access channel with arbitrary inputs set $\mc{A}_i$ and output set $\mc{B}$. To proceed this section, the notations have to be generalized as follows.\par 
Any MAC in $\mc{C}_{N,K}(\bs{\mc{A}},\mc{B})$ or $\mc{Q}_{N,K}(\bs{\mc{A}},\mc{B})$ can be envisioned as a point $\mbf{p}_{B|\bA}$ in ($|\mc{B}|\cdot\prod_{k=1}^N|\mc{A}_k|$)-dimensional Euclidean space with transition coordinates $(p(b|a_1,\cdots,a_N))_{b\in\mc{B},a_i\in\mc{A}_i}$. To facilitate our analysis, we first introduce the new interference coordinates analogous to Eq.~(39):
\begin{align}
\label{eq:new coordinates}
\bigcup_{b\in\mc{B}}\bigcup_{\substack{k\in\{0,\cdots,N\}\\S\subseteq\{1,\cdots,N\}\\|S|=k}}\bigcup_{\substack{\alpha_s\in\mc{A}_s\backslash \{0\}\\ s\in S}}\left\{q(b,S,\{\alpha_s\}_s):=(-1)^k\sum_{\substack{a_s\in\{0,\alpha_s\}\\s\in S}}\prod_{i\in S}f(a_i)p(b|a_1,\cdots,a_N) \;\bigg|\;\text{$a_j=0$ for $j\not\in S$}\right\},
\end{align}
where $f(a_i)=1$ if $a_i=0$ and $f(a_i)=-1$ when $a_i\ne 0$.  Comparing to Eq.~(39), these new coordinates will have extra dependence on the output $b$ and input set $\{\alpha_s\}_s$. And there are exactly $|\mc{B}|\cdot\prod_{k=1}^N|\mc{A}_k|$ elements in this basis since
\begin{equation*}
|\mc{B}|\cdot\sum_{k=0}^N\sum_{i_1=1}^N\sum^N_{i_2>i_1}\cdots\sum^N_{i_k>i_{k-1}}(|\mc{A}_{i_1}|-1)\times\cdots\times(|\mc{A}_{i_k}|-1)=|\mc{B}|\cdot\prod_{k=1}^N|\mc{A}_k|.
\end{equation*}
This equality can be established as follows. We define the polynomial function:
\begin{equation}
p(x)\coloneqq\prod_{k=1}^N [(|\mc{A}_k|-1)+x]=\sum_{k=0}^Na_kx^k,\notag
\end{equation}
with the coefficients of the polynomial given as:
\begin{equation}
a_k=\sum_{i_1=1}^N\sum^N_{i_2>i_1}\cdots\sum^N_{i_k>i_{k-1}}(|\mc{A}_{i_1}|-1)\times\cdots\times(|\mc{A}_{i_k}|-1).\notag
\end{equation}
Hence,
\begin{equation}
\prod_{k=1}^N|\mc{A}_k|=p(1)=\sum_{k=0}^Na_k=\sum_{k=0}^N\sum_{i_1=1}^N\sum^N_{i_2>i_1}\cdots\sum^N_{i_k>i_{k-1}}(|\mc{A}_{i_1}|-1)\times\cdots\times(|\mc{A}_{i_k}|-1).\notag
\end{equation}
Since the transformation between the interference coordinates Eq.~(\ref{eq:new coordinates}) and the original coordinates $\{p(b|a_1,\cdots,a_n)\}$ is invertible, we can safely using them to discuss any given point $\mbf{p}_{B|\bA}$ in both $\mc{C}_{N,K}$ and $\mc{Q}_{N,K}$\par

\begin{theorem}
\label{thm:Rank Ck general}
 $\dim\mathcal{C}_{N,K}(\bs{\mc{A}},\mc{B})=(|\mc{B}|-1)\sum_{k=0}^K(|\mc{A}|-1)^k\binom{N}{k}$ with all $|\mc{A}_i|=|\mc{A}|$ for simplicity. 
\end{theorem}
\begin{proof}

Each $\mbf{p}_{B|\bA}\in\mc{C}_{N,K}(\bs{\mc{A}},\mc{B})$ satisfies $\prod_{k=1}^N|\mc{A}_k|$ normalization conditions $1=\sum p(b|a_1,\cdots,a_N)$, with $a_i\in [\mc{A}_i]$.  As observed before, $I_{K+1}^{(S)}(\mbf{p}_{B|\bA})=0$ for all $|S|\ge K+1$ for binary inputs/output case. It is easy to extend these constraints to arbitrary inputs/output by simply considering inputs/output relabeling, and hence we can write each point in a subset of the interference coordinates in Eq.~(\ref{eq:new coordinates}):
\begin{align}
\label{eq:new coordinates short}
\bigcup_{\substack{b\in\mc{B}\\b\ne |\mc{B}|-1}}\bigcup_{\substack{k\in\{0,\cdots,K\}\\S\subseteq\{1,\cdots,N\}\\|S|=k}}\bigcup_{\substack{\alpha_s\in\mc{A}_s\backslash \{0\}\\ s\in S}}\left\{q(b,S,\{\alpha_s\}_s):=(-1)^k\sum_{\substack{a_s\in\{0,\alpha_s\}\\s\in S}}\prod_{i\in S}f(a_i)p(b|a_1,\cdots,a_N) \;\bigg|\;\text{$a_j=0$ for $j\not\in S$}\right\},
\end{align}
with all the other elements being redundant. e.g. for $k\ge K+1$, those terms are actually zero (no $k^{\text{th}}$-order interference for $k\ge K+1$) and for $b=|B|-1$, the information can be retrieved from the normalization constraints.  \par
There are $(|\mc{B}|-1)\sum_{k=0}^K(|\mc{A}|-1)^k\binom{N}{k}$ elements in the above set, which upper bound  $\dim\mc{C}_{N,K}(\bs{\mc{A}},\mc{B})$ as:
\begin{equation}
\dim\mathcal{C}_{N,K}(\bs{\mc{A}},\mc{B})\le(|\mc{B}|-1)\sum_{k=0}^K(|\mc{A}|-1)^k\binom{N}{k}
\end{equation}
To compute a lower bound for $\dim\mc{C}_{N,K}$ we need to find a set of affinely independent points. To do that,  we consider the following $(N,K)$-local classical MACs (analogy to Eq.~(15)): 
\begin{equation}
\begin{split}
\label{Eq: new affine points}
    &p(b|a_1,\cdots,a_N)=\delta_{b\beta}, \qquad  \beta\in \mc{B}\\
    &p(b|a_1,\cdots,a_N)=d(b|0)q_S(0|(a_s)_{s\in S})=\prod_{s\in S}\delta_{\alpha_s,a_s}\delta_{b\beta},\quad \text{$ \beta\in [|\mc{B}|-1]$, $ |S|\leq K$ and $\alpha_s\in \mc{A}_s\backslash\{0\}$}, 
\end{split}   
\end{equation} 
 These MACs have exactly one interference coordinate being $1$ (only when the inputs set $\{\alpha_s\}_s$ and output $b$ match with the coordinate in Eq.~(\ref{eq:new coordinates})) and are obviously from $\mc{C}_{N,K}(\bs{\mc{A}},\mc{B})$. In total then, we obtain a collection of $(|\mc{B}|-1)\sum_{k=0}^K(|\mc{A}|-1)^k\binom{N}{k}+1$ affinely independent points. Hence we have:
\begin{equation}
\dim\mathcal{C}_{N,K}(\bs{\mc{A}},\mc{B})\ge(|\mc{B}|-1)\sum_{k=0}^K(|\mc{A}|-1)^k\binom{N}{k}
\end{equation}
\par
\noindent In summary, the dimension of the $\mathcal{C}_{N,K}(\bs{\mc{A}},\mc{B})$ polytope is $(|\mc{B}|-1)\sum_{k=0}^K(|\mc{A}|-1)^k\binom{N}{k}$ for any $1\le K\le N$.\par 
Note the above calculation also works for $\mathcal{C}'_{N,K}(\bs{\mc{A}},\mc{B})$,  $\conv[\mathcal{C}_{N,K}(\bs{\mc{A}},\mc{B})]$, $\mathcal{C}^{\text{sep}}_{N,K}(\bs{\mc{A}},\mc{B})$ and $\mf{I}_{N,K}(\bs{\mc{A}},\mc{B})$ since all the above constraints and affinely points constructed in Eq.~(\ref{Eq: new affine points}) work for all of them.
\end{proof}

\begin{remark}
Theorem \ref{thm:Rank Ck general} says the dimension $\mf{J}_{N,K}(\bs{\mc{A}},\mc{B})$ is $(|\mc{B}|-1)\sum_{k=0}^K(|\mc{A}|-1)^k\binom{N}{k}$ with new coordinates we introduced in Eq.~(\ref{eq:new coordinates short}). With these coordinates. In fact, for any $L\in\{1,\cdots,N\}$ and any permutation $\pi$ on parties, 
\begin{align}
    p(b|\overbrace{\alpha_1,\cdots,\alpha_L}^{L},\overbrace{0,\cdots,0}^{N-L})=\sum_{k=0}^K\sum_{\substack{S\subset\{1,\cdots,L\}\\|S|=k}}q(b,S,\{\alpha_s\}_s)
\end{align}
\end{remark} 
\setcounter{proposition}{12}
\begin{proposition}\label{prop:Rank Qk general}
$\dim \mathcal{Q}_{N,K}(\bs{\mc{A}},\mc{B}) = (|\mc{B}|-1)\sum_{k=0}^{2K}(|\mc{A}|-1)^k\binom{N}{k}$ for $K\le \lfloor\frac{N}{2}\rfloor$ 
\end{proposition}
\begin{proof}
The set $\mathcal{Q}_{N,K}(\bs{\mc{A}},\mc{B})$ lives in a $(|\mc{B}|\cdot\prod_{i=1}^N|\mc{A}_i|)$-dimensional space, and there are also $\prod_{i=1}^N|\mc{A}_i|$ normalization condition for each $\mbf{p}_{B|\bA}\in\mc{Q}_{N}(\bs{\mc{A}},\mc{B})$. Additionally, proposition 8 shows that $\mbf{p}_{B|\bA} \in \mc{Q}_{N,K}(\bs{\mc{A}},\mc{B})$ can not violate $I_{2K+1}$ equality and obviously any higher order equality (The proposition works for arbitrary inputs/output), which introduces some extra constraints. And hence, we can write each point in a subset of the interference coordinates in Eq.~(\ref{eq:new coordinates}) with $k\in\{0,\cdots,2K\}$ and $b\ne |\mc{B}|-1$ as:
\begin{align}
\label{eq:new coordinates 1}
\bigcup_{\substack{b\in\mc{B}\\b\ne |\mc{B}|-1}}\bigcup_{\substack{k\in\{0,\cdots,2K\}\\S\subseteq\{1,\cdots,N\}\\|S|=k}}\bigcup_{\substack{\alpha_s\in\mc{A}_s\backslash \{0\}\\ s\in S}}\left\{q(b,S,\{\alpha_s\}_s):=(-1)^k\sum_{\substack{a_s\in\{0,\alpha_s\}\\s\in S}}\prod_{i\in S}f(a_i)p(b|a_1,\cdots,a_N) \;\bigg|\;\text{$a_j=0$ for $j\not\in S$}\right\},
\end{align}
There are $(|\mc{B}|-1)\sum_{k=0}^K(|\mc{A}|-1)^k\binom{N}{k}$ elements in the above set, which upper bound the dimension of $\dim\mc{Q}_{N,K}(\bs{\mc{A}},\mc{B})$ as:
\begin{equation}
\dim\mathcal{Q}_{N,K}(\bs{\mc{A}},\mc{B})\le(|\mc{B}|-1)\sum_{k=0}^{2K}(|\mc{A}|-1)^k\binom{N}{k}
\end{equation}
To get the lower bound, we consider the following two classes of points: \\
(i) Classical MACS: The set of $(|\mc{B}|-1)\sum_{k=0}^K(|\mc{A}|-1)^k\binom{N}{k}$ affinely independent points in Eq.~(\ref{Eq: new affine points}) in $\mc{C}_{N,K}$ as are obviously also in $\mc{Q}_{N,K}(\bs{\mc{A}},\mc{B})$. \\
(ii) Quantum MACS: We consider quantum MACs generated by quantum state $\ket{\psi}^{\msf{A}_i\msf{A}_j}=\frac{1}{\sqrt{2}}(\ket{\mbf{e}_i}+\ket{\mbf{e}_j})$ shared between party $i\in S_i$, $j\in S_j$ in disjoint groups $S_i\cap S_j=\emptyset$. Each group performs a $\{0,\pi\}$ phase encoding on sender $i$ or $j$ conditional on the `parity' of their inputs, i.e. $\mc{E}^{(\msf{A}_s)_{s\in S_i}}_{(a_s)_{s\in S_i}}(X)=\bigotimes_{s\in S_i\setminus \{i\}}\text{id}^{\msf{A}_s}\otimes (\sigma_z^{P_S[(a_s)_{s\in S}]})^{\msf{A}_i}(X)(\sigma_z^{P_S[(a_s)_{s\in S}]})^{\msf{A}_i}$ where $\sigma_z=\op{0}{0}+e^{i\pi}\op{1}{1}$ and the parity function $P_S[(a_s)_{s\in S}]$ depends on the input set $\{\alpha_s\}_{s\in S_i\cup S_j}$ as:
\begin{eqnarray}
P_S[(a_s)_{s\in S}]=(\sum_{s\in S}\delta_{\alpha_s,a_s})\text{(mod) $2$}
\end{eqnarray}
The decoder then measures with POVM $\Pi_b=\op{\psi}{\psi}^{\msf{A}_i\msf{A}_j}$and $\Pi_{|\mc{B}|-1}=\mbb{I}-\op{\psi}{\psi}^{\msf{A}_i\msf{A}_j}$. This quantum MAC can demonstrate maximal $k^{\text{th}}$-order interference ($K<k=|S|$) among and only among senders $S=S_i\cup S_j$, input set $\{\alpha_s\}_s$ and output $b$. \par
The same strategy applies for all groups $S$ with $K<|S|\le2K$, input sets $\{\alpha_s\}_{s\in S_i\cup S_j}$ and outputs $b$, thus we can obtain $(|\mc{B}|-1)\sum_{k=K+1}^{2K}(|\mc{A}|-1)^k\binom{N}{k}$ different MACs, now let's check whether they are affinely independent. \par 
Points in class (i) do not violate $I_k$ equalities for any $k>K$ hence have $q(b,S,\{\alpha_s\}_s)=0$ for any $|S|>K$, while each point considered in class (ii) violates one and only one $I_k$ equality for integer $k=|S|>K$, specifically has $q(b,S,\{\alpha_s\}_s)=2^{k-1}$. By a similar reasoning as for proposition 9, the collection of these points forms a upper-triangular matrix similarly as Eq.~\ref{eq:Qexample} thus being affinely linear independent. In total then, we obtain a collection of $(|\mc{B}|-1)\sum_{k=0}^{2K}(|\mc{A}|-1)^k\binom{N}{k}+1$ affinely independent points belonging to $\mc{Q}_{N,K}(\bs{\mc{A}},\mc{B})$. Hence we have:
\begin{equation}
\dim\mathcal{Q}_{N,K}(\bs{\mc{A}},\mc{B})\ge(|\mc{B}|-1)\sum_{k=0}^{2K}(|\mc{A}|-1)^k\binom{N}{k}
\end{equation}
\par
\noindent In summary, the dimension of the $\mc{Q}_{N,K}(\bs{\mc{A}},\mc{B})$ polytope is $(|\mc{B}|-1)\sum_{k=0}^K(|\mc{A}|-1)^k\binom{N}{k}$ for any $1\le K\le \lfloor\frac{N}{2}\rfloor$, and obviously we can similarly conclude that  $\mc{Q}_{N,K}(\bs{\mc{A}},\mc{B})=(|\mc{B}|-1)\sum_{k=0}^K(|\mc{A}|-1)^k\binom{N}{k}=(|\mc{B}|-1)\prod_{s\in S}|\mc{A}_s|$ for $\lfloor\frac{N}{2}\rfloor<K\le N$
\end{proof}

\section{Separation of $\mc{Q}_{2,1}([2]^2;[2])$ and $\mc{C}_{2,2}([2]^2;[2])$}
\label{sec:A3 seperation}
As we alluded to earlier, in general, while they have the same dimensions, $\mc{Q}_{N,K}$ are fundamentally different from $\mc{C}_{N,2K}$. Here we give a proof for the simplest case with $N=2$, $K=1$ and binary inputs/output, i.e., $\mc{Q}_{2,1}([2]^2;[2])\ne\mc{C}_{2,2}([2]^2;[2])$.
\begin{proposition}
\label{prop: seperation}
$\mc{Q}_{2,1}([2]^2;[2])\ne\mc{C}_{2,2}([2]^2;[2])$.
\end{proposition}
\begin{proof}
Consider the MAC $\mbf{p}_{B|A_1A_2}$ given by $p(0|a_1a_2)=\delta_{0,a_1}\delta_{0,a_2}$, which is clearly in $\mc{C}_{2,2}([2]^2;[2])$ since $\mc{C}_{2,2}([2]^2;[2])$ contains all channels from $[2]^2$ to [2]. Assume, towards a contradiction, that this channel is in $\mc{Q}_{2,1}([2]^2;[2])$ as well. Suppose that the senders $\msf{A}_1$ and $\msf{A}_2$ share the state $\rho$ having support on $\mc{H}_1^{\msf{A}_1\msf{A}_2}$, then we have \begin{align}
    p(0|00)=\tr\left[\Pi_0\left(\mc{E}_0^{\msf{A}_1}\otimes\mc{E}_0^{\msf{A}_2}(\rho)\right)\right],\notag\\
    p(0|01)=\tr\left[\Pi_0\left(\mc{E}_0^{\msf{A}_1}\otimes\mc{E}_1^{\msf{A}_2}(\rho)\right)\right].\notag
\end{align}
Notice that the channel $\mbf{p}_{B|A_1A_2}$ cannot be expressed as a convex combination of two channels; therefore, $\rho$ must in fact be a pure state $\ket{\psi}\bra{\psi}$. Let $\rho'=\mc{E}_0^{\msf{A}_1}\otimes\text{id}^{\msf{A}_2}(\ket{\psi}\bra{\psi})$, then
\begin{align}
p(0|00)&=\tr\left[\Pi_0\left(\text{id}^{\msf{A}_1}\otimes\mc{E}_0^{\msf{A}_2}(\rho')\right)\right],\notag\\
p(0|01)&=\tr\left[\Pi_0\left(\text{id}^{\msf{A}_1}\otimes\mc{E}_1^{\msf{A}_2}(\rho')\right)\right].\notag
\end{align}
We can write $\rho'=p\rho_0+(1-p)\rho_1$ where $\rho_0$ and $\rho_1$ are projections of $\rho$ to the vacuum subspace $\mc{H}_0$ and the one-particle subspace $\mc{H}_1$ respectively. This is because NPE operations always map $\ket{\psi}\bra{\psi}$ to a block-diagonal state (see Eq. \eqref{Eq: M}). Then
\begin{equation}
    \begin{cases}
    \label{Eq:Q_21-C_22-contradiction}
    p\tr\left[\Pi_0\left(\text{id}^{\msf{A}_1}\otimes\mc{E}_0^{\msf{A}_2}(\ket{00}\bra{00})\right)\right] + (1-p)\tr\left[\Pi_0\left(\text{id}^{\msf{A}_1}\otimes\mc{E}_0^{\msf{A}_2}(\rho_1)\right)\right]=1, \\[5pt]
    p\tr\left[\Pi_0\left(\text{id}^{\msf{A}_1}\otimes\mc{E}_1^{\msf{A}_2}(\ket{00}\bra{00})\right)\right] + (1-p)\tr\left[\Pi_0\left(\text{id}^{\msf{A}_1}\otimes\mc{E}_1^{\msf{A}_2}(\rho_1)\right)\right]=0.
    \end{cases}
\end{equation}
We then proceed to show that Eq.~\eqref{Eq:Q_21-C_22-contradiction} always leads to a contradiction for any $p\in[0,1]$.\par
If $p\in(0,1]$, then Eq.~\eqref{Eq:Q_21-C_22-contradiction} implies that
\begin{align}
\tr\left[\Pi_0\left(\text{id}^{\msf{A}_1}\otimes\mc{E}_0^{\msf{A}_2}(\ket{00}\bra{00})\right)\right]=1 \notag   
\end{align}
and
\begin{align}
\tr\left[\Pi_0\left(\text{id}^{\msf{A}_1}\otimes\mc{E}_1^{\msf{A}_2}(\ket{00}\bra{00})\right)\right]=0. \notag
\end{align}
Since $\mc{E}_0^{\msf{A}_2}$ and $\mc{E}_1^{\msf{A}_2}$ are both NPE operations, $\text{id}^{\msf{A}_1}\otimes\mc{E}_0^{\msf{A}_2}(\ket{00}\bra{00})=\text{id}^{\msf{A}_1}\otimes\mc{E}_1^{\msf{A}_2}(\ket{00}\bra{00})=\ket{00}\bra{00}$. Therefore, these two equations cannot be true simultaneously, and we have a contradiction.

Next, suppose that $p=0$. In this case, $\rho'=\mc{E}_0^{\msf{A}_1}\otimes\text{id}^{\msf{A}_2}(\ket{\psi}\bra{\psi})$ remains a one-particle state. Following a similar argument, we can see that $\mc{E}_1^{\msf{A}_1}\otimes\text{id}^{\msf{A}_2}(\ket{\psi}\bra{\psi})$ and $\text{id}^{\msf{A}_1}\otimes\mc{E}_{0,1}^{\msf{A}_2}(\ket{\psi}\bra{\psi})$ also have to be one-particle states. Therefore, all of the encoded states $\sigma_{00}$, $\sigma_{01}$, $\sigma_{10}$, and $\sigma_{11}$ (recall $\sigma_{a_1a_2}=\mc{E}_{a_1}^{\msf{A}_1}\otimes\mc{E}_{a_2}^{\msf{A}_2}(\ket{\psi}\bra{\psi})$) are operators on the two-dimensional space $\mc{H}_1$. Additionally, note that $p(0|a_1a_2)=\delta_{0,a_1}\delta_{0,a_2}$ implies that the encoder can perfectly distinguish $\sigma_{00}$ from $\{\sigma_{01},\sigma_{10},\sigma_{11}\}$, i.e., $\sigma_{00}\perp\sigma_{01},\sigma_{10},\sigma_{11}$. Thus, it must be that $\sigma_{01}=\sigma_{10}=\sigma_{11}$, which means that $\mc{E}_0^{\msf{A}_1}=\mc{E}_1^{\msf{A}_1}$ and $\mc{E}_0^{\msf{A}_2}=\mc{E}_1^{\msf{A}_2}$. However, this will imply that all of the four encoded states are the same, that is, $\sigma_{00}=\sigma_{01}=\sigma_{10}=\sigma_{11}$. We thus arrive at a contradiction as well. \par
To conclude, we have shown that the MAC $\mbf{p}$ is not in $\mc{Q}_{2,1}([2]^2;[2])$, and therefore $\mc{Q}_{2,1}([2]^2;[2])\ne\mc{C}_{2,2}([2]^2;[2])$.
\end{proof}
\begin{remark}
While in our formalism, the two parties $\msf{A}_1$ and $\msf{A}_2$ are limited to NPE operations, the proof above in fact works for a more general class of operations that preserves the vacuum state.
\end{remark}
\begin{conjecture}
$\mc{Q}_{2,1}([2]^2;[2])$ is not a polytope.
\label{prop: nonpolytope}
\end{conjecture}
\begin{proof}
Here we give a proof of the special case with projective measurement (in 1-particle subspace), for the general POVM case, we have strong numerical evidence that this is also true. \par  
From proposition \ref{prop:Rank Qk general}, $\dim \mc{Q}_{2,1}([2]^2;[2])=4$, and MACs $\mbf{p}_{B|A_1A_2}\in \mc{Q}_{2,1}([2]^2;[2])$ can be written in coordinate $\{p(0|00),p(0|01),p(0|10),p(0|11)\}$. By convexity, we can always assume the state to be pure (in subspace $\{\ket{01},\ket{10}\}$):
\begin{equation}
\rho=\begin{pmatrix}
\cos^2\theta_s, &\sin\theta_s\cos\theta_s e^{i\phi_s}\\ 
\sin\theta_s\cos\theta_s e^{-i\phi_s} & \sin^2\theta_s
\end{pmatrix}
\end{equation}
As discussed in Eq.~\eqref{Eq:Kraus-qubit-MAC}, every NPE operation $\mc{E}_{a}^{\msf{A}_i}$ on qubit system $\msf{A}_i$ is characterized by Kraus operators $\left\{\left(\begin{smallmatrix}1&0\\0&y_{a_i}^{i}\end{smallmatrix}\right),\left(\begin{smallmatrix}0&z_{a_i}^{i}\\0&0\end{smallmatrix}\right)\right\}$ with $|y_{a_i}^i|^2+|z_{a_i}^i|^2=1$.  Hence the final state will be given as:
\begin{equation}
\sigma_{a_1a_2}=\mc{E}_{a_1}^{\msf{A}_1}\otimes\mc{E}_{a_2}^{\msf{A}_2}(\rho)=\begin{pmatrix}
|z^1_{a_1}|^2\cos^2\theta_s+|z^2_{a_2}|^2\sin^2\theta_s&0&0\\
0&|y^1_{a_1}|^2\cos^2\theta&y^1_{a_1}y^{2*}_{a_2}\sin\theta_s\cos\theta_s e^{i\phi_s}\\ 
0& y^{1*}_{a_1}y^{2}_{a_2}\sin\theta_s\cos\theta_s e^{-i\phi_s} & |y^2_{a_2}|\sin^2\theta_s
\end{pmatrix}
\end{equation}
We further assume the measurement POVM to be projective in the 1-particle subspace:
\begin{equation}
\Pi_0=q\op{00}{00}\oplus \begin{pmatrix}
\cos^2\theta_m, &\sin\theta_m\cos\theta_m e^{i\phi_m}\\ 
\sin\theta_m\cos\theta_m e^{-i\phi_m} & \sin^2\theta_m
\end{pmatrix},
\label{eq: POVM}
\end{equation}
with $0\le q \le 1$. Then, 
\begin{equation}
p(0|a_1a_2)=\tr\left[\Pi_0\sigma_{a_1a_2}\right]=q(|z^1_{a_1}|^2\cos^2\theta_s+|z^2_{a_2}|^2\sin^2\theta_s)+|y^1_{a_1}\cos\theta_s\cos\theta_m+y^{2}_{a_2}\sin\theta_s\sin\theta_m e^{\phi_s-\phi_m}|^2
\end{equation}
To show that $\mc{Q}_{2,1}([2]^2;[2])$ is not a polytope, we can just look at one of its cross-section. As an illustive example, we set $p(0|00)=1$ and $p(0|10)-p(0|11)=0$ and look at the cross-section represented in the coordinates $\{p(0|01),p(0|10)\}$. \par
With $p(0|00)=1$, we have the following two extreme cases:
\par
case 1: $q=|z^1_{0}|=|z^2_{0}|=1$
\par
case 2: $|y^1_{0}|=|y^2_{0}|=1$, $\theta_s=\theta_m$ and $\phi_s-\phi_m-\phi_{0}^1+\phi_{0}^2=0$\\
All other cases can be written as affine combination of these two, e.g: $q=1$, $|z^1_{0}|=|z^2_{0}|$, $|y^1_{0}|=|y^2_{0}|$, $\theta_s=\theta_m$ and $\phi_s-\phi_m-\phi_{0}^1+\phi_{0}^2=0$
\par
Then with $p(0|10)=p(0|11)$, we have:
\par 
case 1: $|z^2_1|=1$, therefore every MAC can be written as:
$$\{1, |z^1_1|\cos^2\theta+\sin^2\theta_s+|y^1_1|^2\cos^2\theta_s\cos^2\theta_m\}$$
\par 
case 2: $2(\phi^1_1-\phi^1_0)=(\phi^2_1-\phi^2_0):=\phi$ and $|y^2_1|=1$, therefore every MACs can be written as:
$$\{|\cos^2\theta_s+\sin^2\theta_s e^{-\phi}|^2,||y^1_{1}|\cos^2\theta_s+\sin^2\theta_s e^{\phi/2}|^2\}$$
\par 
case 2': $(\phi^1_1-\phi^1_0)=\pm\frac{\pi}{2},~(\phi^2_1-\phi^2_0)-(\phi^1_1-\phi^1_0)=\pm\frac{\pi}{2}$ and $|y^2_1|=1$, and MACs are written as:
$$\{\cos^4\theta+\sin^4\theta,|y^1_1|^2\cos^4\theta+\sin^4\theta\}$$
\begin{figure}[t]
    \centering
    \includegraphics[width=0.5\textwidth]{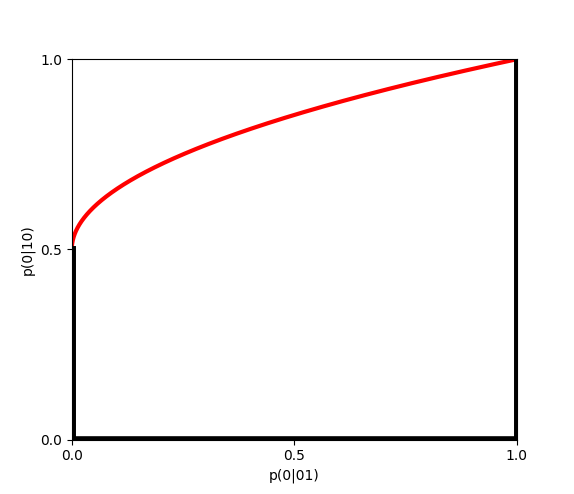}
    \caption{Cross-section of $\mc{Q}_{21}([2]^2,[2])$ with $p(0|00)=1$ and $p(0|10)=p(0|11)$, the boundary is obtained with the subset of measurement considered in eq.~\ref{eq: POVM} (projective in the 1-particle subspace). However, based our numerical simulation, no other points can be found even when general POVM is used.}
    \label{fig:my_label}
\end{figure}
We now identify the boundary of the above set. From the supporting hyperplane theorem, it is necessary to find all the supporting hyperplane, which in the case of two-deimensional space, will be of the form $tp(0|10)-p(0|01)=C$. One can already check that the three extreme points on the lower-right $\{\{0,0\},\{1,0\},\{1,1\}\}$ can be obtained. Therefore, we only have to look at hyperplane on the upper-left corner:  $$tp(0|10)-p(0|01)=C~~~~\text{for $t\in (0,\infty)$},$$ which are maximized only by point in case 2 when 
$$|y^1_1|=1,~~~~\cos\theta=\frac{1}{\sqrt{2}},~~~~\cos(\frac{\phi}{2})=\begin{cases}\frac{t}{4} & \text{for $t\in (0,4)$}\\ 1&\text{for $t\in [4,\infty)$}\end{cases}.$$
And the upper right boundary now can be calculated out as:
$\{\frac{t^2}{16},\frac{t}{8}+\frac{1}{2}\}$, 

\end{proof}
Obviously, the above boundarys form a convex set, but the set is not a polytope, and it follows that the original set $\mc{Q}_{1,1}([2]^,[2])$ is not a polytope with projective measurement. Our numerical search find no point     generates from POVM locates outside this set, thus we conjecture that  $\mc{Q}_{1,1}([2]^,[2])$ is not a polytope 

\end{document}